\documentclass[11pt]{article}

\let\Horig\H

\usepackage[leqno]{amsmath}
\usepackage{amssymb}
\usepackage{amsthm}
\usepackage{enumerate}
\usepackage{verbatim}
\usepackage{mathrsfs}
\usepackage{dsfont}
\usepackage{graphicx}

\newcommand{\ds}{\displaystyle}
\newcommand{\ts}{\textstyle}

\newcommand{\mr}{\mathrm}

\newcommand{\mc}{\mathcal}

\newcommand{\scr}{\mathscr}

\newcommand{\uhr}{\upharpoonright}

\newcommand{\tsfrac}[2]{{\ts\frac{#1}{#2}}}

\newcommand{\eps}{\varepsilon}

\newcommand{\BAR}[1]{\overline{#1}}

\newcommand{\AC}{\complexityfont{AC}}

\newcommand{\probabilityfont}[1]{\mathsf{#1}}
\let\H\undefined
\let\Ex\undefined
\let\Pr\undefined
\newcommand{\H}{\probabilityfont{H}}
\DeclareMathOperator*{\Ex}{\mathds{E}}
\DeclareMathOperator*{\Pr}{\mathds{P}}

\usepackage{mathtools}
\newcommand{\defeq}{\vcentcolon=}

\newcommand{\fieldfont}[1]{\mathbb{#1}}
\newcommand{\N}{\fieldfont{N}}
\newcommand{\R}{\fieldfont{R}}
\newcommand{\Z}{\fieldfont{Z}}

\newtheoremstyle{theorem-style}
  {}
  {}
  {\slshape}
  {}
  {\bf}
  {.}
  {.5em}
  {}

\newtheorem{thm}{Theorem}[section]

\newtheorem{prop}[thm]{Proposition}
\newtheorem{la}[thm]{Lemma}

\newtheorem{main-la}[thm]{Main Lemma}
\newtheorem{redu}[thm]{Reduction}
\newtheorem{cor}[thm]{Corollary}

\newtheorem{claim}[thm]{Claim}

\newtheorem{qu}[thm]{Question}

\theoremstyle{definition}
\newtheorem{df}[thm]{Definition}

\newtheorem{rmk}[thm]{Remark}
\newtheorem{ex}[thm]{Example}
\newtheorem{notn}[thm]{Notation}

\usepackage{fullpage}
\usepackage{stmaryrd}

\usepackage[draft=false,hypertexnames=false,bookmarksopenlevel=1,bookmarksopen,bookmarksnumbered,colorlinks,plainpages=false,linktocpage,linkcolor=black,citecolor=black]{hyperref}

\newcommand{\thedepth}{\frac{\log n}{k^3\log\log n}}

\newcommand{\STCONN}{\textsl{STCONN}}
\newcommand{\USTCONN}{\textsl{USTCONN}}
\newcommand{\PATH}{\textsl{DISTCONN}}

\newcommand{\End}[1]{\mathsf{Ends}(#1)}
\newcommand{\Int}[1]{\mathsf{Interior}(#1)}

\newcommand{\prelim}[1]{\subsubsection*{#1}}

\newcommand{\DTD}{D}
\newcommand{\shift}[1]{{\triangleright}{#1}}

\newcommand{\cost}{\chi}

\newcommand{\FormSD}{\mathsf{Formula}}
\newcommand{\CircSD}{\mathsf{Circuit}}
\newcommand{\DSPACE}{\mathsf{DSPACE}}

\usepackage{qtree}

\newcommand{\depth}{\mr{depth}}

\newcommand{\sib}[1]{#1^{\sim}}
\newcommand{\parent}[1]{#1^{\uparrow}}

\newcommand{\Overlap}[1]{T_{#1}}
\newcommand{\Strict}[1]{\mathit{Strict}(#1)}
\newcommand{\Con}[1]{\mc I(#1)}

\renewcommand{\P}{\scr P}
\newcommand{\smallP}{\scr P^{\hspace{.35pt}\textup{small}}}
\newcommand{\smallR}{\scr R^{\textup{small}}}

\renewcommand{\l}[1]{\ell_{\smash{#1}}}
\renewcommand{\c}[1]{\Delta_{\smash{#1}}}
\newcommand{\PHI}[1]{\Phi_{\smash{#1}}}

\newcommand{\cc}[1]{\Delta(#1)}
\newcommand{\PHII}[1]{\Phi(#1)}
\newcommand{\VV}[1]{V_{\smash{#1}}}

\newcommand{\AB}{\{A,B\}}
\newcommand{\BA}{\{B,A\}}

\newcommand{\PalPat}{M}
\newcommand{\Pal}{\mc{P}al}

\newcommand{\A}{\mc{A}}
\newcommand{\B}{\mc{B}}
\newcommand{\C}{\mc{C}}

\newcommand{\pRHO}[3]{(#2)|_{#1}^{#3}}
\newcommand{\RHO}[3]{#2|_{#1}^{#3}}

\renewcommand{\AC}{\mr{AC}}

\newcommand{\COPY}{\mr{Copy}}
\newcommand{\Live}{\mathit{Live}}

\renewcommand{\AC}{\mathsf{AC}}

\newcommand{\m}{\tilde{n}}
\newcommand{\psino}{\psi}
\newcommand{\psifc}{\xi}

\newcommand{\proj}{\mathrm{proj}}

\newcommand{\X}{R}
\newcommand{\size}{\mr{size}}

\renewcommand{\AA}[3]{\A^{#1}_{#3,#2}}
\newcommand{\edge}[2]{#1#2}
\newcommand{\lift}[2]{#1^{#2}}

\newcommand{\pcost}{\bar{\chi}}

\makeatletter
\newtheorem*{rep@theorem}{\rep@title}
\newcommand{\newreptheorem}[2]{%
\newenvironment{rep#1}[1]{%
 \def\rep@title{#2 \ref{##1}}%
 \begin{rep@theorem}}%
 {\end{rep@theorem}}}
\makeatother
\newreptheorem{theorem}{Theorem}

\usepackage{titlesec}
\titleformat{\subsubsection}[runin]
{\bf\normalsize}{\thesubsubsection}{1em}{}[:]

\begin{document}

\title{Formulas vs.\ Circuits for Small Distance Connectivity}
\author{Benjamin Rossman\footnote{National Institute of Informatics, 2-1-2 Hitotsubashi, Chiyoda-ku, Tokyo 101-8430, Japan}}
\date{November 28, 2013}
\maketitle{}

\begin{abstract}
We give the first super-polynomial separation in the power of bounded-depth boolean formulas vs.\ circuits. Specifically, we consider the problem \textsc{Distance $k(n)$ Connectivity}, which asks whether two specified nodes in a graph of size $n$ are connected by a path of length at most $k(n)$. This problem is solvable (by the recursive doubling technique) on {\bf circuits} of depth $O(\log k)$ and size $O(kn^3)$. In contrast, we show that solving this problem on {\bf formulas} of depth $\log n/(\log\log n)^{O(1)}$ requires size $n^{\Omega(\log k)}$ for all $k(n)$ $\le \log\log n$. As corollaries:
\begin{enumerate}[(i)]
  \item
    It follows that polynomial-size circuits for \textsc{Distance $k(n)$ Connectivity} require depth $\Omega(\log k)$ for all $k(n) \le \log\log n$. This matches the upper bound from recursive doubling and improves a previous $\Omega(\log\log k)$ lower bound of Beame, Pitassi and Impagliazzo [BIP98].
  \item
    We get a tight lower bound of $s^{\Omega(d)}$ on the size required to simulate size-$s$ depth-$d$ circuits by depth-$d$ formulas for all $s(n) = n^{O(1)}$ and $d(n) \le \log\log\log n$. No lower bound better than $s^{\Omega(1)}$ was previously known for any $d(n) \nleq O(1)$.
\end{enumerate}

Our proof technique is centered on a new notion of {\em pathset complexity}, which roughly speaking measures the minimum cost of constructing a set of (partial) paths in a universe of size $n$ via the operations of union and relational join, subject to certain density constraints. Half of our proof shows that bounded-depth formulas solving \textsc{Distance $k(n)$ Connectivity} imply upper bounds on pathset complexity. The other half is a combinatorial lower bound on pathset complexity. 
\end{abstract}

\newpage
\tableofcontents{}
\newpage

\section{Introduction}\label{sec:intro}

Understanding the relative power of formulas vs.\ circuits is a central challenge in complexity theory, especially in the important boolean setting.  Whereas boolean circuits are the most general non-uniform model of computation (in particular, boolean circuits can efficiently simulate Turing machines), there is a strong intuition that boolean formulas ($=$ tree-like circuits with fan-out $1$) are a very weak model of computation. Many natural problems solvable by small circuits, such as st-connectivity, are believed to require large formulas. However, no super-polynomial gap between the formula complexity and circuit complexity of any problem has ever been established. The existence of such a gap is a major open question.

\begin{qu}\label{qu:NC1}
Are polynomial-size boolean circuits strictly more powerful than polynomial-size boolean formulas?
\end{qu}

There are two versions of Question \ref{qu:NC1} for the {\bf uniform} and {\bf non-uniform} settings.\footnote{Whenever we speak of a circuit (or formula), this is understood to mean a sequence $(C_n)_{n=1}^\infty$ of circuits, one for each input size $n$. In the {\em uniform} setting, there is an underlying algorithm which, given $1^n$ as input, outputs a description of the circuit $C_n$. In the {\em non-uniform} setting, $C_n$ are arbitrary. All bounds mentioned in this paper may be interpreted in the stronger sense: uniform upper bounds and non-uniform lower bounds.} In terms of complexity classes, this is equivalent to asking whether \textsf{uniform-}$\mathsf{NC^1}$ (resp.\ $\mathsf{NC^1}$) is a proper subclass of $\mathsf{P}$ (resp.\ $\mathsf{P/poly}$).\footnote{By Spira's Theorem \cite{spira1971time}, $\mathsf{NC^1}$ is equivalent to the class of languages recognized by polynomial-size boolean formulas (of unbounded depth).} 
Both the uniform and non-uniform versions of this question are wide open:
\begin{itemize}
\item
An obvious prerequisite of the separation of \textsf{uniform-}$\mathsf{NC^1}$ from $\mathsf{P}$ is a super-polynomial lower bound on the formula complexity of {\em any} explicit boolean function. However, despite the fact that {\em almost all} boolean functions have formula complexity $\Omega(2^n/\log n)$ by a classic theorem of Riordan and Shannon \cite{riordan1942number}, the best lower bound for any explicit function, due to H{\aa}stad \cite{hastad1998shrinkage}, is only $\Omega(n^{3-o(1)})$. Unfortunately, $n^3$ is known to be the limit of existing techniques, and it appears that any improvement will require a major breakthrough.
\item
The situation is no better in the non-uniform setting. By a striking theorem of Savick\'y and Woods \cite{savicky1998number}, for every constant $k > 1$, almost all boolean functions with formula complexity $\le n^k$ have circuit complexity $\ge n^k/k$. This shows that $\mathsf{NC^1}$ cannot be separated from $\mathsf{P/poly}$ by a straightforward counting argument (in contrast with results like the Circuit Size Hierarchy Theorem, see \cite{jukna2012boolean}). Other than by counting arguments, it is not clear how to take advantage of non-uniformity.
\end{itemize}

In short, it appears that we are a long way from answering Question \ref{qu:NC1}. In  the meantime, we can hope to gain insight by studying the question of formulas vs.\ circuits in {\em restricted settings} where strong lower bounds are available. In particular, Question \ref{qu:NC1} has natural analogues in both the {\bf monotone} setting and the {\bf bounded-depth} (boolean) setting, where exponential lower bounds have been around for decades. However, as we will explain, while question of monotone formulas vs.\ circuits has been settled for 25 years, essentially nothing was known in bounded-depth setting prior to the results of this paper.

\subsubsection*{Monotone Formulas vs.\ Circuits}

The separation of monotone formulas from monotone circuits was shown by Karchmer and Wigderson \cite{karchmer1990monotone} via a lower bound for directed st-connectivity ($\STCONN$). 

\begin{thm}\label{thm:KW}
Monotone formulas solving $\STCONN$ require size $n^{\Omega(\log n)}$.
\end{thm}

As it was already known that $\STCONN$ has polynomial-size monotone circuits, Theorem \ref{thm:KW} implies the separation of monotone classes $\mathsf{mNC^1}$ and $\mathsf{mP}$ (in fact, it shows $\mathsf{mNC^1} \ne \mathsf{mAC^1}$). (In a notable recent development, Potechin \cite{potechin2010bounds} showed that monotone switching networks for $\STCONN$ require size $n^{\Omega(\log n)}$. This result strengthens Theorem \ref{thm:KW} and implies the sharper separation $\mathsf{mL} \ne \mathsf{mNL}$.)

\subsubsection*{Bounded-Depth Formulas vs.\ Circuits}

The {\em bounded-depth} setting refers to the class of unbounded fan-in boolean circuits and formulas of depth $\le d(n)$ for some (not necessarily constant) function $d : \N \to \N$. Unlike the monotone setting, the question of bounded-depth formulas vs.\ questions gives a natural approach to Question \ref{qu:NC1}: by comparing the power of depth-$d$ formulas vs.\ depth-$d$ circuits, we can hope to get a separation for as large a depth $d(n)$ as possible, noting that a super-polynomial separation for any $d(n) = \log n$ would imply $\mathsf{NC^1} \ne \mathsf{AC^1}$ (answering Question \ref{qu:NC1}).

We write $\CircSD(s,d)$ (resp.\ $\FormSD(s,d)$) for the class of languages computable by unbounded fan-in boolean circuits (resp.\ formulas) of size $\le s(n)$ and depth $\le d(n)$. Consider the elementary fact that $\CircSD(s,d) \subseteq \FormSD(s^d,d)$, that is, every depth-$d$ circuit of size $s$ is equivalent to a depth-$d$ formula of size $\le s^d$. In the naive simulation of circuits by formulas, we simply replace overlapping sub-circuits with non-overlapping copies until the circuit becomes a tree. Note that this give a slightly better upper bound of $\text{fan-in}^d$. It is natural to ask: is this naive simulation of depth-$d$ circuits by depth-$d$ formulas asymptotically optimal? To make this question meaningful, we focus on the case where $s(n)$ is any $n^{O(1)}$ and $d(n) \le \log n$. Thus, $\CircSD(n^{O(1)},d) \subseteq \FormSD(n^{O(d)},d)$ and we can ask whether $n^{O(d)}$ can be improved to $n^{o(d)}$.

\begin{qu}\label{qu:max-sep}
For which functions $d(n) \le \log n$ do we have 
\begin{equation}
\tag{$\ast$}\label{eq:it}
  \CircSD(n^{O(1)},d) \nsubseteq \FormSD(n^{o(d)},d)?
\end{equation}
\end{qu}

On the basis of problems like $\STCONN$, we conjecture that (\ref{eq:it}) holds for all $d(n) \le \log n$. Of course, since this (more than) implies $\mathsf{NC}^1 \ne \mathsf{AC}^1$, we should not expect to prove (\ref{eq:it}) all the way to depth $\log n$ anytime soon. On the other hand, more modest depths like $O(\log\log n)$ are well within the range of techniques like switching lemmas (after all, the super-polynomial lower bounds for parity extend to depth $o(\log n/\log\log n)$ \cite{hastad1987computational}). For this reason, it might seem that (\ref{eq:it}) is the kind of statement that ought to be known (or follow from known results) for modest but super-constant $d(n)$. (Note that (\ref{eq:it}) is trivial for constant $d(n) \le O(1)$.) However, it turns out that the status of (\ref{eq:it}) was entirely unknown for all $d(n) \nleq O(1)$. Even the weakest possible separation $\CircSD(n^{O(1)},d) \nsubseteq \FormSD(n^{O(1)},d)$ (i.e.\ $\FormSD(n^{O(1)},d) \subsetneqq \CircSD(n^{O(1)},d)$) was not known to hold for any $d(n) \nleq O(1)$. In this paper, we improve this state of affairs by showing that (\ref{eq:it}) holds for all $d(n) \le \log\log\log n$ (Corollary \ref{cor:CleF}). 

At this point, we should ask: why do the previous techniques (in particular, switching lemmas \cite{hastad1987computational} and approximation by low-degree polynomials \cite{Razborov87,Smolensky87}) fail to distinguish formulas from circuits? In other words, why don't these techniques imply stronger lower bounds for depth-$d$ formula complexity vis-\`a-vis depth-$d$ circuit complexity of a given boolean function?  We suggest that this is the consequence of a certain kind of {\em bottom-up} depth-reduction argument.\footnote{
This style of lower bound has the following elements:
\begin{enumerate}[\ \ \ \hspace{.2pt}$\bullet$]
\setlength{\itemsep}{0pt}
\item
For some notion of ``simple'' functions and some notion of ``approximation'', there is a lemma of the form: if $g$ is the AND or OR of simple functions $f_1,\dots,f_{\mr{poly}(n)}$, then $g$ is approximated by a (slightly less) simple function. (For example: if $f_i$ are small decision trees, then after a random restriction, $g$ simplifies to a small decision tree with high probability; if $f_i$ are low-degree polynomials, then $g$ agrees with a low-degree polynomial up to small error.)
\item
Given a polynomial-size depth-$d$ circuit $C$ sitting on top of a layer of input variables (themselves simple functions), the approximation lemma is applied {\em independently} to all gates at the bottom level (directly above the inputs). $C$ is thus transformed into a circuit $C'$ of depth $d-1$ sitting on top of a layer of simple functions. This depth-reduction step is repeated $d$ times. Finally, we have a function which is sufficiently simple to imply a lower bound for whatever (non-simple) function one has in mind.
\end{enumerate}
It is precisely because the approximation lemma is applied {\em independently} to all bottom-level gates that the distinct between circuits and formulas is lost.}
At the same time, {\em top-down} lower bound techniques (in particular, Karchmer-Wigderson games \cite{karchmer1990monotone}) have never been successful used in the boolean setting, in contrast to the monotone setting.\footnote{One exception is a top-down lower bound for depth-$3$ circuits due to Jukna, Pudl\'ak and H{\aa}stad \cite{haastad1995top}, who pose the problem of proving {\em any} super-polynomial lower bound at depth $4$ by a purely top-down argument.} Our lower bound technique gets around the limitations of previous techniques by a novel combination of bottom-up and top-down arguments. In particular, the part of our proof which distinguishes formulas from circuits is a new top-down argument (Lemma \ref{la:halcali}).

\subsubsection*{Distance $k(n)$ Connectivity}

As with the separation of monotone formulas vs.\ circuits in \cite{karchmer1990monotone}, our separation of bounded-depth formulas vs.\ circuits comes by way of a lower bound for (a parameterized version of) st-connectivity. 
As Wigderson wrote in his excellent survey on graph connectivity \cite{wigderson1992complexity}, ``Of all computational problems, graph connectivity is the one that has been studied on the largest variety of computational models, such as Turing machines, PRAMs, Boolean circuits, decision trees and communication complexity. It has proven a fertile test case for comparing basic resources such as time vs.\ space, nondeterminism vs.\ randomness vs.\ determinism, and sequential vs.\ parallel computation.'' There has been some significant progress in the 20 years since \cite{wigderson1992complexity}. One notable result is Reingold's theorem \cite{reingold2008undirected} that $\USTCONN$ (undirected st-connectivity) $\in \DSPACE(\log n)$. However, many questions remain open. Chief among these is the space complexity of $\STCONN$. Savitch's theorem \cite{savitch1970relationships} that $\STCONN \in \DSPACE(\log^2 n)$ is still the best known upper bound. 

As for lower bounds for $\STCONN$, in addition to various results in monotone models of computation \cite{karchmer1990monotone, potechin2010bounds, raz1989probabilistic, shamir1979depth, tiwari1994direct}, there are results on structured models of computations whose basic operations manipulate pebblings on graphs.
One result of this type, due to Edmonds, Poon and Achlioptas \cite{edmonds1999tight}, gives a tight space lower bound of $\Omega(\log^2 n)$ on the NNJAG model. Another interesting result, in the unusual restricted model of arithmetic circuits with $\times$ gates of odd fan-in, is a tight lower bound of $n^{\Omega(\log n)}$ for $\STCONN$ (or more accurately its algebraic cousin, iterated matrix multiplication) was shown by Nisan and Wigderson \cite{nisan1996lower} using the method of partial derivatives.

In this paper, we consider a version of $\STCONN$ parameterized by distance. For a function $k : \N \to \N$ with $k(n) \le n$, {\em distance $k(n)$ connectivity}, denoted $\STCONN(k(n))$, is the following problem: given a directed graph with $n$ vertices and specified vertices $s$ and $t$, determine whether or not there is a path of length at most $k(n)$ from $s$ to $t$. Unlike $\STCONN$ and $\USTCONN$, the directed and undirected versions of distance $k(n)$ connectivity are essentially equivalent.\footnote{The reduction from $\STCONN(k(n))$ to $\USTCONN(k'(n'))$ converts a directed graph on $n$ vertices into a layered undirected graph on $n' = (k+1)n$ vertices where $k'(n') = k(n)$.}
It is easy to show that $\STCONN(k(n))$ has circuits (moreover, semi-unbounded monotone circuits) of size $O(kn^3)$ and depth $2\log k$ using the recursive doubling (a.k.a.\ repeated squared) method of Savitch \cite{savitch1970relationships}. (At the expense of larger depth, one can get smaller circuits of size $O(kn^{2.38})$ using fast matrix multiplication.)

An important relationship between $\STCONN$ and its parameterized version $\STCONN(k(n))$ is the fact every algorithm for $\STCONN(k(n))$ ``scales up'' to an algorithm for $\STCONN$ by recursive $k$th powering. (Conversely, every lower bound for $\STCONN$ ``scales down'' to a lower bound for $\STCONN(k(n))$. In particular, Theorem \ref{thm:KW} implies that monotone formulas solving $\STCONN(k(n))$ require size $n^{\Omega(\log k)}$.) For circuits, we have the implication:
\[
  \STCONN(k(n)) \in \CircSD(s,d) \,\Longrightarrow\, \STCONN \in \CircSD(n^{O(1)}\cdot s,\frac{\log n}{\log k}\cdot d).
\]
As noted in \cite{wigderson1992complexity}, if $\STCONN(k(n))$ has polynomial-size circuits of depth $o(\log k)$, then $\STCONN$ has polynomial-size circuits of depth $o(\log n)$ and hence $\STCONN \in \DSPACE(o(\log^2 n))$. This observation strongly motivates the following:

\begin{qu}\label{qu:min-depth}
What is the minimum depth of polynomial-size circuits solving $\STCONN(k(n))$?
\end{qu}

Furst, Saxe and Sipser \cite{FSS84} showed that $\STCONN \notin \AC^0$ via the reduction from parity to $\STCONN$. Via the same reduction, it follows from the parity lower bound of H{\aa}stad \cite{hastad1987computational} that $\STCONN(k(n)) \notin \AC^0$ for all $k(n) \nleq \log^{O(1)}n$. However, this says nothing when $k(n) \le \log^{O(1)} n$. 

Ajtai \cite{ajtai1989first} proved the first lower bound for small distances $k(n)$, showing that $\STCONN(k(n))$ $\notin \mathsf{AC^0}$ for all $k(n) \nleq O(1)$. Via an explicit vetsion of Ajtai's originally non-constructive proof, Bellantoni, Pitassi and Urquhart \cite{bellantoni1992approximation} proved a lower bound of $\Omega(\log^\ast k)$ on the depth of polynomial-size circuits solving $\STCONN(k(n))$. This was subsequently improved to $\Omega(\log\log k)$ for all $k(n) \le \log^{O(1)} n$ by Beame, Impagliazzo and Pitassi \cite{beame1998improved}, using a special-purpose ``connectivity switching lemma'' tailored to $\STCONN(k(n))$. It was left as an open problem to further narrow the gap between the $O(\log k)$ and $\Omega(\log\log k)$ upper and lower bounds. In this paper, we completely close this gap by proving a lower bound of $\Omega(\log k)$ for all $k(n) \le \log\log n$ (Corollary \ref{cor:min-depth}). (While our current proof is restricted to $k(n) \le \log\log n$, we believe this can be extended $k(n) \le \log^{O(1)}\log n$ as in \cite{beame1998improved}.) The significance of this result is that, for small but super-constant $k(n)$, we rule out the possibility of showing that $\STCONN \in \DSPACE(o(\log^2 n))$ by constructing polynomial-size circuits for $\STCONN(k(n))$ of depth $o(\log k)$.

\section{Our Results}

Our main theorem is a tight lower bound on the size of bounded-depth formulas solving distance $k(n)$ connectivity.

\begin{thm}[Main Result]\label{thm:main}
Formulas of depth $\log n/(\log\log n)^{O(1)}$ solving $\STCONN(k(n))$ have size $n^{\Omega(\log k)}$ for all $k(n) \le \log\log n$.
\end{thm}

To be precise, we get a lower bound of $n^{(1/6)\log k - O(1)}$ for formulas of depth $\smash{c\log n/\max\{k^3{\log}{}^2 k,}$ $\smash{k\log k\log\log n\}}$ 
where $c > 0$ is an absolute constant. Moreover, this lower bound is not only worst-case: it applies to formulas solving $\STCONN(k(n))$ in the natural {\em average-case} sense (see \S\ref{sec:conclusion}).

The following two corollaries of Theorem \ref{thm:main} were already mentioned in the introduction. As discussed, these corollaries answer Questions \ref{qu:max-sep} and \ref{qu:min-depth} for a limited range of $d(n)$ and $k(n)$.

\begin{cor}\label{cor:min-depth}
Polynomial-size circuits solving $\STCONN(k(n))$ require depth $\Omega(\log k)$ for all $k(n) \le \log\log n$.
\end{cor}

\begin{proof}
For contradiction, assume $C$ is a circuit of size $s(n) = n^{O(1)}$ and depth $d(n) = o(\log k)$ solving $\STCONN(k(n))$ for some $k(n) \le \log\log n$. By the naive simulation of circuits by formulas, $C$ is equivalent to a depth-$d$ formula of size at most $s^d = n^{o(\log k)}$. But since $d(n) = o(\log\log\log n) \ll \log n/(\log\log n)^{O(1)}$, we get contradiction with Theorem \ref{thm:main}.
\end{proof}

\begin{cor}\label{cor:CleF}
It is impossible to simulate polynomial-size depth-$d$circuits by depth-$d$ formulas of size $n^{o(d)}$ (that is, we get the optimal separation $\CircSD(n^{O(1)},d) \nsubseteq \FormSD(n^{o(d)},d)$) for all $s(n) = n^{O(1)}$ and $d(n) \le \log\log\log n$.
\end{cor}

\begin{proof}
The separating language is $\STCONN(k(n))$ where $k(n) = 2^{d(n)/2}$ ($\le \log\log n$). We have $\STCONN(k(n)) \in \CircSD(n^{O(1)},d)$ by the circuits (of depth $2\log k = d$) which implement recursive doubling. The lower bound $\STCONN(k(n)) \notin \FormSD(n^{o(d)},d)$ is by Theorem \ref{thm:main}, noting that $d(n) \le \log\log\log n \ll \log n/(\log\log n)^{O(1)}$.
\end{proof}

\section{Proof Overview}

Our proof technique is centered on a new notion of {\em pathset complexity}. Informally, a {\em pathset} is a subset $\A \subseteq [n]^{k+1}$ whose elements represent potential paths of length $k$ in a graph of size $n$. The {\em pathset complexity} of $\A$, denoted $\cost(\A)$, measures of the minimum number of operations required to construct $\A$ via unions ($\cup$) and relational join ($\bowtie$), subject to certain density constraints. (The formal definition of $\cost(\A)$, given in \S\ref{sec:pathset}, is not important for this overview.)

The proof of Theorem \ref{thm:main} has two parts. Part 1 shows that every bounded-depth formula $F$ solving $\STCONN(k(n))$ implies an upper bound on the pathset complexity of a certain (random) pathset $\A^\Gamma$. Part 2 is a general lower bound on $\cost(\A)$ for arbitrary pathsets $\A$. Combining these two parts, we get the desired $n^{\Omega(\log k)}$ lower bound on the size of $F$.

Before explaining Parts 1 and 2 in more detail, we state the key property of $\STCONN(k(n))$ which our proof exploits. Instances for $\STCONN(k(n))$ are directed graphs with vertex set $[n]$ and distinguished vertices $s$ and $t$ (without loss of generality $s=1$ and $t=2$). An {\em $st$-path} is a sequence $(x_0,\dots,x_k) \in [n]^{k+1}$ such that $x_0 = s$ and $x_k = t$ and $x_i \ne x_j$ for all $i \ne j$. 

Denote by $\Gamma$ the random directed graph with edge probability $1/n$. (Note that $1/n$ is below the threshold for $\STCONN(k(n))$, that is, almost surely $\Gamma$ contains no $st$-path of length $k$.) Define $\A^\Gamma$ as the set of $st$-paths $(x_0,\dots,x_k) \in [n]^{k+1}$ such that
\begin{itemize}
  \item
    $(x_0,x_1),\dots,(x_{k-1},x_k)$ are non-edges of $\Gamma$,
  \item
    $\Gamma \cup \{(x_0,x_1),\dots,(x_{k-1},x_k)\}$ contains a unique $st$-path of length $k$ (namely, $(x_0,\dots,x_k)$).
\end{itemize}
Then the (average-case) property of $\STCONN(k(n))$ that our proof exploits is:

\paragraph{{\bf Key Property }
{\bf(\S\ref{sec:output}):}} {\it
Almost surely, $\A^\Gamma$ contains 99\% of $st$-paths of length $k$.}\bigskip

We now state Parts 1 and 2 of the proof of Theorem \ref{thm:main} in more detail.

\paragraph{\bf{Part 1 
(\S\ref{sec:outline}--\ref{sec:smallness}):}} 
{\it Suppose $F$ is a formula of depth $\log n/(\log\log n)^{O(1)}$ solving $\STCONN(k(n))$. Then, almost surely (with respect to $\Gamma$),}
\begin{equation}\label{eq:part1}
  \size(F) \ge 2^{-O(k^2)} \cdot n^{-O(1)} \cdot \cost(\A^\Gamma).
\end{equation}

\paragraph{\bf{Part 2 (\S\ref{sec:patterns}--\ref{sec:lower-bound}):}}
{\it For all pathsets $\A \subseteq [n]^{k+1}$, writing $\delta(\A) \defeq |\A|/n^{k+1}$ for the density of $\A$, 
}
\begin{equation}\label{eq:part2}
  \cost(\A) \ge 
  2^{-O(2^k)} 
  \cdot
  n^{\Omega(\log k)}
  \cdot 
  \delta(\A).
\end{equation}

Combining (\ref{eq:part1}) and (\ref{eq:part2}) with $\delta(\A^\Gamma) \ge .99n^{-2}$ (by the key property), we get the lower bound $\size(F) \ge 2^{-O(2^k)} \cdot n^{\Omega(\log k)}$. Since $2^{-O(2^k)}$ is $n^{-O(1)}$ for $k(n) \le \log\log n$, Theorem \ref{thm:main} is proved.

Part 1 builds on the technique of \cite{rossman2008constant,rossman2010average}. An essential new ingredient, which distinguishes formulas from circuits, is a top-down argument (Lemma \ref{la:halcali}) relating formula size to pathset complexity. 

For Part 2, we develop a combinatorial framework for studying pathset complexity. This involves analyzing the {\em pattern} of joins which predominates the construction of a given pathset $\A$. In \S\ref{sec:patterns} we define an auxiliary notion of {\em pathset complexity with respect to a pattern}, denoted $\pcost(\A)$. Part 2 then consists of 2a and 2b:

\paragraph{\bf{Part 2a (\S\ref{sec:relationship}):}}
{\it For every pathset $\A$, there exists $\A' \subseteq \A$ such that $\cost(\A) \ge \pcost(\A')$ and 
$\delta(\A') \ge 2^{-O(2^k)} \cdot \delta(\A)$.}

\paragraph{\bf{Part 2b (\S\ref{sec:lower-bound}):}}
{\it For all pathsets $\A'$, $\pcost(\A') \ge n^{\Omega(\log k)} \cdot \delta(\A')$.}\bigskip

Part 2a is relatively straightforward. This move from $\cost$ to $\pcost$ is precisely where we lose the factor of $2^{O(2^k)}$, which is the reason that our main theorem is limited to $k(n) \le \log\log n$. (If this factor can be removed, which I believe is possible (with a lot more work) within the current framework, then Theorem \ref{thm:main} and Corollary \ref{cor:min-depth} would hold up to $k(n) \le \log^{1/3} n$ and Corollary \ref{cor:CleF} would hold up to $d(n) \le \log\log n$.)

Part 2b is the true combinatorial lower bound at the heart the paper. The proof involves an intricate induction on patterns.

\subsubsection*{Organization of the Paper}

Section \ref{sec:prelims} sets out the basic terminology and notation for the paper. Section \ref{sec:pathset} introduces the key notion of {\em pathset complexity}. Sections \ref{sec:outline}--\ref{sec:smallness} contain Part 1 of the proof of Theorem \ref{thm:main}. Sections \ref{sec:patterns}--\ref{sec:lower-bound} contain Part 2 of the proof. We state some conclusions and discuss future directions in Section \ref{sec:conclusion}. Three appendices (Sections \ref{sec:key-examples}--\ref{sec:rectangular}) contain supplementary material including key examples and relatively easier special cases of our main lower bound.

\newpage

\section{Preliminaries}\label{sec:prelims}

Let $n$ be an arbitrary positive integer (which we view as growing to infinity). Let $[n] \defeq \{1,\dots,n\}$. We note that, for all purposes in this paper, $[n]$ may be regarded as an arbitrary fixed set of size $n$. Let $k = k(n)$ and $d = d(n)$ be arbitrary functions of $n$. As parameters, $k$ represents {\em distance} and $d$ represents {\em depth}. No bound on $k$ or $d$ is assumed throughout the paper; assumptions like $k(n) \le \log\log n$ are explicitly stated where needed. All constants in asymptotic notation ($O(\cdot)$, etc.)\ are universal (with no dependence on $n,k,d$).

\prelim{Circuits and Formulas}

The {\em circuits} and {\em formulas} considered in this paper are unbounded fan-in boolean circuits and formulas with a single output node and NOT gates at the bottom level. Formally, a {\em circuit} is a finite acyclic directed graph with a unique output (node of out-degree $0$) where each input (node of in-degree $0$) is labeled by a literal (i.e.\ $X_i$ or $\BAR X_i$) and each gate (node of in-degree $\ge 1$) is labeled by AND or OR. A {\em formula} is a tree-like circuit in which every node other than the output has out-degree $1$. The {\em size} of a circuit is the number of gates, while the {\em size} of a formula is the number of leaves. (For a formula $F$, the circuit-size of $F$ equals the formula-size of $F$ minus $1$.)

\prelim{Graphs}

All {\em graphs} in this paper are directed graph $G = (V_G,E_G)$ where $V_G$ is a (possibly empty) set and $E_G \subseteq V_G \times V_G$. The edge from $v$ to $w$ is written simply as $\edge{v}{w}$ to cut down on unnecessary parentheses.

Two important graphs in this paper are $P_k$ (the directed path of length $k$) and $P_{k,n}$ (the ``complete $k$-layered graph'' with $k+1$ layers of $n$ vertices and $kn^2$ edges). Formally, let
\[
  P_k = (V_k,E_k)
  \text{ where }
  V_k = \{v_0,\dots,v_k\}
  \text{ and } 
  E_k = \{\edge{v_i}{v_{i+1}} : 0 \le i < k\}
\] 
where $v_0,\dots,v_k$ are fixed abstract vertices. We will usually omit subscripts writing simply $v$ and $\edge{v}{w}$ for arbitrary elements of $V_k$ and $E_k$. To define $P_{k,n}$, we create $(k+1)n$ fresh vertices denoted $\lift{v}{i}$ for each $v \in V_k$ and $i \in [n]$. Then
\[
  P_{k,n} = (V_{k,n},E_{k,n})
  \text{ where }
  V_{k,n} = \{\lift{v}{i} : v \in V_k,\, i \in [n]\}
  \text{ and }
  E_{k,n} = \{\edge{\lift{v}{i}}{\lift{w}{j}} : \edge{v}{w} \in E_k,\, i,j \in [n]\}.
\]

We refer to subgraphs $\Gamma \subseteq P_{k,n}$ with $V_\Gamma = V_{k,n}$ as {\em $k$-layered graphs}. Throughout the paper, $\Gamma$ consistently represents a (random) $k$-layered graph, while $G,H,K$ are reserved for subgraphs of $P_k$. We sometimes view $\Gamma$ as the input to a circuit or formula; in this case, we identify the set of layered graphs with $\{0,1\}^N$ where $N$ is a set of $kn^2$ variables indexed by elements of $E_{k,n}$.

\prelim{Layered Distance $k(n)$ Connectivity}

As with previous lower bounds for distance $k(n)$ connectivity \cite{ajtai1989first, beame1998improved}, we consider a variant of the problem on $k$-layered graphs. Let $s,t$ denote vertices $\lift{v_0}{1},\lift{v_k}{1}$ respectively. {\em Layered distance $k(n)$ connectivity} is the problem of determining whether a layered graph $\Gamma \in \{0,1\}^N$ contains a path from $s$ to $t$. Following \cite{beame1998improved}, we denote this problem by $\PATH(k,n)$.
The layered and unlayered versions of distance $k(n)$ connectivity are essentially equivalent.\footnote{Since $k$-layered graphs are graphs with $(k+1)n$ vertices, there is a trivial reduction from $\PATH(k,n)$ to $\STCONN(k'(n'))$ where $n' = (k+1)n$ and $k'(n') = k$. In the opposite direction, there is a simple reduction from $\STCONN(k(n))$ to $\PATH(k,n)$ which converts graphs to $k$-layered graphs.} This allows us to restate Theorem \ref{thm:main} as a lower bound on $\PATH(k,n)$: 

\begin{reptheorem}{thm:main}
\textup{(restated)}
\itshape
Solving $\PATH(k,n)$ on formulas of depth $\smash{\thedepth}$ requires size $n^{(1/6)\log k - O(1)}$ for all $k(n) \le \log\log n$.
\end{reptheorem}

This restatement includes explicit expressions $n^{(1/6)\log k - O(1)}$ for $n^{\Omega(\log k)}$ and $\smash{\thedepth}$ for $\log n/(\log\log n)^{O(1)}$.\footnote{As mentioned earlier, our proof actually extends to depth $O(\log n/\max\{k^3{\log}{}^2 k,\, k\log k\log\log n\})$; in particular, this is $O(\frac{\log n}{k\log k\log\log n})$ for very small $k(n) \le \log^{1/3}\log n$. We state Theorem \ref{thm:main} with depth $\thedepth$ for the sake of simplicity.}

\prelim{Boolean Functions and Restrictions}

Let $f : \{0,1\}^I \to \{0,1\}$ be a boolean function where $I$ is an arbitrary finite set (of ``variables''). We say that a variable $i \in I$ is {\em live} with respect to $f$ if there exists $x \in \{0,1\}^N$ such that $f(x) \ne f(x')$ where $x'$ equals $x$ with its $i$th coordinate flipped. Let $\Live(f) \defeq \{i \in I : i \text{ is live w.r.t.\ }f\}$.

A {\em restriction} on $I$ is any function $\theta : I \to \{0,1,\ast\}$. We denote by $f{\lceil}\theta : \{0,1\}^{\theta^{-1}(\ast)} \to \{0,1\}$ the function (over the ``unrestricted'' variables $i$ such that $\theta(i) = \ast$) obtained from $f$ by applying the restriction $\theta$.

\prelim{Probabilistic Notation}

For a finite set $I$ and $0 \le p,q \le 1$, we write:
\begin{itemize}
\item
$x \in \{0,1\}^I_p$ for the random tuple $x \in \{0,1\}^I$ where $\Pr[\,x_i = 1\,] = p$ independently for all $i \in I$ (in particular, we will consider the random layered graph $\Gamma \in \{0,1\}^N_{1/n}$),
\item
$R \subseteq_p I$ for the random subset $R$ of $I$ where $i \in R$ independently with probability $p$ for all $i \in I$,
\item
$\theta \in \mc R(p,q)$ for the random restriction $\theta : I \to \{0,1,\ast\}$ where $\Pr[\,\theta(i) = \ast\,] = q$ and $\Pr[\,\theta(i) = 1\,] = (1-q)p$ for all $i \in I$. 
\end{itemize}
Whenever we say {\em almost surely}, this is understood to mean {\em asymptotically almost surely} as $n \to \infty$ (i.e.\ with probability that goes to $1$ as $n \to \infty$).

\prelim{Tuples and Relations}

The following notation pertains to ``$V$-ary'' tuples $x \in [n]^V$ and relations $\A \subseteq [n]^V$ where $V$ is an arbitrary finite set.

\begin{df}[$V$-tuples] 
For $x \in [n]^V$ and $S \subseteq V$, we denote by $x_S \in [n]^S$ the restriction of $x$ to coordinates in $S$. For $x \in [n]^V$ and $y \in [n]^W$ where $V \cap W = \emptyset$, let $xy \in [n]^{V \cup W}$ denote the unique $z \in [n]^{V \cup W}$ such that $z_i = x_i$ for all $i \in V$ and $z_j = y_j$ for all $j \in W$; here $xy = yx$, as there is no intrinsic linear order on $V \cup W$. We adopt the convention $[n]^\emptyset = \{()\}$ where $()$ denotes the unique $\emptyset$-tuple.
\end{df}

\begin{df}[Join]
For finite sets $V$ and $W$ and $\A \subseteq [n]^V$ and $\B \subseteq [n]^W$, the {\em join} of $\A$ and $\B$ is the set
\[
   \A \bowtie \B \defeq \{x \in [n]^{V \cup W} : x_V \in \A \text{ and } x_W \in \B\}.
\]
\end{df}

The join operation $\bowtie$ is a hybrid of intersection $\cap$ and cartesian product $\times$: if $V = W$ then $\A \bowtie \B = \A \cap \B$, and if $V \cap W = \emptyset$ then $\A \bowtie \B$ is the product $\A \times \B$. Note that $\A \bowtie \emptyset = \emptyset$ and $\A \bowtie \{()\} = \A$.

\newpage

\begin{df}[Density, Projection, Restriction]\label{df:density-etc}
Let $\mc A \subseteq [n]^V$.
\begin{enumerate}[(i)]
\item
The {\em density} of $\mc A$ is defined by
$\ds\delta(\A) \defeq |\A|\,/\,n^{|V|}$.
\item
For $S \subseteq V$, the {\em $S$-projection} and {\em $S$-projection density} of $\mc A$ are defined by
\begin{align*}
  \proj_S(\A) 
    &\defeq \{x_S : x \in \A\},
  \quad\ \ 
  \pi_S(\A) 
    \defeq \delta(\proj_S(\A)).
\end{align*}
That is, $\pi_S(\A) = |\proj_S(\A)| \,/\, n^{|S|}$, as $\delta$ here refers to the density of the $S$-ary relation $\proj_S(\A) \subseteq [n]^S$.
\item
For $S \subseteq V$ and $z \in [n]^{V \setminus S}$, the {\em $S$-restriction} of $\A$ at $z$ and {\em maximum $S$-restriction density} of $\A$ are defined by
\begin{align*}
  \RHO{S}{\A}{z} 
    &\defeq \{y \in [n]^S : yz \in \A\},
  \quad\ \ 
  \mu_S(\A) 
    \defeq \max_{z \in [n]^{V \setminus S}} 
  \delta(\RHO{S}{\A}{z}).
\end{align*}
It will be convenient (later on in \S\ref{sec:proj-rest}) to extend this notation as follows: for any sets $S$ and $\BAR S$ such that $S \cap \BAR S = \emptyset$ and $V \subseteq S \cup \BAR S$ and any $z \in [n]^{\BAR S}$, let $\RHO{S}{\A}{z}$ be understood as $\RHO{V \cap S}{\A}{z'}$ where $z' = z_{V \cap \BAR S}$.
\end{enumerate}
\end{df}

We conclude this section with a lemma which gives some basic inequalities relating the densities of projections, restrictions and joins. In particular, inequality (\ref{eq:basic2}), bounding the density of a join, will play a crucial role later on.

\begin{la}\label{la:basic-ineqs}
For all $\A \subseteq [n]^V$ and $\B \subseteq [n]^W$ and $S^{-} \subseteq S \subseteq S^{+} \subseteq V$ and $T \subseteq W$,
\begin{enumerate}[\quad\,\normalfont(a)]
  \item\label{eq:basic0}
    $\mu_{S^+}(\A) 
    \le \mu_S(\A) 
    \le \mu_S(\proj_{S^+}(\A)) 
    \le \pi_S(\A) 
    \le \pi_{S^-}(\A)$,
  \item\label{eq:basic1}
    $\delta(\A) \le \pi_S(\A) \cdot  \mu_{V \setminus S}(\A)$,
  \item\label{eq:basic2}
    $\delta(\A \bowtie \B) \le \pi_S(\A) \cdot  \mu_{T \setminus S}(\proj_T(\B)) \cdot  \mu_{(V \cup W) \setminus (S \cup T)}(\A \bowtie \B)$.
\end{enumerate}
\end{la}

\begin{proof}
Inequalities (\ref{eq:basic0}) and (\ref{eq:basic1}) are obvious once the notation is understood. 
Inequality (\ref{eq:basic2}) is mainly derived by two applications of inequality (\ref{eq:basic1}). We first project $\A \bowtie \B$ to $S \cup T$:
\[
  \delta(\A \bowtie \B) 
  \le  
  \pi_{S \cup T}(\A \bowtie \B) 
  \cdot  
  \mu_{(V \cup W) \setminus (S \cup T)}(\A \bowtie \B).
\]
We then project $\proj_{S \cup T}(\A \bowtie \B)$ to $S$:
\[
  \pi_{S \cup T}(\A \bowtie \B) 
  \le
  \pi_S(\A \bowtie \B) 
  \cdot  
  \mu_{T \setminus S}(\proj_{S \cup T}(\A \bowtie \B)).
\]
Finally, we have $\pi_S(\A \bowtie \B) \le \pi_S(\A)$ and 
  $\mu_{T \setminus S}(\proj_{S \cup T}(\A \bowtie \B))
  \le
  \mu_{T \setminus S}(\proj_T(\A \bowtie \B))
  \le
  \mu_{T \setminus S}(\proj_T(\B))$.
Combined, these inequalities give (\ref{eq:basic2}).
\end{proof}

\newpage

\section{Pathset Complexity}\label{sec:pathset}

In this section, we define the key notion of pathset complexity, state our lower bound for pathset complexity (Theorem \ref{thm:pathset-lb}, to be proved in \S\ref{sec:patterns}--\ref{sec:lower-bound}), and present a matching upper bound (Proposition \ref{prop:upper-bound}).

\begin{df}[Pattern Graph]\
Recall that $P_k = (V_k,E_k)$ is the directed path of length $k$ where $V_k = \{v_i : 0 \le i \le k\}$ and $E_k = \{v_iv_{i+1} : 0 \le i < k\}$. A {\em pattern graph} is a subgraph of $P_k$ with no isolated vertices. That is, $G = (V_G,E_G)$ is a pattern graph if, and only if, $E_G \subseteq E_k$ and $V_G = \bigcup_{\edge{v}{w} \in E_G} \{v,w\}$. We write $\wp_k$ for the set of pattern graphs. (We appropriate the power set notation, since pattern graphs are in 1-1 correspondence with subsets of $E_k$.)

Note that every pattern graph is a (possibly empty) disjoint union of directed paths of length $\ge 1$. We refer to maximal connected subsets of $V_G$ simply as {\em components} of $G$. Two important parameters of pattern graphs are the number of components ($=$ the number of maximal paths) and the length of the longest path ($=$ the number of edges in the largest component). These are denoted by
\begin{align*}
  \c{G} &\defeq \#\text{ of components in }G \text{ ($= |V_G| - |E_G|$)},\\
  \l{G} &\defeq \text{length of the longest path in }G.
\end{align*}
\end{df}

\begin{df}[Pathset]
For a pattern graph $G$, let $\P_G$ denote the power set of $[n]^{V_G}$. We refer to elements of $\P_G$ as {\em $G$-pathsets} (or just {\em pathsets} if $G$ is clear from context).
\end{df}

The intuition for pathsets is as follows. For a pattern graph $G$, we view each $x \in [n]^{V_G}$ as corresponding to a ``lifting'' of $G$ inside the complete layered graph $P_{k,n}$, namely isomorphic copy of $G$ with vertex set $\{\lift{v}{i} \in V_{k,n} : i = x_v\}$ and edge set $\{\edge{\lift{v}{i}}{\lift{w}{j}} \in E_{k,n} : i = x_v$ and $j = x_w\}$. In this view, a pathset $\A \subseteq [n]^{V_G}$ corresponds to a set of liftings of $G$. I have chosen to define {\em pathset} as a relation (a subset of $[n]^{V_G}$) rather than a set of liftings (which better matches intuition) in order to more naturally apply operations like $\bowtie$ and $\proj_S$ and $\mu_S$, etc.

\begin{df}[$G$-small]\label{df:small}\
\begin{enumerate}[(i)]
\item
Let $\eps \defeq 1/\log k$ and $\m \defeq n^{1-\eps}$.
\item
A pathset $\A \in \P_G$ is {\em $G$-small} (we simply say {\em small} when $G$ is understood from context) if, for all $1 \le t \le \c{G}$ and $S \subseteq V_G$ such that $S$ is the union of $t$ components of $G$, $\A$ satisfies the density constraint
\[
  \mu_S(\A) \le 
  \m^{-t},
  \quad\text{ that is, }\quad
  \frac{|\{x \in \A : x_{V_G \setminus S} = y\}|}{n^{|S|}}
  \le 
  \m^{-t}
  \text{ for all }y \in [n]^{V_G \setminus S}.
\]
\item
The set of $G$-small pathsets is denoted $\smallP_G$.
\end{enumerate}
\end{df}

\noindent
A few quick remarks:
\begin{enumerate}[---\hspace{.5pt}]
\item
As the terminology suggests, $G$-smallness is a monotone decreasing property (i.e.\ if $\A$ is $G$-small, then so is every $\A' \subseteq \A$).

\item
$G$-smallness consists of $2^{\c{G}}-1$ density constraints on $\A$, corresponding to the nonempty unions of the $\c{G}$ components of $G$. Note that for $t = \c{G}$ and $S = V_G$, the constraint $\mu_S(\A) \le \m^{-t}$ is equivalent to $\delta(\A) \le \m^{-\c{G}}$. In the special case that $G$ is connected (i.e.\ $\c{G} = 1$), $\A$ is $G$-small $\iff$ $\delta(\A) \le \m^{-1}$.

\item
The precise value of $\eps$ is not important: any $\eps$ between $1/k$ and $1/2$ would suit our purposes, modulo a slight weakening in the parameters of our main theorem.\footnote{We choose $\eps = 1/\log k$ so that $\m^{\Omega(\log k)} = n^{\Omega(\log k)}$ with the same constant in the $\Omega(\log k)$. The proof of Lemma \ref{la:pre-main-rest} is the only place where $\eps$ really shows up. Outside this lemma, the difference between $\m$ and $n$ may be ignored (in particular, all statements in \S\ref{sec:patterns}--\ref{sec:lower-bound} are valid if $\m = n$.).} 
\end{enumerate}

\begin{ex}
Let $G$ be the pattern graph with components $U = \{v_1,v_2,v_3\}$ and $U' = \{v_5,v_6\}$ (i.e.\ $V_G = \{v_1,v_2,v_3,v_5,v_6\}$ and $E_G = \{\edge{v_1}{v_2},\edge{v_2}{v_3},\edge{v_5}{v_6}\}$). A pattern $\A \in \P_G$ is $G$-small if, and only if,
\[
  \delta(\A) \le \m^{-2},\quad\
  \mu_U(\A) \le \m^{-1},\quad\
  \mu_{U'}(\A) \le \m^{-1}.
\]
For example, the pathset $\A_1 \defeq \{x : x_1 = x_5 = 1\}$ is $G$-small (here $x$ ranges over $[n]^{V_G}$ and we write $x_i$ for $x_{v_i}$) since $\delta(\A_1) = n^{-2} < \m^{-2}$ and $\mu_U(\A_1) = \mu_{U'}(\A_1) = n^{-1} < \m^{-1}$. The pathset $\A_2 \defeq \{x : x_1 = x_5 \text{ and } x_2 = x_6\}$ is $G$-small as well since $\delta(\A_2) = \mu_U(\A_2) = \mu_{U'}(\A_2) = n^{-2}$. However, pathsets
\[
  \A_3 \defeq \{x : x_1 = x_2 = 1\},\quad\
  \A_4 \defeq \{x : x_1 = x_5\}
\]
are not $G$-small since $\mu_{U'}(\A_3) = 1 > \m^{-1}$ and $\delta(\A_4) = n^{-1} > \m^{-2}$.
\end{ex}

The next lemma shows that smallness is preserved under joins. (Note to the reader: Although it natural to state Lemma \ref{la:join} now, we will not use this lemma until \S\ref{sec:lower-bound}.)

\begin{la}\label{la:join}
If $\A$ is a small $G$-pathset and $\B$ is a small $H$-pathset, then $\A \bowtie \B$ is a small $G \cup H$-pathset.
\end{la}

\begin{proof}
Assume $\A$ is a small $G$-pathset and $\B$ is a small $H$-pathset. To show that $\A \bowtie \B$ is a small $G\cup H$-pathset, consider any $1 \le t \le \c{G \cup H}$ and $S \subseteq V_G \cup V_H$ such that $S$ contains $t$ distinct components $U_1,\dots,U_t$ of $G \cup H$. We must show that $\mu_S(\A \bowtie \B) \le \m^{-t}$. 

Without loss of generality, assume $U_1,\dots,U_t$ are ordered such that, for some $t' \le t$, we have $U_i \cap V_G \ne \emptyset$ for all $1 \le i \le t'$ and $U_j \cap V_G = \emptyset$ for all $t' < j \le t$. Let $S' = S \cap V_G$ and $S'' = U_{t'+1} \cup \dots \cup U_t$. Then $S'$ contains $\ge t'$ components of $G$, since $U_i \cap V_G$ contains $\ge 1$ component of $G$ for all $1 \le i \le t'$. Next note that $U_j$ is a component of $H$ for all $t' < j \le t$, hence $S''$ is a union of $t-t'$ components of $H$. By $G$-smallness of $\A$ and $H$-smallness of $\B$, it follows that
\[
  \mu_{S'}(\A) \le \m^{-t'}
  \quad\text{ and }\quad
  \mu_{S''}(\B) \le \m^{t'-t}.
\]

Now fix $z \in [n]^{(V_G \cup V_H) \setminus S}$ which maximizes $\delta(\pRHO{S}{\A \bowtie \B}{z})$. Using the basic properties of restrictions and joins (Lemma \ref{la:basic-ineqs}(\ref{eq:basic0},\ref{eq:basic1})), we have
\begin{align*}
  \mu_S(\A \bowtie \B)
  = 
  \delta(\pRHO{S}{\A \bowtie \B}{z})
  &=
  \delta((\RHO{S'}{\A}{z}) \bowtie (\RHO{S \cap V_H}{\B}{z}))\\
  &\le
  \delta(\RHO{S'}{\A}{z}) \cdot \mu_{S \setminus V_G}(\RHO{S \cap V_H}{\B}{z})
  \\
  &\le
  \mu_{S'}(\A) \cdot \mu_{S''}(\B).
\end{align*}
It follows that $\mu_S(\A \bowtie \B) \le \m^{-t}$, which completes the proof.
\end{proof}

\begin{df}[Pathset Complexity]\label{df:pc1}
For every pattern graph $G$ and pathset $\A \in \P_G$, the {\em pathset complexity} $\chi_G(\A)$ of $\A$ with respect to $G$ is defined by the following induction:
\begin{enumerate}[(i)]
  \item
    If $G$ is the empty graph, then $\chi_G(\A) \defeq 0$.
  \item
    If $G$ consists of a single edge, then $\chi_G(\A) \defeq |\A|$.
  \item
    If $G$ has $\ge 2$ edges, then 
    \[
  \cost_G(\A) \defeq \min_{(H_i,K_i,\B_i,\C_i)_i} \sum_i \cost_{H_i}(\B_i) + \cost_{K_i}(\C_i)
    \]
    where $(H_i,K_i,\B_i,\C_i)_i$ ranges over sequences\footnote{Without loss of generality, $i$ ranges over $\N$ since $H_i=K_i=G$ and $\B_i=\C_i=\emptyset$ can occur infinitely often.} where
    \[
      H_i,K_i \subset G,
      \quad
      H_i \cup K_i = G,
      \quad
      \B_i \in \smallP_{H_i},
      \quad
      \C_i \in \smallP_{K_i}
      \quad
      \text{and}
      \quad
      {\A\subseteq\bigcup_i \B_i\bowtie\C_i.}\vphantom{\big|}
    \]
\end{enumerate}
In plain language, we consider {\bf\em coverings} of $\A$ by {\bf\em joins} of {\bf\em small} pathsets over {\bf\em proper} subgraphs of $G$. The pathset complexity $\cost_G(\A)$ is the minimum possible value---over all such coverings---of the sum of pathset complexities of the constituent small pathsets.
\end{df}

Note that pathset complexity satisfies the following inequalities:
\begin{align}
\tag{base case}
\vphantom{\big|}
  \cost_\emptyset(\{()\}) \le 0
  \ \hspace{1pt}&\hspace{-1pt}\text{and}\ 
  \cost_G(\A) \le 1
  &&\hspace{-20pt}\text{if }|E_G| = |\A|=1,\\
\tag{monotonicity}
\vphantom{\Big|}
  \cost_G(\A') &\le \cost_G(\A)
  &&\hspace{-20pt}\text{if }\A' \subseteq \A,\\
\tag{sub-additivity}\vphantom{\big|}
  \cost_G(\A_1 \cup \A_2) &\le \cost_G(\A_1) + \cost_G(\A_2)
  &&\hspace{-20pt}\text{for all }\A_1,\A_2,\\
\tag{join rule}\vphantom{\Big|}
  \cost_{G \cup H}(\A \bowtie \B) &\le \cost_G(\A) + \cost_H(\B)
  &&\hspace{-20pt}\text{if }\A \in \smallP_G,\, \B \in \smallP_H.
\end{align}
We will refer to these inequalities repeatedly throughout the paper.

\begin{rmk}\label{rmk:dual1}
Pathset complexity has a {\em dual characterization} as the unique pointwise maximal function from pairs $(G,\A)$ to $\R$ which satisfies (base case), (monotonicity), (sub-additivity) and (join rule). We will expand on this observation later in Remark \ref{rmk:dual2}.
\end{rmk}

We now state our lower bound on pathset complexity (to be proved in \S\ref{sec:patterns}--\ref{sec:lower-bound}).

\begin{thm}[Pathset Complexity Lower Bound]\label{thm:pathset-lb}
For all $\A \in \P_{P_k}$, 
$$
  \ds\cost_{P_k}(\A) \ge 
  \frac{n\vphantom{t}^{(1/6)\log k}}{2^{O(2^k)}} \cdot \delta(\A).
$$
\end{thm}

In particular, for $k \le \log\log n$ and non-negligible $\delta(\A) = n^{-O(1)}$, Theorem \ref{thm:pathset-lb} implies $\cost_{P_k}(\A) \ge n^{(1/6)\log k - O(1)}$. In a moment, we will give an upper bound (Proposition \ref{prop:upper-bound}) which shows that Theorem \ref{thm:pathset-lb} is tight in the regime of $k \le \log\log n$ and non-negligible $\delta(\A)$. First, a couple of remarks which give a different perspective on the definition of $\cost_G(\A)$:

\begin{rmk}[Pathset Complexity as Construction Cost]\label{rmk:construction}
Pathset complexity can be seen as a minimum construction cost. In this view, the goal is to construct a pathset $\A \in \P_G$ out of the fewest possible ``atomic'' pathsets (i.e.,\ individual edges). The rules of construction are as follows:
\begin{enumerate}[\quad(a)]\setlength{\itemsep}{1.5pt}
  \item
    A single ``atomic'' pathset of the form $\A \in \P_G$ where $|E_G| = |\A| = 1$ may be bought for unit cost.
   \item
    Once a pathset $\A$ has been constructed, we may freely discard elements from $\A$ (i.e.\ replace $\A$ with any smaller $\A' \subseteq \A$).
  \item
    Having constructed two $G$-pathsets $\A$ and $\A'$, we may merge $\A$ and $\A'$ into a single $G$-pathset $\A \cup \A'$ (i.e.\ replace $\A$ and $\A'$ with $\A \cup \A'$) at no additional cost.
  \item
    Having constructed a $G$-pathset $\A$ and a $H$-pathset $\B$, provided both $\A$ and $\B$ are {\em small}, we may join $\A$ and $\B$ into a single $G\cup H$-pathset $\A \bowtie \B$ at no additional cost.
\end{enumerate}
For a pathset $\A \in \P_G$, $\cost_G(\A)$ is equal to the minimum cost of constructing $\A$ according to these rules. Construction rules (a), (b), (c), (d) respectively correspond to inequalities (base case), (monotonicity), (sub-additivity), (join rule). Only applications of rule (a) increase cost (so minimum construction cost $=$ fewest application of rule (a)). Rule (b) can be used to convert a non-small pathset into a small pathset (in order to use rule (d), for example). Note that only rule (c) can increase the density of pathsets.
\end{rmk}

\begin{rmk}[The Role of Smallness]\label{rmk:role-of-smallness}
Suppose we modify construction rule (d) by dropping the smallness constraint on $\A$ and $\B$ (this is equivalent to substituting $\P_{H_i}$ and $\P_{K_i}$ for $\smallP_{H_i}$ and $\smallP_{K_i}$ in Definition \ref{df:pc1}(iii)). We could then construct the complete $P_k$-pathset $[n]^{V_k}$ at a total cost of $kn^2$ simply by joining pathsets $[n]^{\{v_i,v_{i+1}\}}$ for $0 \le i < k$. This shows that the smallness constraint on joins is essential to Theorem \ref{thm:pathset-lb}. Intuitively, smallness is responsible for bottlenecks which drive up the cost of constructing sufficiently dense pathsets. However, smallness is not necessarily an obstacle for very sparse pathsets like $[\sqrt n]^{P_k}$: since $[\sqrt n]^{\{v_i,v_{i+1}\}}$ are small, we can take joins showing $\cost_{P_k}([\sqrt n]^{P_k}) \le kn$.
\end{rmk}

We conclude this section with an upper bound.

\begin{prop}[Pathset Complexity Upper Bound]\label{prop:upper-bound}
For all $\A \in \P_{P_k}$, 
$$\ds
  \cost_{P_k}(\A) \le O(kn^{(1/2)\lceil\log k\rceil + 2}).
$$
\end{prop}

For $k \le \log\log n$ and $\A \in \P_{P_k}$ with $\delta(\A) = n^{-O(1)}$, our lower and upper bounds show that $\cost_{P_k}(\A) = n^{\Theta(\log k)}$ where the constant in $\Theta(\log k)$ is between $\frac16$ and $\frac12$. 

\begin{notn}\label{notn:shift1}
For a pattern graph $G$ and an integer $s$, we denote by $G^{\shift s}$ the {\em $s$-shifted} pattern graph with vertex set $\{v_{i+s} : v_i \in V_G\}$ and edge set $\{\edge{v_{i+s}}{v_{i+s+1}} : \edge{v_i}{v_{i+1}} \in E_G\}$. For a pathset $\A \in \P_G$, we denote by $\A^{\shift s} \in \P_{G^{\shift s}}$ the corresponding {\em $s$-shifted} pathset. Note that pathset complexity is invariant under shifts (i.e.\ $\cost_G(\A) = \cost_{G^{\shift s}}(\A^{\shift s})$).
\end{notn}

\begin{proof}[Proof of Proposition \ref{prop:upper-bound}]
For simplicity we assume $\sqrt n$ is an integer. For all $k \ge 1$, define $\A_k \in \smallP_{P_k}$ by
\[
  \A_k \defeq \{x \in [n]^{\{0,\dots,k\}} : x_0,x_k \le \sqrt n\}.
\]
(Note that $\delta(\A_k) = 1/n < 1/\m$, so $\A_k$ is indeed $P_k$-small.)

Letting $j = \lceil k/2 \rceil$, we have
\[
  \A_j \bowtie \A_{k-j}^{\shift j} = \{x \in [n]^{\{0,\dots,k\}} : x_0,x_j,x_k \le \sqrt n\}.
\]
Note that $\A_k$ is covered by $\sqrt n$ ``copies'' of $\A_j \bowtie \A_{k-j}^{\shift j}$ where, for $1 \le t \le \sqrt n$, 
\[
  \COPY_t(\A_j \bowtie \A_{k-j}^{\shift j}) 
  \defeq
  \{x \in [n]^{\{0,\dots,k\}} : x_0,x_k \le \sqrt n
  \text{ and } (t-1)\sqrt n < x_j \le t\sqrt n\}.
\]
Note that pathset complexity is invariant under ``copies'' in this sense (i.e.\ $\cost_G$ is invariant under the action of coordinate-wise permutations of $[n]$ on $\P_G$):
\begin{align*}
  &&\cost_{P_k}(\COPY_t(\A_j \bowtie \A_{k-j}^{\shift j}))
  &=
  \cost_{P_k}(\A_j \bowtie \A_{k-j}^{\shift j})
  &&\text{(invariance under ``copies'')}&&\\
  &&&=
  \cost_{P_j}(\A_j) + \cost_{P_{k-j}^{\shift j}}(\A_{k-j}^{\shift j})
  &&\text{(join rule)}\\
  &&&=
  \cost_{P_j}(\A_j) + \cost_{P_{k-j}}(\A_{k-j})
  &&\text{(invariance under shifts).}
\end{align*}
Since $\A_k \subseteq \bigcup_{1 \le t \le \sqrt n} \COPY_t(\A_j \bowtie \A_{k-j}^{\shift j})$, sub-additivity of $\cost$ implies
\[
  \cost_{P_k}(\A_k) \le \sum_{1 \le t \le \sqrt n} \cost_{P_k}(\COPY_t(\A_j \bowtie \A_{k-j}^{\shift j}))
  = \sqrt n \cdot \big(\cost_{P_j}(\A_j) + \cost_{P_{k-j}}(\A_{k-j})\big).
\]
This recurrence implies
\[
  \cost_{P_k}(\A_k) 
  \le (2\sqrt n)^{\lceil\log k\rceil} \cdot \cost_{P_1}(\A_1)
  = O(kn^{(1/2)\lceil\log k\rceil + 1}).
\]

Now note that the complete $P_k$-pathset $[n]^{V_k}$ is covered by $n$ ``copies'' of $P_k$. Therefore, by a similar argument,
\[
  \cost_{P_k}([n]^{V_k})
  \le n \cdot \cost_{P_k}(\A)
  = O(kn^{(1/2)\lceil\log k\rceil + 2}).
\]
Finally, monotonicity of $\cost$ implies that $\cost_{P_k}(\A) \le O(kn^{(1/2)\lceil\log k\rceil + 2})$ for all $\A \in \P_{P_k}$.
\end{proof}

\section{From Formulas to Pathset Complexity}\label{sec:outline}

In this section we derive our main result (Theorem \ref{thm:main}) from our lower bound on pathset complexity (Theorem \ref{thm:pathset-lb}). Let $F_0$ be a formula of depth $d(n)$ which solves $\PATH(k,n)$ where $k(n) \le \log\log n$ and $d(n) \le \log n/k^3\log\log n$. We must show that $F_0$ has size $n^{\Omega(\log k)}$.

As a first preliminary step: without loss of generality, we assume that $F_0$ has minimal size among all depth $d(n)$ formulas solving $\PATH(k,n)$. In particular, we have $\size(F_0) \le kn^{k-1}$ since $\PATH(k,n)$ has DNFs of this size.

As a second preliminary step, we convert $F_0$ into a fan-in $2$ formula $F$ by replacing each unbounded fan-in AND/OR gate by a balanced binary tree of fan-in $2$ AND/OR gates. We have
\begin{gather*}
  \size(F) = \size(F_0) \le n^k
  \quad\text{ and }\quad
  \depth(F) \le \depth(F_0) \cdot \log(\size(F_0)) \le \log^2 n.
\end{gather*}
We write $F_{\mr{in}}$ for the set of inputs (i.e.\ leaves) in $F$, and $F_{\mr{gate}}$ for the set of gates in $F$, and $f_{\mr{out}}$ for the output gate in $F$. Note that each $f \in F$ is computed by an (unbounded fan-in) formula of size $\le n^k$ and depth $\le d(n)$ (by collapsing all adjacent AND/OR gates below $f$).

In order to lower bound $\size(F)$ in terms of pathset complexity, we define a family of pathsets $\A_{f,G}^\Gamma$ associated with each $f \in F$ and $G \in \wp_k$ and $\Gamma \in \{0,1\}^N$. Recall that we identify $\{0,1\}^N$ with the set of $k$-layered graphs where
$
  N = E_{k,n} 
  = \{\edge{\lift{v}{i}}{\lift{w}{j}} : \edge{v}{w} \in E_k,\, i,j \in [n]\}. 
$

\begin{df}[Pathsets $\A^\Gamma_{f,G}$]
For all $G \in \wp_k$ and $x \in [n]^{V_G}$ and $\Gamma \in \{0,1\}^N$ and $f \in F$: 
\begin{enumerate}[(i)]
  \item
    Let 
    $N_{G,x} \defeq \{\edge{\lift{v}{i}}{\lift{w}{j}} \in N : i = x_v \text{ and } j = x_w\}$ ($= \{\edge{\lift{v}{x_v}}{\lift{w}{x_w}} : \edge{v}{w} \in E_G\}$).
  \item
    Let $\rho^\Gamma_{G,x} : N \to \{0,1,\ast\}$ be the restriction which equals $\ast$ over $N_{G,x}$ and agrees with $\Gamma$ over $N \setminus N_{G,x}$.
    In particular, applying $\rho^\Gamma_{G,x}$ to $f$, we get a function $f\lceil\rho^\Gamma_{G,x} : \{0,1\}^{N_{G,x}} \to \{0,1\}$ (whose variables correspond to edges of $G$ via the bijection $N_{G,x} \cong E_G$).
  \item
    Let $\A^\Gamma_{f,G}$ be the $G$-pathset defined by
    \[
      \A^\Gamma_{f,G} \defeq \{x \in [n]^{V_G} : 
      \Live(f{\lceil}\rho^\Gamma_{G,x}) = N_{G,x}\}.
    \]
    That is, $\A^\Gamma_{f,G}$ is the set of $x \in [n]^{V_G}$ such that the restricted function $f{\lceil}\rho^\Gamma_{G,x}$ depends on all $|N_{G,x}|$ ($=|E_G|$) of its variables.   
\end{enumerate}
\end{df}

In the next three subsections, we prove a sequence of claims about pathsets $\A^\Gamma_{f,G}$ in three cases where $f \in F_{\mr{in}}$ and $f \in F_{\mr{gate}}$ and $f = f_{\mr{out}}$. 

\begin{rmk}
Claims \ref{claim:at-inputs}, \ref{claim:at-gates}, \ref{claim:at-output} rely on few assumptions about $F$. In particular, these claims do not depend on the assumption that $F_0$ has bounded depth (i.e.\ $F$ has bounded alternations), nor even that $F$ is a formula as opposed to a circuit. In fact, these claims are valid if $F$ is any $B_2$-circuit computing $\PATH(k,n)$ where $B_2$ is the full binary basis.

Of course, we will eventually use both assumptions that (I) $F_0$ has bounded depth (i.e.\ $F$ has bounded alternations), and (II) $F$ is a formula as opposed to a circuit. Our main technical lemma (Lemma \ref{la:main-rest}) relies on (I) but not (II) (not surprisingly, since the proof uses the Switching Lemma, which does not distinguish between circuits and formulas). A second key lemma (Lemma \ref{la:halcali}) relies on (II) but not (I) (using a novel top-down argument which only works for formulas).
\end{rmk}

\subsection{Inputs of $F$}\label{sec:inputs}

Suppose $f$ is an input in $F$ labeled by a literal (i.e.\ a variable or its negation) corresponding to some $\edge{\lift{v}{i}}{\lift{w}{j}} \in N$. Then we have the following explicit description of $\A^\Gamma_{f,G}$: 
\begin{itemize}
\setlength{\itemsep}{0pt}
  \item
    if $G$ is the empty graph, then $\A^\Gamma_{f,G} = \{()\}$ (i.e.\ the singleton containing the 0-tuple),
  \item
    if $E_G = \{\edge{v}{w}\}$, then $\A^\Gamma_{f,G} = \{x\}$ for the unique $x \in [n]^{\{v,w\}}$ with $x_v = i$ and $x_w = j$,
  \item
    otherwise (i.e.\ if $|E_G| \ge 2$), $\A^\Gamma_{f,G} = \emptyset$.
\end{itemize}

By the base case conditions (i) and (ii) in Definition \ref{df:pc1} of pathset complexity, we have $\cost_\emptyset(\A) = 0$ and $\cost_G(\A) = |\A|$ if $G$ has a single edge. The upshot of these observations is the following claim.

\begin{claim}[Inputs of $F$]\label{claim:at-inputs}
For all $f \in F_{\mr{in}}$, $\ds\sum_{G \in \wp_k} \cost_G(\A^\Gamma_{G,f}) = 1$.
\end{claim}

\subsection{Gates of $F$}\label{sec:gates}

Suppose $f$ is an AND or OR gate in $F$ with children $f_1$ and $f_2$. Consider any $G \in \wp_k$ and $x \in \A^\Gamma_{f,G}$ (assuming $\A^\Gamma_{f,G}$ is nonempty). By definition of $\A^\Gamma_{f,G}$, the function $f{\lceil}\rho^\Gamma_{G,x} : \{0,1\}^{N_{G,x}} \to \{0,1\}$ depends on all variables in $N_{G,x}$. Since $f{\lceil}\rho^\Gamma_{G,x}$ is the AND or OR of functions $f_1{\lceil}\rho^\Gamma_{G,x}$ and $f_2{\lceil}\rho^\Gamma_{G,x}$, each variable in $N_{G,x}$ is a live variable for one or both $f_1{\lceil}\rho^\Gamma_{G,x}$ and $f_2{\lceil}\rho^\Gamma_{G,x}$. 

Define sub-pattern graph $G_1 \subseteq G$ as follows: for each $\edge{v}{w} \in E_G$, let $\edge{v}{w}$ be an edge in $G_1$ if and only if $\edge{\lift{v}{x_v}}{\lift{w}{x_w}}$ ($\in N_{G,x}$) is a live variable for the function $f_1{\lceil}\rho^\Gamma_{G,x}$. Define $G_2 \subseteq G$ in the same way with respect to $f_2$. Since
\[
  \{\edge{\lift{v}{x_v}}{\lift{w}{x_w}} : \edge{v}{w} \in E_G\} 
  = N_{G,x} 
  = \Live(f{\lceil}\rho^\Gamma_{G,x}) 
  = \Live(f_1{\lceil}\rho^\Gamma_{G,x}) \cup \Live(f_2{\lceil}\rho^\Gamma_{G,x}),
\]
it follows that $G_1 \cup G_2 = G$.

Let $y = x_{V_{G_1}}$ be the restriction of $x$ ($\in [n]^{V_G}$) to coordinates in $V_{G_1}$. By definition of $G_1$, we have
\begin{itemize} 
\setlength{\itemsep}{0pt}
\item
$\edge{\lift{v}{y_v}}{\lift{w}{y_w}} = \edge{\lift{v}{x_v}}{\lift{w}{x_w}} \in \Live(f_1{\lceil}\rho^\Gamma_{G,x})$ for all $\edge{v}{w} \in E_{G_1}$, and
\item
$\edge{\lift{v}{x_v}}{\lift{w}{x_w}} \notin \Live(f_1{\lceil}\rho^\Gamma_{G,x})$ for all $\edge{v}{w} \in E_G \setminus E_{G_1}$.
\end{itemize}
It follows that $\Live(f_1{\lceil}\rho^\Gamma_{G_1,y}) = \Live(f_1{\lceil}\rho^\Gamma_{G,x}) = N_{G_1,y}$, hence $y \in \A^\Gamma_{f_1,G_1}$. Similarly, for $z = x_{V_{G_2}}$, we have $z \in \A^\Gamma_{f_2,G_2}$. This shows that $x \in \A^\Gamma_{f_1,G_1} \bowtie \A^\Gamma_{f_2,G_2}$. 

The observation may be succinctly expressed as
\begin{align*}
  \A^\Gamma_{f,G} \subseteq 
  \bigcup_{\substack{G_1,G_2 \subseteq G \,:\, G_1 \cup G_2 = G}}
  \AA{\Gamma}{G_1}{f_1} \bowtie \AA{\Gamma}{G_2}{f_2}.
\end{align*}
Splitting this union into the cases that $G_1 = G$ or $G_2 = G$ or $G_1,G_2 \subset G$, we have proved:

\begin{claim}[Gates of $F$]\label{claim:at-gates}
For every $f \in F_{\mr{gates}}$ with children $f_1,f_2$ and every $G \in \wp_k$,
\[
  \A^\Gamma_{f,G} \subseteq 
  \AA{\Gamma}{G}{f_1} \cup 
  \AA{\Gamma}{G}{f_2} \cup
  \bigcup_{\substack{G_1,G_2 \subset G \,:\, G_1 \cup G_2 = G}}
  \AA{\Gamma}{G_1}{f_1} \bowtie \AA{\Gamma}{G_2}{f_2}.
\]
\end{claim}

\subsection{Output of $F$}\label{sec:output}

We now use the fact that $F$ computes $\PATH(k,n)$. Our previous Claims \ref{claim:at-inputs} and \ref{claim:at-gates} applied to arbitrary $\Gamma \in \{0,1\}^N$. We now shift perspective and consider {\em random} $\Gamma \in \{0,1\}^N_{1/n}$. That is, $\Gamma$ is the random $k$-layered graph (i.e.\ subgraph of $P_{k,n}$) with edge probability $1/n$. Recall that $V_{k,n} = \{\lift{v}{i} : v \in V_k \text{ and } i \in [n]\}$ and $s,t$ are the vertices $\lift{v_0}{1},\lift{v_k}{1}$. Each $x \in \smash{[n]^{V_k}}$ corresponds to a path of length $k$ in $P_{k,n}$, where $x$ is an $st$-path if and only if $x_0 = x_k = 0$ (writing $x_i$ instead of $x_{v_i}$ for the coordinates of $x$). 

Almost surely, $\Gamma$ satisfies the following properties: 
\begin{enumerate}[(i)]\setlength{\itemsep}{0pt}
  \item
    $\Gamma$ contains no $st$-path, and
  \item
   all vertices in $\Gamma$ have total degree (in-degree plus out-degree) $\le \log^2 n$.
\end{enumerate}
Both (i) and (ii) follow from simple union bounds. For (i), the number of $st$-paths is $n^{k-1}$, and each $st$-path only has probability $n^{-k}$ of being in $\Gamma$. For (ii), the number of vertices is $kn^2$, and the probability of any given vertex having total degree $\ge \log^2 n$ is $\le \binom{2n}{\log^2 n} n^{-\log^2 n} \le (\frac{2e}{\log^2 n})^{\log^2 n} \le n^{-\omega(1)}$.

For an $st$-path $x$, we will say that $x$ is {\em $\Gamma$-independent} if $\Gamma$ contains no path from $x_i$ to $x_j$ for all $0 \le i < j \le k$.  We claim that, if $\Gamma$ satisfies (i) and (ii), then 99\% of $st$-paths are $\Gamma$-independent. To see this, consider the following greedy procedure for constructing a $\Gamma$-independent $st$-path. Sequentially, for $i = 1,\dots,k-1$, choose any $x_i$ in the $i$th layer of $V_{k,n}$ such that $\Gamma$ contains no path from $s$ to $x_i$ (this kills $\le \log{}^{2i} n$ choices for $x_i$), nor a path from $x_i$ to $t$ (this kills $\le \log{}^{2(k-i)} n$ choices), nor a path from $x_{i'}$ to $x_i$ for any $1 \le i' < i$ (this kills $\le \sum_{i'=1}^{i-1} \log{}^{2(i-i')} n$ choices). Setting $x_0 = s$ and $x_k = t$, note that $x$ is $\Gamma$-independent. In total we get $\ge (n - k^2\log{}^{2k} n)^{k-1} \ge .99n^{k-1}$ distinct $\Gamma$-independent $st$-paths.

Suppose $x$ is a $\Gamma$-independent $st$-path and let $e_1,\dots,e_k$ be the $k$ edges in $x$. We claim that $\Gamma \cup \{e_1,\dots,e_{i-1},e_{i+1},\dots,e_k\}$ contains no $st$-path for all $1 \le i \le k$. To see this, assume for the sake of contradiction that $x'$ is an $st$-path in $\Gamma \cup \{e_1,\dots,e_{i-1},e_{i+1},\dots,e_k\}$. Let $e_1',\dots,e_k'$ be the edges of $x'$. Since $e_i$ is a non-edge of $\Gamma$, we have $e_i \ne e_i'$. Starting at the endpoint of $e_i'$, we can follow the path $x'$ forwards until reaching a vertex in $x$; we can also follow $x'$ backwards from the initial vertex of $e_i'$ until reaching a vertex in $x$. This segment of $x'$ is a path in $\Gamma$ between two vertices of $x$, contradiction $\Gamma$-independence of $x$. 

Since $f_{\mr{out}}$ computes $\PATH(k,n)$, it follows that
\[
f_{\mr{out}}(\Gamma \cup \{e_1,\dots,e_k\}) = 1
\quad\text{and}\quad
f_{\mr{out}}(\Gamma \cup \{e_1,\dots,e_{i-1},e_{i+1},\dots,e_k\}) = 0 \text{ for all }1 \le i \le k.
\] 
This shows that the restricted function $f_{\mr{out}}{\lceil}\rho^\Gamma_{P_k,x}$ depends on all $k$ unrestricted variables (corresponding to the edges of $x$); in fact, $f_{\mr{out}}{\lceil}\rho^\Gamma_{P_k,x}$ is the AND function. Therefore, $x \in \A^\Gamma_{f_{\mr{out}},P_k}$ for every $\Gamma$-independent $st$-path $x$.

By this argument, we have proved:

\begin{claim}[Output of $F$]\label{claim:at-output}
$\ds\lim_{n \to \infty} \Pr_{\Gamma \in \{0,1\}^N_{1/n}} [\,\delta(\A^\Gamma_{f_{\mr{out}},P_k}) \ge .99 n^{-2} \,] = 1$.
\end{claim}

\subsection{Reduction to Pathset Complexity}\label{sec:the-reduction}

We now present the two main lemmas in the reduction from formula size to pathset complexity. Lemma \ref{la:main-rest}, below, is the main technical lemma (the proof, which relies in part on the switching lemma, is given in \S\ref{sec:smallness}). This lemma is the only place in the overall proof of Theorem \ref{thm:pathset-lb} which depends on the assumption that $F$ has bounded depth (though not on the fact that $F$ is a formula as opposed to a circuit).

\begin{la}[Pathsets $\A^\Gamma_{f,G}$ are Small]\label{la:main-rest}
Suppose $f : \{0,1\}^N \to \{0,1\}$ is computed by a circuit of depth $\le \thedepth{}$ and size $\le n^k$. Then, for all $G \in \wp_k$,
\[
  \Pr_{\Gamma \in \{0,1\}^N_{1/n}}
  [\,
    \A^\Gamma_{f,G} 
    \text{ is not $G$-small}
  \,] 
  \le O(n^{-2k}).
\]
\end{la}

Lemma \ref{la:halcali}, below, is the nexus between formula size and pathset complexity. The proof involves a novel top-down argument, which is key to distinguishing formulas and circuits. (Though we will apply Lemma \ref{la:halcali} to the formula $F$ which we have been considering so far, Lemma \ref{la:halcali} is stated in general terms for arbitrary boolean functions with fan-in $2$.)

\begin{la}[``Top-Down Lemma'']
\label{la:halcali}
Let $F$ be any fan-in $2$ formula and let $\Gamma \in \{0,1\}^N$. If $\A^\Gamma_{f,G} \in \smallP_G$ for all $f \in F$ and $G \in \wp_k$, then
\[
  \cost_{P_k}(\A^\Gamma_{f_{\mr{out}},P_k}) 
  \le 2^{O(k^2)} \cdot \depth(F)^k \cdot \size(F).
\]
\end{la}

\begin{proof}
Assume $\A^\Gamma_{f,G} \in \smallP_G$ for all $f \in F$ and $G \in \wp_k$. Consider any $f \in F_{\mr{gates}}$ with children $f_1$ and $f_2$. By Claim \ref{claim:at-gates}, together with the key properties (monotonicity), (sub-additivity) and (join rule) of pathset complexity, we have
\begin{align*}
  \cost_G(\A^\Gamma_{f,G}) 
  &\le
  \cost_G\Big(
  \AA{\Gamma}{G}{f_1} \cup \AA{\Gamma}{G}{f_2} \cup
  \bigcup_{\substack{G_1,G_2 \subset G \,:\, G_1 \cup G_2 = G}}
  \AA{\Gamma}{G_1}{f_1} \bowtie \AA{\Gamma}{G_2}{f_2}
  \Big)
  \\
  &\le 
  \cost_G(\AA{\Gamma}{G}{f_1}) + \cost_G(\AA{\Gamma}{G}{f_2})
  + \sum_{\substack{G_1,G_2 \subset G \,:\, G_1 \cup G_2 = G}}
  \Big(
  \cost_{G_1}(\AA{\Gamma}{G_1}{f_1}) 
  + \cost_{G_2}(\AA{\Gamma}{G_2}{f_2})
  \Big)
  \\
  &\le
  \Big(
  \cost_G(\AA{\Gamma}{G}{f_1}) 
  + 2^k\sum_{H \subset G} \cost_H(\AA{\Gamma}{H}{f_1})
  \Big)
  +
  \Big(
  \cost_G(\AA{\Gamma}{G}{f_2})
  + 2^k\sum_{H \subset G} \cost_H(\AA{\Gamma}{H}{f_2})
  \Big).
\end{align*}

If we start from $\cost_{P_k}(\A^\Gamma_{f_{\mr{out}},P_k})$ and repeatedly apply the above inequality until reaching the inputs of $F$, we get a bound of the form
\[
  \cost_{P_k}(\A^\Gamma_{f_{\mr{out}},P_k})
  \le
  \sum_{\substack{f \in F_{\mr{in}},\, G \in \wp_k}} 
  c_{f,G} \cdot \cost_G(\A^\Gamma_{f,G})
\]
for some $c_{f,G} \in \Z_{\ge 0}$. We claim that 
\[
  c_{f,G}
  \le \sum_{i,H_0,\dots,H_i \,:\, 
  P_k = H_0 \supset \dots \supset H_i = G} 2^{ik} \cdot \binom{\text{depth of $f$ in $F$}}{i}
  \le 2^{O(k^2)} \cdot \depth(F)^k.
\]
To see this, consider any $f \in F_{\mr{in}}$ and $G \in \wp_k$ and let $f_{\mr{out}} = f_0,\dots,f_d = f$ be the branch in $F$ from the output gate down to $f$. Then in the expansion of $\cost_{P_k}(\A^\Gamma_{f_{\mr{out}},P_k})$, we get a contribution of $2^{ik}$ ($\le 2^{k^2}$) from each sequence $(i,t_0,H_0,t_1,H_1,\dots,t_i,H_i)$ where $0 = t_0 < \dots < t_i = d$ and $P_k = H_0 \supset \dots \supset H_i = G$; here $t_i$ is the location where the expansion of $\cost_{P_k}(\A^\Gamma_{f_{\mr{out}},P_k})$ branches as we move from $\cost_{H_{i-1}}(\A^\Gamma_{f_{i-1},H_{i-1}})$ to $2^k\cost_{H_i}(\A^\Gamma_{f_i,H_i})$. Finally, we bound the number of $(t_0,\dots,t_i)$ by $\binom{d}{i}$ ($\le \depth(F)^k$) and the number of $(H_0,\dots,H_i)$ by $2^{ik}$ ($\le 2^{k^2}$). Summing over $i$ adds only a factor of $k$, so in total we get $c_{f,G} \le 2^{O(k^2)} \cdot \depth(F)^k$.

We now use the fact that $\sum_{G \in \wp_k} \cost_G(\A^\Gamma_{f,G}) = 1$ for all $f \in F_{\mr{in}}$ (Claim \ref{claim:at-inputs}) and $\size(F) = |F_{\mr{in}}|$ (since $F$ is a formula!). Concluding the proof, we have
\[
  \cost_{P_k}(\A^\Gamma_{f_{\mr{out}},P_k}) 
  \le 2^{O(k^2)} \cdot \depth(F)^k \cdot \size(F).\qedhere
\]
\end{proof}

We conclude this section by giving the proof of Theorem \ref{thm:main} assuming our pathset complexity lower bound (Theorem \ref{thm:pathset-lb}) and main technical lemma (Lemma \ref{la:main-rest}).

\begin{redu}\label{redu:1}
\normalfont
Theorem \ref{thm:pathset-lb} and Lemma \ref{la:main-rest} 
$\Longrightarrow$
Theorem \ref{thm:main}.
\end{redu}

\begin{proof}
Assuming Theorem \ref{thm:pathset-lb} and Lemma \ref{la:main-rest}, we must show that $\size(F) \ge n^{\Omega(\log k)}$. By Claim \ref{claim:at-output} and Lemma \ref{la:main-rest}, there exists $\Gamma \in \{0,1\}^N$ such that $\delta(\A^\Gamma_{f_{\mr{out}},P_k}) \ge .99 n^{-2}$ and $\A^\Gamma_{f,G} \in \smallP_G$ for all $f \in F$ and $G \in \wp_k$. Fix any such $\Gamma$. We now have
\begin{align*}
  &&\size(F)
  &\ge
    \frac{1}{2^{O(k^2)} \cdot \depth(f)^k}
    \cdot 
    \cost_{P_k}(\A^\Gamma_{f_{\mr{out}},P_k})
  &&\text{(Lemma \ref{la:halcali})}&&\\
  \vphantom{\bigg|}&&&\ge
    \frac{1}{2^{O(k^2)} \cdot \depth(f)^k}
    \cdot
    \frac{n^{(1/6)\log k}}{2^{O(2^k)}}
    \cdot 
    \delta(\A^\Gamma_{f_{\mr{out}},P_k})
  &&\text{(Theorem \ref{thm:pathset-lb})}.
\end{align*}
Using inequalities
\[
  \depth(F) \le \log^2 n,\quad\
  \delta(\A^\Gamma_{f_{\mr{out}},P_k}) \ge .99n^{-2},\quad\
  k \le \log\log n,
\]
we get the desired bound $\size(F) \ge n^{(1/6)\log k - O(1)}$.
\end{proof}

\section{Small Pathsets from Random Restrictions}
\label{sec:smallness}

In this section, we prove Lemma \ref{la:main-rest} showing that, with high probability over random $\Gamma \in \{0,1\}^N_{1/n}$, pathsets $\A^\Gamma_{f,G}$ are small for all $f \in F$ and $G \in \wp_k$. The proof has the following scheme:
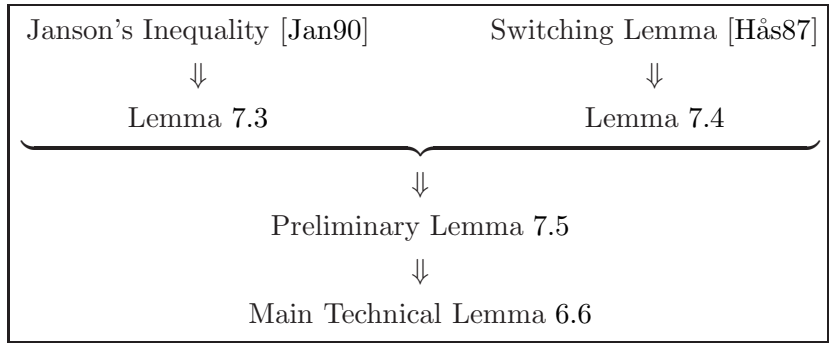
\[
  \fbox{$
  \begin{gathered}
  \underbrace{\begin{gathered}
      \text{Janson's Inequality \cite{janson1990poisson}}\\
      \Downarrow\\
      \text{Lemma \ref{la:critical}}
    \end{gathered}
    \qquad\qquad
    \begin{gathered}
      \text{Switching Lemma \cite{hastad1987computational}}\\
      \Downarrow\\
      \text{Lemma \ref{la:from-switching}}
    \end{gathered}}\\
    \Downarrow\\
    \text{Preliminary Lemma \ref{la:pre-main-rest}}\\
    \Downarrow\\
    \text{Main Technical Lemma \ref{la:main-rest}}
  \end{gathered}
  $}
\]
The central argument is contained in the proof of Preliminary Lemma \ref{la:pre-main-rest} (from which Lemma \ref{la:main-rest} essentially follows as a corollary). In the interest of presenting this central argument first, the proofs of Lemmas \ref{la:critical} and \ref{la:from-switching} are given afterwards in \S\ref{sec:critical} and \S\ref{sec:from-switching}.

\begin{rmk}
Lemma \ref{la:main-rest} is similar to the main technical lemma in the $k$-clique lower bound of \cite{rossman2008constant,rossman2010average}. One important difference is that here we require a concentration of measure inequality in a place where a mere expectation bound sufficed for the $k$-clique result.\footnote{Consider the fact that functions $f : \{0,1\}^n \to \{0,1\}$ in $\AC^0$ have low average sensitivity. This is a statement of the form $\Pr_{x \in \{0,1\}^n,\,i \in [n]}[i \in S(f,x)] \le \eps$ where $S(f,x)$ is the set $\{i \in [n] : f(x) \ne f(x \oplus i)\}$. The main technical lemma of \cite{rossman2008constant} is an inequality of the same form; the difference is that $S(f,x)$, rather than the set of sensitive coordinates, is instead the set of {\em sensitive $G$-shaped sets of coordinates} where $G$ is a pattern graph (which, in the context of $k$-clique, means a subgraph of $K_k$). By contrast, Lemma \ref{la:main-rest} is a concentration of measure inequality analogous to showing $\Pr_x[\Pr_i[i \in S(f,x)] > \eps] \le \delta$.} This makes the proof of Lemma \ref{la:main-rest} somewhat more complicated.
\end{rmk}

Recall that $G$-smallness consists of $2^{\c{G}}-1$ density constraints corresponding to the nonempty unions of components of $G$. We say that non-small pathset $\A$ is {\em $G$-critical} if it violates only the ``top'' constraint $\delta(\A) > \m^{-\c{G}}$. Formally:

\begin{df}[Critical Pathsets]
For a pattern graph $G$ and pathset $\A \in \P_G$, we say that $\A$ is {\em $G$-critical} if $\delta(\A) > \m^{-\c{G}}$ and $\mu_{V_G \setminus S}(\A) \le \m^{s-\c{G}}$ for all $1 \le s < \c{G}$ and $S \subseteq V_G$ such that $S$ intersects exactly $s$ components of $G$.
\end{df}

Our first lemma, proved in \S\ref{sec:critical}, gives a concentration of measure inequality for critical pathsets. (Recall that $\eps = 1/\log k$ from Def.\ \ref{df:small}(i).)

\begin{la}\label{la:critical}
Let $G$ be a pattern graph and suppose $\A$ is a $G$-critical pathset. Let $q = (1/n)^{1 + (\eps/2k)}$. Then
\[
  \Pr_{\X \subseteq_q N}
  \Big[\,
    \#\big\{x \in \A : N_{G,x} \subseteq \X\big\}
    \le 
    \frac{n^{\eps/2}}{2}
  \,\Big]
  \le 
  \exp\Big({-}\Omega\Big(
  \frac{n^{\eps/2}}{2^k}
  \Big)\Big).
\]
\end{la}

Our second lemma, proved in \S\ref{sec:from-switching}, is a straightforward corollary of H{\aa}stad's Switching Lemma \cite{hastad1987computational}.

\begin{la}\label{la:from-switching}
Suppose $f$ is a boolean function\footnote{This statement is independent of the number of variables of $f$.} computed by a circuit of size $s$ and depth $d$. For all $0 < q \le p \le 1/2$ and $\tau,r > 0$, 
\[
  d \le
  \frac{\log(p/q)}{r^{-1}\log(s/\tau)+\log(5r)}
  \ \Longrightarrow\ 
  \Pr_{\theta \in \mc R(p,q)} 
  \big[\,
  |\Live(f{\lceil}\theta)| > 2^r
  \,\big] 
  \le \tau.
\]
\end{la}

We now give the core argument in the proof of Lemma \ref{la:main-rest}.

\begin{la}\label{la:pre-main-rest}
Suppose $f : \{0,1\}^N \to \{0,1\}$ is computed by a circuit of size $\le n^k$ and depth $\le \thedepth{}$. Let $G,H$ be pattern graphs with $V_G \cap V_H = \emptyset$ and let $y \in [n]^{V_H}$. For $\Gamma \in \{0,1\}^N$, define $G$-pathset $\A^\Gamma$ by
\[
  \A^\Gamma \defeq
  \big\{x \in [n]^{V_G} : 
  N_{G,x} \subseteq \Live(f{\lceil}\rho^\Gamma_{G \cup H,xy})\big\}.
\]
Then
$\ds
  \Pr_{\Gamma \in \{0,1\}^N_{1/n}}
  [\,\A^\Gamma \text{ is $G$-critical}\,]
  \le O(n^{-10k}).
$
\end{la}

Note that Lemma \ref{la:pre-main-rest} does not rely on the assumption that $k \le \log\log n$. It holds up to $k \le \log^{1/3}n$ (at which point the statement is trivial). We also remark that, while the bound $O(n^{-10k})$ is sufficient for our purposes, we could easily get a stronger bound like $O(n^{-k\log k})$.

\begin{proof}
Define $\mc I \subseteq \{0,1\}^N$ by
\[
  \mc I \defeq \big\{
    I \in \{0,1\}^N : I_\nu = 0 \text{ for all } \nu \in N \setminus N_{H,y}
  \big\}.
\]
Note that $|\mc I| = 2^{|N_{H,y}|} = 2^{|E_H|} \le 2^k$. 

For $I \in \mc I$, define $f_I : \{0,1\}^N \to \{0,1\}$ by $f_I(\Gamma) = f(\Gamma \oplus I)$. For all $x \in [n]^{V_G}$ and $\Gamma \in \{0,1\}^N$, we have
\begin{align}\label{eq:mcI}
  \Live(f{\lceil}\rho^\Gamma_{G \cup H,xy})
  =
  \bigcup_{I \in \mc I} \Live(f_I{\lceil}\rho^\Gamma_{G,x}).
\end{align}

For $\X \subseteq N$, let $\theta^\Gamma_\X : N \to \{0,1,\ast\}$ be the restriction taking value $\ast$ over $\X$ and equal to $\Gamma$ over $N \setminus \X$. Define pathsets $\B_\X$ and $\C^\Gamma_\X$ by
\begin{align*}
  \B_\X &\defeq \big\{x \in [n]^{V_G} : N_{G,x} \subseteq \X\big\},\\
  \C^\Gamma_\X &\defeq \big\{x \in [n]^{V_G} : N_{G,x} \subseteq 
  \ts\bigcup_{I \in \mc I} \Live(f_I\lceil\theta^\Gamma_\X)
  \big\}.
  \vphantom{t^{\big|}}
\end{align*}
It follows from (\ref{eq:mcI}) that $\A^\Gamma \cap \B_\X \subseteq \C^\Gamma_\X$.
Also, since $|N_{G,x}| = |E_G| \le k$ and $|\mc I| \le 2^k$,
\begin{align}\label{eq:1overk}
  \big|\C^\Gamma_\X\big|{}^{1/k} 
  \le \Big|\bigcup_{I \in \mc I} \Live(f_I\lceil\theta^\Gamma_\X)\Big|
  \le 2^k\cdot \ds\max_{I \in \mc I} \big|\Live(f_I\lceil\theta^\Gamma_\X)\big|.
\end{align}

We now consider independent random $\Gamma \in \{0,1\}^N_{1/n}$ and random $\X \subseteq_q N$ where $q = (1/n)^{1+(\eps/2k)}$. Note that $\theta^\Gamma_\X$ has distribution $\mc R(1/n,q)$. Also, note that $\A^\Gamma$ and $\B_\X$ are independent, as $\A^\Gamma$ depends only on $\Gamma$ and $\B_\X$ depends only on $\X$. 

We may assume that $k \le \log^{1/3}n$, since otherwise the lemma is trivial. In particular, $2^k = o(n^{\eps/2})$ (recall that $\eps = 1/\log k$) and hence $\exp(-\Omega(n^{\eps/2}/2^k)) = o(1)$. We have
\begin{align*}
  \Pr_{\Gamma}
  \big[\,
    \A^\Gamma \text{ is $G$-critical}
  \,\big]
  &\le 
  \Pr_{\Gamma}
  \Big[\,
    \Pr_\X\ds
    \Big[\,
      |\A^\Gamma \cap \B_\X| \le \frac{n^{\eps/2}}{2}
    \,\Big]
  \le \exp\Big({-}\Omega\Big(
    \frac{n^{\eps/2}}{2^k}
    \Big)\Big)
  \,\Big]
  &&\text{(Lemma \ref{la:critical})}\\
  &= 
  \Pr_{\Gamma}
  \Big[\,
    \Pr_\X\ds
    \Big[\,
      |\A^\Gamma \cap \B_\X| > \frac{n^{\eps/2}}{2}
    \,\Big]
  \ge 1-o(1)
  \,\Big]
  \\
  &\le
  (1+o(1))
  \Pr_{\Gamma,\X}\ds
    \Big[\,
      |\A^\Gamma \cap \B_\X| > \frac{n^{\eps/2}}{2}
    \,\Big]
  &&\text{(Markov ineq.)}\\
  &\le
  (1+o(1))
  \Pr_{\Gamma,\X}\ds
    \Big[\,
      |\C^\Gamma_\X| > \frac{n^{\eps/2}}{2}
    \,\Big]
  &&\text{($\A^\Gamma \cap \B_\X \subseteq \C^\Gamma_\X$)}\\
  &\le
  (1+o(1))
  \Pr_{\Gamma,\X}\ds
    \Big[\,
      \max_{I \in \mc I} \big|\Live(f_I\lceil\theta^\Gamma_\X)\big| 
      > \frac{n^{\eps/2k}}{2^{k+1}}
    \,\Big]
  &&\text{(by (\ref{eq:1overk}))}\\
  &\le
  (1+o(1)) 
  \sum_{I \in \mc I}
  \Pr_{\Gamma,\X}\ds
    \Big[\,
      \big|\Live(f_I\lceil\theta^\Gamma_\X)\big| 
      > \frac{n^{\eps/2k}}{2^{k+1}}
    \,\Big].
\end{align*}

Using $|\mc I| \le 2^k$ and the fact that $\theta^\Gamma_\X$ has distribution $\mc R(1/n,q)$, it suffices to show that for every $I \in \mc I$ and sufficiently large $n$,
\begin{align}\label{eq:using-from}
  \Pr_{\theta \in \mc R(1/n,q)}\ds
  \Big[\,
      \big|\Live(f_I\lceil\theta)\big| 
      > \frac{n^{\eps/2k}}{2^{k+1}}
  \,\Big]
  \le
  n^{-11k}.
\end{align}
In order to apply Lemma \ref{la:from-switching}, let
\[
  p = n^{-1},\quad\
  s = n^k,\quad\
  d = \frac{\log n}{k^3\log\log n},\quad\
  \tau = n^{-11k},\quad\
  r = \log\Big(\frac{n^{\eps/2k}}{2^{k+1}}\Big).
\]
Recall that $q = (1/n)^{1+(\eps/2k)}$ and $k \le \log^{1/3} n$. We have $\log(p/q) = (\eps/2k)\log n$ and $\log(5r) = O(\log\log n)$ and
$r^{-1}\log(s/\tau)$
  $=$ 
  $((\eps/2k)\log n - k - 1)^{-1} 12k\log n$
  $=$ $O(k^2/\eps)$.
Since $\eps = 1/\log k$, we have
\[
  \frac{\log(p/q)}{r^{-1}\log(s/\tau)+\log(5r)}
  \ge
  \Omega\Big(
    \frac{\log n}{(k^3/\eps^2) + (k/\eps)\log\log n}
  \Big)
  \ge \omega(d).
\]
Note that $f_I$ has circuits of the same size ($\le s$) and depth ($\le d$) as $f$, since the operation $\Gamma \mapsto \Gamma \oplus I$ simply exchanges the positive and negative literals for variables corresponding to coordinates of $I$ with value $1$. Therefore, for sufficiently large $n$, Lemma \ref{la:from-switching} implies (\ref{eq:using-from}). This completes the proof.
\end{proof}

Finally, we derive Lemma \ref{la:main-rest} from  Lemma \ref{la:pre-main-rest}.

\begin{proof}[Proof of Lemma \ref{la:main-rest}]
Suppose $f : \{0,1\}^N \to \{0,1\}$ is computed by circuits of size $\le n^k$ and depth $\le \log n/k^3\log\log n$. Fix a pattern graph $G$. We must show
\[
  \Pr_{\Gamma \in \{0,1\}^N_{1/n}}
  [\,
    \A^\Gamma_{f,G} \text{ is not $G$-small} 
  \,] 
  \le O(n^{-2k}).
\]

Suppose $\Gamma \in \{0,1\}^N$ is any layered graph such that $\A^\Gamma_{f,G}$ is not $G$-small. We claim that there exist $S \subseteq V_G$ and $z \in [n]^{V_G \setminus S}$ such that $S$ is a nonempty union of components of $G$ and the pathset $B^\Gamma_{S,z}$ is $G|_S$-critical where $G|_S$ is the induced subgraph of $G$ on $S$ and
\[
  \B^\Gamma_{S,z} \defeq \big\{y \in [n]^S : 
  N_{G|_S,y} \subseteq \Live(f{\lceil}\rho^\Gamma_{G,yz})\big\}.
\]
We first note that it suffices to prove this claim. Since there are $2^{\c{G}}-1$ ($\le 2^k$) choices for $S$ and $\le n^k$ choices for $z$, assuming the claim we have
\begin{align*}
  \Pr_{\Gamma \in \{0,1\}^N_{1/n}}
  [\,
    \A^\Gamma_{f,G} \text{ is not $G$-small}
  \,] 
  &\le
  \Pr_{\Gamma \in \{0,1\}^N_{1/n}}
  \Big[\, \bigvee_{S,z}
    \B^\Gamma_{S,z} \text{ is $G|_S$-critical}
  \,\Big]
  \\
  &\le
  2^k n^k O(n^{-10k})
  \qquad\text{(by Lemma \ref{la:pre-main-rest})}\\
  &\le O(n^{-2k}).
\end{align*}

To see why the claim holds, assume that $\A^\Gamma_{f,G}$ is not $G$-small and consider the following procedure. Initially set $S \leftarrow V_G$ and $z \leftarrow ()$ (the empty tuple). If $\B^\Gamma_{S,z}$ is $G|_S$-critical, then we are done. Otherwise, since $\B^\Gamma_{S,z}$ is neither $G|_S$-small nor $G|_S$-critical, there is a proper subset $T \subset S$ such that $T$ is a union of $t \ge 1$ components of $V_{G|_S}$ and $\mu_T(\B^\Gamma_{S,z}) > \m^{-t}$. By definition of $\mu_T$, there exists $y \in [n]^{S \setminus T}$ such that $\delta(\RHO{T}{\B^\Gamma_{S,z}}{y}) > \m^{-t}$. Note that $T$ is a union of components of $G$ and $yz \in [n]^{V_G \setminus T}$. Also, for all $u \in [n]^T$, we have $N_{G|_T,u} \subseteq N_{G|_S,uy}$ and hence
\begin{align*}
  u \in \RHO{T}{\B^\Gamma_{S,z}}{y} 
  \,&\Longrightarrow\, uy \in \B^\Gamma_{S,z}\\
  \,&\Longrightarrow\, N_{G|_S,uy} \subseteq \Live(f{\lceil}\rho^\Gamma_{G,uyz})\\
  \,&\Longrightarrow\, N_{G|_T,u} \subseteq \Live(f{\lceil}\rho^\Gamma_{G,uyz})\\
  \,&\Longrightarrow\, u \in \B^\Gamma_{T,yz}.
\end{align*}
Therefore, $\RHO{T}{\B^\Gamma_{S,z}}{y} \subseteq \B^\Gamma_{T,yz}$. It follows that $\delta(\B^\Gamma_{T,yz}) > \m^{-t}$ and, hence, $\B^\Gamma_{T,yz}$ is not $G|_T$-small. We now update $S \leftarrow T$ and $z \leftarrow yz$. Since $\B^\Gamma_{S,z}$ is not $G|_S$-small, we may repeat this process so long as $\B^\Gamma_{S,z}$ is not $G|_S$-critical. Since $S$ shrinks with every step, eventually this process will terminate, at which point $\B^\Gamma_{S,z}$ is $G|_S$-critical (and $S$ is nonempty by definition of $G|_S$-criticality). Thus, the claim holds and the lemma is proved.
\end{proof}

\subsection{Proof of Lemma \ref*{la:critical}}\label{sec:critical}

For the proof of Lemma \ref{la:critical} we use a concentration of measure inequality due to Janson \cite{janson1990poisson}.

\begin{la}[Janson's Inequality \cite{janson1990poisson}]\label{la:janson-tail}
Let $\Omega$ be a finite universal set and let $R$ be a random subset of $\Omega$ given by $\Pr[\,r \in R\,] = p_r$, these events mutually independent over $r \in \Omega$. Let $\{S_i\}_{i \in I}$ be an indexed family of subsets of $\Omega$. Define $\lambda$ and $\Upsilon$ by
\begin{align*} 
  &&&&\lambda &\defeq \sum_{i \in I} \Pr\big[\,S_i \subseteq R\,\big], 
  &\Upsilon &\defeq \sum_{(i,j) \in I^2 \,:\, i \ne j,\: S_i \cap S_j \ne \emptyset} 
  \Pr\big[\, S_i \cup S_j \subseteq R\,\big].&&&&
\end{align*}
Then, for all $0 \le t \le \lambda$,
$\ds
  \Pr\Big[\,
    \#\big\{i \in I \,:\, S_i \subseteq R\big\} \le \lambda - t
  \,\Big]
  \le
  \exp\Big({-}\frac{t^2}{2(\lambda + \Upsilon)}\Big)$.
\end{la}

\begin{proof}[Proof of Lemma \ref{la:critical}]
Let $G$ be a nonempty pattern graph, let $\A$ be a $G$-critical pathset, and let $q = (1/n)^{1+(\eps/2k)}$. We must show 
\[
  \Pr_{\X \subseteq_q N}
  \Big[\,
    \#\big\{x \in \A : N_{G,x} \subseteq \X\big\}
    \le 
    \frac{n^{\eps/2}}{2}
  \,\Big]
  \le 
  \exp\Big({-}\Omega\Big(
  \frac{n^{\eps/2}}{2^k}
  \Big)\Big).
\]
As in Janson's Inequality, define $\lambda$ and $\Upsilon$ by
\begin{align*} 
  &&\lambda &\defeq \sum_{x \in \A} \Pr\big[\,N_{G,x} \subseteq \X\,\big], 
  &\Upsilon &\defeq \sum_{(x,y) \in \A^2 \,:\, x \ne y,\, N_{G,x} \cap N_{G,y} \ne \emptyset} 
  \Pr\big[\,N_{G,x} \cup N_{G,y} \subseteq \X\,\big].&&
\end{align*}
Taking $t = \lambda/2$ in Lemma \ref{la:janson-tail}, we get
\begin{equation}\label{eq:janson}
  \Pr\Big[\,
    \#\big\{
    x \in \A \,:\, N_{G,x} \subseteq \X
    \big\}
    \le 
    \frac{\lambda}{2}
  \,\Big]
  \le
  \exp\Big({-}\frac{1}{16}
  \min\Big\{\lambda,\,
  \frac{\lambda^2}{\Upsilon}
  \Big\}\Big).
\end{equation}

Recall that $\c{G} = |V_G| - |E_G|$ and $\m = n^{1-\eps}$. By $G$-criticality of $\A$,
\[
  |\A| 
  = n^{|V_G|} \delta(\A) 
  > n^{|V_G|}  \m^{-\c{G}} 
  = n^{|E_G| + \eps \c{G}}.
  \vphantom{\Big|}
\]
Note that $\Pr[\,N_{G,x} \subseteq \X\,] = q^{|E_G|}$ for all $x \in \A$. Since $|E_G| \le k$ and $\c{G} \ge 1$, it follows that
\begin{align}\label{eq:lambda}
  \lambda 
  &= |\A| \cdot  q^{|E_G|}
  > n^{\eps(\c{G} - |E_G|/2k)}
  \ge n^{\eps(\c{G} - (1/2))}
  \ge n^{\eps/2}.
  \vphantom{\Big|}
\end{align}
To complete the proof, it suffices to show that $\ds\frac{\lambda^2}{\Upsilon} \ge \frac{n^{\eps/2}}{2^k}$.

For all $(x,y) \in \A^2$, let
\[
  \Overlap{x,y} \defeq \big\{\edge{v}{w} \in E_G : x_v=y_v \text{ and } x_w=y_w\big\}.
\]
Note that $x = y$ iff $\Overlap{x,y} = E_G$, and $N_{G,x} \cap N_{G,y} \ne \emptyset$ iff $\Overlap{x,y} \ne \emptyset$, and $|N_{G,x} \cup N_{G,y}| = 2|E_G| - |\Overlap{x,y}|$. Next, note that $\Upsilon = \sum_{T \,:\, \emptyset \subset T \subset E_G} \Upsilon_T$ where
\begin{align*}
  \Upsilon_T 
  &\defeq \sum_{(x,y) \in \A^2 \,:\, \Overlap{x,y} = T} 
     \Pr\big[\, N_{G,x} \cup N_{G,y} \subseteq \X \,\big]
  =  \#\big\{(x,y) \in \A^2 : \Overlap{x,y} = T\big\} 
     \cdot q^{2|E_G|-|T|}.
\end{align*}

Now consider any fixed $\emptyset \subset T \subset E_G$. Let $S = \bigcup_{\edge{v}{w} \in T} \{v,w\}$ and let $s$ be the number of components of $G$ which $S$ intersects. Note that $1 \le s \le |S|-|T|$, since $|S|-|T|$ equals the number of components in the induced subgraph $G|_S$. We have
\begin{align*}
  \#\big\{(x,y) \in \A^2 : \Overlap{x,y} = T\big\} 
  &= \sum_{z \in [n]^S} \big|\RHO{V_G \setminus S}{\A}{z}\big|^2
     \\
  &\le
     \Big(\sum_{z \in [n]^S} \big|\RHO{V_G \setminus S}{\A}{z}\big|\Big)
     \Big(\max_{z \in [n]^S} \big|\RHO{V_G \setminus S}{\A}{z}\big|\Big)
     \\
  &= |\A| \cdot  n^{|V_G|-|S|} \cdot \mu_{V_G \setminus S}(\A)\\
  &\le |\A| \cdot  n^{|V_G|-|S|} \cdot \m^{s-\c{G}}
  \qquad\quad\text{(by $G$-criticality of $\A$)}\\
  &= |\A| \cdot  n^{|E_G| - |S| + s + \eps(\c{G}-s)}.
\end{align*}
It follows that
\begin{align}
\notag
  \Upsilon_T 
  &= \#\big\{(x,y) \in \A^2 : \Overlap{x,y} = T\big\}
     \cdot q^{2|E_G|-|T|}
     \\
\notag
  &\le  
     |\A| \cdot n^{|E_G| - |S| + s + \eps(\c{G}-s)} 
     \cdot q^{2|E_G|-|T|}
     \\
\notag
  &\le
     \lambda\cdot n^{|T|-|S|+s+\eps(\c{G}-s)}
  &&\text{(using $\lambda = |\A|\cdot q^{|E_G|}$ and $q \le n^{-1}$)}
     \\
\notag
  &\le \lambda\cdot n^{\eps(\c{G}-1)}
  &&\text{(using $1 \le s \le |S|-|T|$).}
\intertext{We now have}
\label{eq:lambda-Upsilon}
  \frac{\lambda^2}{\Upsilon}
  &= \frac{\lambda^2}{\sum_{T \,:\, \emptyset \subset T \subset E_G} \Upsilon_T}
  \ge \frac{\lambda \cdot n^{-\eps(\c{G}-1)}}{2^k}
  &&\text{(by the above)}\\
  \notag
  &\phantom{= \smash{\frac{\lambda^2}{\sum_{T \,:\, \emptyset \subset T \subset E_G} \Upsilon_T}}\hspace{3pt}}
  \ge \frac{n^{\eps/2}}{2^k}
  &&\text{(since $\lambda \ge n^{\eps(\c{G}-(1/2))}$ by (\ref{eq:lambda})).}
\end{align}
Plugging (\ref{eq:lambda}) and (\ref{eq:lambda-Upsilon}) into (\ref{eq:janson}) completes the proof.
\end{proof}

\subsection{Proof of Lemma \ref*{la:from-switching}}\label{sec:from-switching}

For a boolean function $f$, we write $\DTD(f)$ for the {\em decision-tree depth} of $f$ (i.e.\ the minimum depth of a decision tree computing $f$). Note that $|\Live(f)| \le 2^{\DTD(f)}$.

The following lemma is a special case of the original Switching Lemma of H{\aa}stad \cite{hastad1987computational}. (For simplicity, we consider depth-$r$ decision trees as opposed to $r$-DNFs and $s$-CNFs.)

\begin{la}[Switching Lemma \cite{hastad1987computational}]\label{la:switching-lemma}
Suppose $f$ is a boolean function which is an AND or OR of (arbitrary many) depth-$r$ decision trees. Then for all $q \in [0,1/2]$, 
$$\ds
	\Pr_{\theta \in \mc R(1/2,q)} [\, \DTD(f{\lceil}\theta) > r\,] \le (5qr)^r.
$$
\end{la}

Lemma \ref{la:from-switching} follows directly from the following lemma via the fact that $|\Live(f)| \le 2^{\DTD(f)}$. (This lemma originally appeared in the author's Ph.D.\ thesis \cite{rossman2010average}.)

\begin{la}\label{la:pre-from-switching}
Suppose $f$ is a boolean function computed by a circuit of size $s$ and depth $d$. 
For all $0 < q \le p \le 1/2$ and $\tau,r > 0$, 
\[
  d \le 
  \frac{\log(p/q)}{r^{-1}\log(s/\tau)+\log(5r)}
  \quad\Longrightarrow\quad
  \Pr_{\theta \in \mc R(p,q)} 
  \big[\,
    \DTD(f{\lceil}\theta) > r
  \,\big] 
  \le \tau.
\]
\end{la}

This lemma originally appeared in the author's Ph.D.\ thesis \cite{rossman2010average}. The proof is included here for completeness.

\begin{proof}
We generate a sequence $\theta_0,\dots,\theta_d$ of random restrictions as follows:
\begin{itemize}\setlength{\itemsep}{0pt}
  \item
    let $\theta_0 \in \mc R(p_0,p)$ where 
    $p_0 = \ds
    \frac{1}{1-p}
    \Big(\frac{p+q}{2} - pq\Big)$
    (note that $0 < p_0 \le 1$),
  \item
    for $i \in \{1,\dots,d\}$, let 
    $\theta_i \in \mc R(1/2, (q/p)^{1/d})$ 
    applied to the variables left unrestricted by $\theta_0,\dots,\theta_{i-1}$.
\end{itemize}
For $i \in \{0,\dots,d\}$, let $\Theta_i$ denote the composition of restrictions $\theta_0,\dots,\theta_i$. Note that $\Theta_d$ has distribution $\mc R(p,q)$:
\begin{align*}
   \Pr\big[\,\Theta_d = \ast\,\big]
	&= \Pr\big[\,\theta_0 = \dots = \theta_d = \ast\,\big]
	= p\Big(\Big(\frac{q}{p}\Big)^{1/d}\Big)^d
	= q,\\
  \Pr\big[\,\Theta_d = 1\,\big]
  &= \Pr\big[\,\theta_0 = 1\,\big] + 
  \sum_{i=1}^d\Pr\big[\,\theta_0 = \cdots = \theta_{i-1} = \ast \text{ and } \theta_i = 1\,\big]\\
  &= p_0(1-p) + \sum_{i=1}^d 
  \frac{p}{2}
  \Big(\frac{q}{p}\Big)^{(i-1)/d}
  \Big(1-\Big(\frac{q}{p}\Big)^{1/d}\Big)\\
  &= p_0(1-p) + \frac{p-q}{2}\\ 
  &= (1-q)p.
\end{align*}
Therefore, $\Theta_d$ has distribution $\mc R(p,q)$.

Let $C$ be the circuit of size $s$ and depth $d$ computing $f$. For each input/gate $g \in C$ at height $i$ from the bottom (where inputs have height $0$ and the output gate $g_{\mr{out}}$ has height $d$), let $X_g$ denote the event that $\DTD(g\lceil\Theta_i) \le r$. Let $C_{<g}$ denote the set of $g' \in C$ such that $g'$ lies below $g$.

If $g$ has height $0$, then $\Pr[\,X_g\,] = 1$. If $g$ has height $i \ge 1$, then
\begin{align*}
  \Pr\Big[\,\neg{X_g}\ \Big|\ \bigwedge_{g' \in C_{<g}} X_{g'} \,\Big]
    &= \Pr\Big[\,\DTD((g\lceil\Theta_{i-1}) \lceil \theta_i) > r\
       \Big|\ 
       \bigwedge_{g' \in C_{<g}} \DTD(g'\lceil\Theta_{i-1}) \le r \,\Big]
       \hspace{-2.05in}\\
	&\le \Big(5r \Big(\frac{q}{p}\Big)^{1/d} \Big)^r 
	&&\text{(Lemma \ref{la:switching-lemma})}\\
	&\le \frac{\tau}{s}
	&&\text{($d \le \frac{\log(p/q)}{\log(5r(s/\tau)^{1/r})}$).}
\end{align*}
Completing the proof, we have
\begin{align*}
  \Pr_{\theta \in \mc R(p,q)} \Big[\,\DTD(f\lceil\theta) > r\,\Big]
    &= \Pr\Big[\,\neg X_{g_{\mr{out}}}\,\Big]
    \le \Pr\Big[\,\bigvee_{g \in C} \neg{X_g}\,\Big]
    \le \sum_{g \in C}
      \Pr\Big[\,\neg X_g\ \Big|\ \bigwedge_{g' \in C_{<g}} X_g\,\Big]
	\le \tau.\qedhere
\end{align*}
\end{proof}

\newpage
\section{Patterns}\label{sec:patterns}

At this point in the paper, it only remains to prove Theorem \ref{thm:pathset-lb}, our lower bound on pathset complexity $\cost$. As the first step in the proof, we introduce the notion of {\em patterns} and {\em pathset complexity w.r.t.\ patterns}, denoted $\pcost$. Intuitively, a pattern is a blueprint for constructing a pattern graph via pairwise unions starting from individual edges. This leads to a more constrained notion of pathset complexity where the allowable joins are prescribed by a given pattern. 

Fixing the pattern of allowable joins can only increase the cost of constructing a pathset, hence $\cost \le \pcost$. Counterintuitively, the lower bound on $\cost$ is derived from a lower bound on $\pcost$. This lower bound on $\pcost$ is the true combinatorial lower bound in this paper. (Unfortunately, in the shift from $\cost$ to $\pcost$ we lose a factor of $2^{O(2^k)}$, which is the reason that Theorem \ref{thm:pathset-lb} only holds up to $k(n) \le \log\log n$.)

In this section, we present the definition of $\pcost$ and state our lower bound for $\pcost$ (Theorem \ref{thm:pattern-lb}). The reduction from Theorem \ref{thm:pattern-lb} to Theorem \ref{thm:pathset-lb} is given in \S\ref{sec:relationship}. In \S\ref{sec:proj-rest} we prove some preliminary lemmas (on properties of $\pcost$ with respect to projection and restriction). Finally we prove Theorem \ref{thm:pattern-lb} in \S\ref{sec:lower-bound}.

\begin{df}[Patterns]
A {\em pattern} is a (rooted, unordered) binary tree whose leaves are labeled by edges of $P_k$ (i.e.\ elements of $E_k = \{\edge{v_i}{v_{i+1}} : 0 \le i < k\}$). Every pattern $A$ is associated with a pattern graph denoted $G_A = (V_A,E_A)$ where $E_A$ is the set of edges of $P_k$ which label leaves in $A$. 

The {\em empty pattern} (of size $0$) is denoted $\emptyset$. Patterns of size $1$ (corresponding to elements of $E_k$) are said to be {\em atomic}. Patterns of size $\ge 2$ are {\em non-atomic}. Throughout, $A$ and $B$ represent non-empty patterns. Let $\AB$ ($=\BA$) denote the pattern with children $A$ and $B$. Note that every non-atomic pattern has the form $\AB$ for some $A$ and $B$; also, $G_{\AB} = G_A \cup G_B$.

For a pattern $A$, {\em sub-patterns} of $A$ are sub-trees of $A$ consisting a node in $A$ and all nodes below that node with the inherited labeling of leaves. The sub-pattern and strict sub-pattern relations are denoted by $\preceq$ and $\prec$ respectively.

To simplify notation, for a pattern $A$ we write $\P_A$ for $\P_{G_A}$ and $\proj_A$ for $\proj_{V_A}$ and $\l{A}$ for $\l{G_A}$, etc.\ We consistently write $\A,\B,\C$ for pathsets with underlying patterns $A,B,C$ respectively.
\end{df}

\begin{df}[Pathset Complexity w.r.t.\ Patterns]\label{df:pc2}
For every pattern $A$ and pathset $\A \in \P_A$, the {\em pathset complexity} of $\A$ with respect to $A$, denoted $\pcost_A(\A)$, is defined by the following induction:
\begin{enumerate}[(i)]
  \item
    $\pcost_\emptyset(\{()\}) \defeq 0$, that is, the pathset complexity of $\{()\}$ w.r.t.\ the empty pattern $\emptyset$ is $0$.
  \item
    If $A$ is atomic and $|\A| = 1$, then $\pcost_A(\A) \defeq 1$.
  \item
    For non-atomic $A = \{B,C\}$, 
    \[
      \pcost_A(\A) \defeq \min_{(\B_i,\C_i)_i} 
      \sum_i \pcost_B(\B_i) + \pcost_C(\C_i)
    \]
    where $(\B_i,\C_i)_i$ ranges over sequences such that $\B_i \in \smallP_B$, $\C_i \in \smallP_C$ and $\A \subseteq \bigcup_i \B_i \bowtie \C_i$.
\end{enumerate}
\end{df}

In Appendices \ref{sec:key-examples}--\ref{sec:rectangular} we present some key examples of patterns and prove upper and lower bounds for $\pcost$ with respect to some special classes of patterns. The material in these appendices is not directly needed for our main results. However, these appendices serve as a warm-up and motivation for the lower bound that follows.

The following inequalities (analogous to the inequalities following Definition \ref{df:pc1} of $\cost$) are essentially built into Definition \ref{df:pc2} of $\pcost$:
\begin{align}
\tag{base case}
\vphantom{\big|}
  \pcost_\emptyset(\{()\}) \le 0
  \ \hspace{1pt}&\hspace{-1pt}\text{and}\ 
  \pcost_A(\A) \le 1
  &&\hspace{-20pt}\text{if $A$ is atomic and }|\A|=1,\\
\tag{monotonicity}
\vphantom{\Big|}
  \pcost_A(\A') &\le \pcost_A(\A)
  &&\hspace{-20pt}\text{if }\A' \subseteq \A,\\
\tag{sub-additivity}
\vphantom{\big|}
  \pcost_A(\A_1 \cup \A_2) &\le \pcost_A(\A_1) + \pcost_A(\A_2)
  &&\hspace{-20pt}\text{for all }\A_1,\A_2,\\
\tag{join rule}
\vphantom{\Big|}
  \pcost_{\AB}(\A \bowtie \B) &\le \pcost_A(\A) + \pcost_B(\B)
  &&\hspace{-20pt}\text{if }\A \in \smallP_A,\, \B \in \smallP_B.
\end{align}

The essential difference between $\cost$ and $\pcost$ is that $\cost$ allows arbitrary joins, while $\pcost$ only allows joins as prescribed by the given pattern. Viewed as a minimum construction cost (see Remark \ref{rmk:construction}), this means that $\pcost$ has more highly constrained rules of construction compared with $\cost$. Consequently, $\cost_{G_A}(\A) \le \pcost_A(\A)$ for every pattern $A$ and $\A \in \P_G$. Note that this inequality goes in the wrong direction for the purpose of proving a lower bound on $\cost$. In \S\ref{sec:relationship} we give a different inequality between $\cost$ and $\pcost$ in the right direction. 

\begin{thm}[Lower Bound for $\pcost$]\label{thm:pattern-lb}
For every pattern $A$ and pathset $\A \in \P_A$, 
$$
  \pcost_A(\A) \ge \m^{(1/6)\log(\l{A})+\c{A}} \cdot \delta(\A).
$$
\end{thm}

The game plan for the rest of the paper is as follows: in \S\ref{sec:relationship} we derive our lower bound for $\cost$ (Theorem \ref{thm:pathset-lb}) from Thereom \ref{thm:pattern-lb}. In \S\ref{sec:proj-rest} we establish some important properties of $\pcost$. Finally, in \S\ref{sec:lower-bound} we give the proof of Theorem \ref{thm:pattern-lb}.

\begin{rmk}[Dual Characterization of $\pcost$]\label{rmk:dual2}
Similar to the dual characterization of $\cost$ mentioned in Remark \ref{rmk:dual1}, $\pcost$ has a {\em dual characterization} as the unique pointwise maximal function from $\{(A,\A) : A$ is a pattern and $\A \in \P_A\}$ to $\R$ which satisfies inequalities (base case), (monotonicity), (sub-additivity) and (join rule). This fact is established by a straightforward induction on patterns (omitted here since we don't actually use this dual characterization in our lower bound).

This dual characterization suggests an obvious ``direct method'' for proving a lower bound on $\pcost$: find an explicit function from pairs $(A,\A)$ to $\R$ and show that this function satisfies inequalities (base case), (monotonicity), (sub-additivity) and (join rule). This is analogous to the ``direct method'' of proving a formula size lower bound via a {\em complexity measure}, defined as a function $M$ from $\{$boolean functions on $n$ variables$\}$ to $\R$ satisfying inequalities $M(f \wedge g) \le M(f) + M(g)$ and $M(f \vee g) \le M(f) + M(g)$ in addition to base case inequalities $M(f) \le 0$ if $f$ is constant and $M(f) \le 1$ if $f$ is a coordinate function.

Using the direct method, we were only able to prove lower bounds on $\pcost$ for a few restricted classes patterns (see Appendix \ref{sec:easier}). For general patterns, we could not prove a lower bound along the lines of Theorem \ref{thm:pattern-lb} using the direct method. We still do not know of any explicit function which satisfies (base case), (monotonicity), (sub-additivity) and (join rule) and maps $(A,[n]^{P_k})$ to $n^{\Omega(\log k)}$ for all patterns $A$ with graph $P_k$. A priori, it is not even clear whether any such nice explicit function exists.\footnote{A natural approach is to consider functions of the form $n^{c_A} \cdot \nu(\A)$ where $c_A$ is a constant depending only on $A$ and $\nu : \P_A \to \R$ is a monotone sub-additive function, such as $\delta$ or $\mu_S$ or $\pi_S$ or any norm on $\R^{[n]^{V_A}}$ (viewing $\P_A \cong \smash{\{0,1\}^{[n]^{V_A}}}$ as a subset of $\R^{[n]^{V_A}}$). For such functions, one only needs to show (join rule); (base case) can be handled by appropriate scaling.}

The proof of Theorem \ref{thm:pattern-lb} which we present in \S\ref{sec:lower-bound} does not proceed via the direct method. In particular, neither the function $\m^{(1/6)\log(\l{A})+\c{A}} {\cdot} \delta(\A)$ nor $\m^{\PHI{A}} {\cdot} \delta(\A)$ (defined in \S\ref{sec:PHI}) satisfies inequality (join rule). Rather, our proof involves a more subtle induction on patterns. 
\end{rmk}

\section{From $\cost$ to $\pcost$}\label{sec:relationship}

In this section, we prove:
\begin{redu}\label{redu:2}
\normalfont
Theorem \ref{thm:pattern-lb} (lower bound on $\pcost$)
$\Longrightarrow$
Theorem \ref{thm:pathset-lb} (lower bound on $\cost$).
\end{redu}
The following definition of {\em strict pattern} is only needed in this section. Rather than $A,B,C$, we write $\alpha,\beta,\gamma$ for this special class of patterns.

\begin{df}
A pattern $\alpha$ is {\em strict} if $G_{\alpha''} \subset G_{\alpha'}$ for all $\alpha'' \prec \alpha' \preceq \alpha$. For a pattern graph $G$, let $\Strict{G}$ denote the set of strict patterns $\alpha$ with graph $G$.
\end{df}

It is important that the number of strict patterns with a given pattern graph is bounded (though doubly exponential in $|E_G|$).

\begin{la}\label{la:number-strict}
For every pattern graph $G$ with $r$ edges, there are only $2^{O(2^r)}$ strict patterns with graph $G$.
\end{la}

\begin{proof}
Denote by $s(r)$ the number of strict patterns supported on any fixed set of $r$ edges. (Note that $|\Strict{G}|$ depends only on $|E_G|$.) Then we have $s(1) = 1$ and $s(r) \le (r \cdot s(r-1))^2$ for all $r \ge 2$. Therefore,
\[
  s(r) \le  
  r^2(r-1)^4(r-2)^8 \cdots 3^{2^{r-2}} 2^{2^{r-1}} 1^{2^r}
  = 2^{O(2^r)}.\qedhere
\]
\end{proof}

We now give the main lemma needed for Reduction \ref{redu:2}.

\begin{la}\label{la:for-redu2}
For every pattern graph $G$ and pathset $\A \in \P_G$, there is an indexed family $\{\A^{(\alpha)}\}_{\alpha \in \Strict{G}}$ of sub-pathsets $\A^{(\alpha)} \subseteq \A$ such that
\[
  \A = \bigcup_{\alpha \in \Strict{G}} \A^{(\alpha)}
  \quad\text{ and }\quad
  \forall \alpha \in \Strict{G},\ \pcost_\alpha(\A^{(\alpha)}) \le \cost_G(\A).
\]
\end{la}

\begin{proof}
By induction on $|E_G|$. The lemma is trivial if $|E_G| \le 1$ (since in this case $|\Strict{G}|=1$). For the induction step, suppose $G$ is a pattern graph with $\ge 2$ edges. By Definition \ref{df:pc1} of $\cost$, there exists a sequence $(H_i,K_i,\B_i,\C_i)_i$ with 
\[
  H_i,K_i \subset G,
  \quad\
  H_i \cup K_i = G,
  \quad\
  \B_i \in \smallP_{H_i},
  \quad\
  \C_i \in \smallP_{K_i}
\]
such that
\[
  \A \subseteq \bigcup_i \B_i \bowtie \C_i
  \quad\
  \text{and}
  \quad\
  \cost_G(\A) = \sum_i \cost_{H_i}(\B_i) + \cost_{K_i}(\C_i).
\]
For each $\alpha = \{\beta,\gamma\} \in \Strict{G}$, define $\A^{(\alpha)}$ inductively by
\[
  \A^{(\alpha)} \defeq \A \cap 
  \bigcup_{i \,:\, H_i = G_\beta,\, K_i = G_\gamma}
  \B_i^{(\beta)} \bowtie \C_i^{(\gamma)}.
\]

First, we show that
$
  \A = \bigcup_{\alpha \in \Strict{G}} \A^{(\alpha)}
$.
The inclusion $\supseteq$ is obvious. For the inclusion $\subseteq$, consider any $x \in \A$. Then $x$ belongs to $\B_i \bowtie \C_i$ for some $i$. This means that $x_{V_{H_i}} \in \B_i$ and $x_{V_{K_i}} \in \C_i$. By the induction hypothesis, there exist $\beta \in \Strict{H_i}$ and $\gamma \in \Strict{K_i}$ such that $x_{V_{H_i}} \in \B_i^{(\beta)}$ and $x_{V_{K_i}} \in \C_i^{(\gamma)}$. Let $\alpha = \{\beta,\gamma\}$ and note that $\alpha \in \Strict{G}$. Since $x \in \B_i^{(\beta)} \bowtie \C_i^{(\gamma)}$, it follows that $x \in \A^{(\alpha)}$, proving the inclusion $\subseteq$.

Finally, for all $\alpha \in \Strict{G}$, we show $\pcost_\alpha(\A) \le \cost_G(\A)$ as follows:
\begin{align*}
  \pcost_\alpha(\A)
  &\le \pcost_\alpha(\bigcup_{i \,:\, H_i = G_\beta,\, K_i = G_\gamma}
    \B_i^{(\beta)} \bowtie \C_i^{(\gamma)})
  &&\text{(monotonicity)}\\
  &\le \sum_{i \,:\, H_i = G_\beta,\, K_i = G_\gamma}
    \pcost_\alpha(\B_i^{(\beta)} \bowtie \C_i^{(\gamma)})
  &&\text{(sub-additivity)}
\intertext{Noting that $\B_i^{(\beta)}$ and $\C_i^{(\gamma)}$ are small (since $\B_i^{(\beta)} \subseteq \B_i \in \smallP_{H_i}$ and $\C_i^{(\gamma)} \subseteq \C_i \in \smallP_{K_i}$), we continue:}
  &\le \sum_{i \,:\, H_i = G_\beta,\, K_i = G_\gamma}
    \pcost_\beta(\B_i^{(\beta)}) + \pcost_\gamma(\C_i^{(\gamma)})
  &&\text{(join rule)}\\
  &\le \sum_{i \,:\, H_i = G_\beta,\, K_i = G_\gamma}
    \cost_{H_i}(\B_i) + \cost_{K_i}(\C_i)
  &&\text{(ind.\ hyp.)}\\  
  &\le 
  \cost_G(\A).
  &&\qedhere
\end{align*}
\end{proof}

The next corollary follows directly from Lemma \ref{la:for-redu2}.

\begin{cor}\label{cor:from-patterns}
For every pattern graph $G$ and pathset $\A \in \P_G$, there is a strict pattern $\alpha \in \Strict{G}$ and a sub-pathset $\A' \subseteq \A$ such that $\pcost_\alpha(\A') \le \cost_G(\A)$ and $\delta(\A) \le |\Strict{G}| \cdot \delta(\A')$.\qed
\end{cor}

We conclude this section with the proof of Reduction \ref{redu:2}.

\begin{proof}[Proof of Reduction \ref{redu:2}]
Assume Theorem \ref{thm:pattern-lb} and consider arbitrary $\A \in \P_{P_k}$. By Corollary \ref{cor:from-patterns}, there exist $\alpha \in \Strict{P_k}$ and $\A' \subseteq \A$ such that $\cost_\alpha(\A') \le \cost_{P_k}(\A)$ and $\delta(\A) \le |\Strict{P_k}| \cdot \delta(\A') \le 2^{O(2^k)} \cdot \delta(\A')$ (Lemma \ref{la:number-strict}). 
We now have
\begin{align*}
  &&\vphantom{\Big|}\cost_{P_k}(\A) \ge \cost_\alpha(\A')
  &\ge \m^{(1/6)\log(\l{\alpha})+\c{\alpha}} \cdot \delta(\A')
  &&\text{(Theorem \ref{thm:pattern-lb})}&&\\
  &&&\ge \frac{\m^{(1/6)\log(k) + 1}}{2^{O(2^k)}} \cdot \delta(\A)\\
  &&&\ge \frac{n^{(1/6)\log k}}{2^{O(2^k)}} \cdot \delta(\A)
  &&\text{(as $\m = n^{1-\eps} = n^{1-(1/\log k)}$).}&&
\end{align*}
This shows that Theorem \ref{thm:pathset-lb} holds, which completes the proof of the reduction.
\end{proof}

\section{Projection and Restriction}\label{sec:proj-rest}

In this section we establish two key properties of $\pcost$: it is monotone decreasing with respect to projection to sub-patterns (Lemma \ref{la:proj-lemma}) and restriction to unions of components (Lemma \ref{la:rest-lemma}). We also introduce an operation on patterns $A \ominus B$ (Definition \ref{df:ominus}), read as ``$A$ restricted away from $B$''. This notation will be extremely convenient in \S\ref{sec:lower-bound}.

\subsection{$\pcost$ Decreases Under Projection}

\begin{claim}\label{claim:pre-proj-lemma}
For every non-atomic pattern $\AB$ and pathset $\C \in \P_{\AB}$, we have $\pcost_A(\proj_A(\C)) \le \pcost_{\AB}(\C)$.
\end{claim}

\begin{proof}
By Definition \ref{df:pc2}(iii) of $\pcost_{\AB}(\C)$, there is a sequence $(\A_i,\B_i)_i$ such that 
\[
  \A_i \in \smallP_A,\quad
  \B_i \in \smallP_B,\quad
  \C \subseteq \bigcup_i \A_i \bowtie \B_i
  \quad\text{and}\quad
  \pcost_{\AB}(\C) = \sum_i \pcost_A(\A_i) + \pcost_B(\B_i).
\]
Note that
$
  \proj_A(\C) \subseteq \proj_A(\,\bigcup_i \A_i \bowtie \B_i) 
  \subseteq \bigcup_i \A_i.
$
By monotonicity and sub-additivity of $\pcost_A$, it follows that
\begin{align*}
  \pcost_A(\proj_A(\C))
  &\le \pcost_A(\,\bigcup_i \A_i)
  \le \sum_i\pcost_A(\A_i)
  \le \pcost_{\AB}(\C).
  \qedhere
\end{align*}
\end{proof}

\begin{la}[$\pcost$ decreases under projections]\label{la:proj-lemma}
For every pattern $A$ and pathset $\A \in \P_A$ and sub-pattern $A' \preceq A$, $\pcost_{A'}(\proj_{A'}(\A)) \le \pcost_A(\A)$.
\end{la}

\begin{proof}
Induction using Claim \ref{claim:pre-proj-lemma}.
\end{proof}

\subsection{$\pcost$ Decreases Under Restriction}

For a pattern $A$ and a pathset $\A \in \P_A$, Lemma \ref{la:proj-lemma} concerns projections of $\A$ of the form $\proj_{A'}(\A)$ where $A'$ is a sub-pattern of $A$. The restrictions of $\A$ that we consider next are \underline{not} restrictions of the form $\RHO{A'}{\A}{z}$ where $z \in [n]^{V_A \setminus V_{A'}}$. Note that $\RHO{A'}{\A}{z} \subseteq \proj_{A'}(\A)$, so we already have $\pcost_{A'}(\RHO{A'}{\A}{z}) \le \pcost_A(\A)$ by Lemma \ref{la:proj-lemma} and monotonicity of $\pcost_{A'}$.

Rather than restrictions over sub-patterns, we instead consider restrictions of the form $\RHO{S}{\A}{z}$ where $z \in [n]^{V_A \setminus S}$ and $S \subseteq V_A$ is a union of components of $G_A$. We define an operation of restriction on patterns; the restriction $A{\uhr}S$ is a pattern with $V_{A{\uhr}S} = S$. Even though $A{\uhr}S$ is not necessarily a sub-pattern of $A$, we will show that $\pcost_{A{\uhr}S}(\RHO{S}{\A}{z}) \le \pcost_A(\A)$.

\begin{df}[Restriction of Patterns]\label{df:restriction}
\
\begin{enumerate}[(i)]
\item
For all $S \subseteq V_k$, let $\BAR S$ denote the complement $V_k \setminus S$ of $S$ in $V_k$.
\item
For a pattern $A$, we say that $S$ is {\em $A$-respecting} if $V_A \cap S$ is a union of components of $G_A$. 

Note that $S$ is $A$-respecting $\Longleftrightarrow$ $\BAR S$ is $A$-respecting $\Longleftrightarrow$ every leaf in $A$ is labeled by an edge $\edge{v_i}{v_{i+1}} \in E_k$ such that $\{v_i,v_{i+1}\} \subseteq S$ or $\{v_i,v_{i+1}\} \subseteq \BAR S$. Also note that if $S$ is $\AB$-respecting, then it is both $A$-respecting and $B$-respecting and $\AB{\uhr}S = \{A{\uhr}S,B{\uhr}S\}$.
\item
If $S$ is $A$-respecting, we denote by $A{\uhr}S$ the pattern obtained from $A$ by pruning all leaves labeled by elements of $\BAR S \times \BAR S$. 

For example, if $A$ is the pattern $\{\{\edge{v_1}{v_2},\edge{v_5}{v_6}\},\{\edge{v_2}{v_3},\edge{v_6}{v_7}\}\}$ and $S$ is the $A$-respecting set $\{v_1,v_2,v_3\}$, then $A{\uhr}S = \{\edge{v_1}{v_2},\edge{v_2}{v_3}\}$. Note that $A{\uhr}S = \{\edge{v_1}{v_2},\edge{v_2}{v_3}\}$ also when $S$ is the $A$-respecting set $\{v_1,v_2,v_3,v_4\}$; in general, $A{\uhr}S = A{\uhr}(V_A \cap S)$. Also note that $A{\uhr}S$ is not a sub-pattern of $A$ in this example.
\end{enumerate}
\end{df}

Before stating the main lemma of this subsection, recall our convention concerning notation $\RHO{S}{\A}{z}$ (see Definition \ref{df:density-etc}): for every pattern $A$ and pathset $\A \in \P_A$ and $S \subseteq V_k$ and $z \in [n]^{\BAR S}$, the pathset $\RHO{S}{\A}{z} \in \P_{A{\uhr}S}$ is defined by
$\RHO{S}{\A}{z} \defeq \RHO{V_A \cap S}{\A}{z'} = \{y \in V_A \cap S : yz' \in \A\}$ where $z' = z_{V_A \setminus S}$.

\begin{la}[$\pcost$ decreases under restrictions]\label{la:rest-lemma}
For every pattern $A$ and pathset $\A \in \P_A$ and $A$-respecting $S \subseteq V_k$ and $z \in [n]^{\BAR S}$, we have $\pcost_{A{\uhr}S}(\RHO{S}{\A}{z}) \le \pcost_A(\A)$.
\end{la}

\begin{proof}
By induction on patterns. The lemma is trivial for empty and atomic patterns. For the induction step, consider a non-atomic pattern $\AB$ and assume the lemma holds for $A$ and $B$. Let $\C \in \P_{\AB}$, let $S$ be a $\AB$-respecting subset of $V_k$, and let $z \in [n]^{\BAR S}$. By Def.\ \ref{df:pc2}(iii) of $\pcost_{\AB}(\C)$, there is a sequence $(\A_i,\B_i)_i$ such that 
\[
  \A_i \in \smallP_A,\quad
  \B_i \in \smallP_B,\quad
  \C \subseteq \bigcup_i \A_i \bowtie \B_i
  \quad\text{and}\quad
  \pcost_{\AB}(\C) = \sum_i \pcost_A(\A_i) + \pcost_B(\B_i).
\]
By Lemma \ref{la:rest-small}, $\RHO{S}{\A_i}{z} \in \smallP_{A{\uhr}S}$ and $\RHO{S}{\B_i}{z} \in \smallP_{B{\uhr}S}$. We now have
\begin{align*}
  \pcost_{\AB{\uhr}S}(\RHO{S}{\C}{z})
  &\le \vphantom{\big|}\pcost_{\AB{\uhr}S}\Big(\bigcup_i \pRHO{S}{\A_i \bowtie \B_i}{z}\Big)
  &&\text{(monotonicity)}\\
  &\le \vphantom{\Big|}\sum_i\pcost_{\AB{\uhr}S}(\pRHO{S}{\A_i \bowtie \B_i}{z})
  &&\text{(sub-additivity)}\\
  &= \vphantom{\big|}\sum_i\pcost_{\{A{\uhr}S,B{\uhr}S\}}(\RHO{S}{\A_i}{z} \bowtie \RHO{S}{\B_i}{z})\\
  &\le \vphantom{\Big|}\sum_i \pcost_{A{\uhr}S}(\RHO{S}{\A_i}{z})+\pcost_{B{\uhr}S}(\RHO{S}{\B_i}{z})
  &&\text{(join rule)}\\
  &\le \vphantom{\big|}\sum_i \pcost_A(\A_i)+\pcost_B(\B_i)
  &&\text{(ind.\ hyp.)}\\
  &= \vphantom{\big|}\pcost_{\AB}(\C).
  &&\qedhere
\end{align*}
\end{proof}

\begin{la}[Smallness is preserved under restriction]\label{la:rest-small}
For every pattern $A$ and small pathset $\A \in \smallP_A$ and $A$-respecting $S \subseteq V_k$ and $z \in [n]^{\BAR S}$, we have $\RHO{S}{\A}{z} \in \smallP_{A{\uhr}S}$.
\end{la}

\begin{proof}
Immediate from Definition \ref{df:small} of small pathsets.
\end{proof}

\begin{rmk}\label{rmk:proj}
Smallness is preserved under joins (Lemma \ref{la:join}) and restrictions to union of components (Lemma \ref{la:rest-small}). However, smallness is not preserved under projection to unions of components. A counterexample is the pattern $A = \{\edge{v_1}{v_2},\edge{v_3}{v_4}\}$ and pathset $\A = \{x \in [n]^{V_A} : x_1 = x_3 \text{ and } x_2 = x_4\} \in \smallP_A$. Letting $A'$ be the atomic sub-pattern $\edge{v_1}{v_2}$ of $A$, we have $\pi_{A'}(\A) = 1$, hence $\proj_{A'}(\A) \notin \smallP_{A'}$.
\end{rmk}

\subsection{The Operation $A \ominus B$}

As a convenient notation, we introduce an operation $A \ominus B$ on patterns, read as {\em $A$ restricted away from $\B$}.

\begin{df}\label{df:ominus}
For patterns $A$ and $B$, we write $A \ominus B$ for the pattern $A{\uhr}S$ where $S \subseteq V_A$ consists of the components of $G_A$ which do not intersect $V_B$. For example, if $A = \{\{\edge{v_1}{v_2},\edge{v_4}{v_5}\},\{\edge{v_2}{v_3},\edge{v_5}{v_6}\}\}$ (so $G_A$ is the union of paths $v_1v_2v_3$ and $v_4v_5v_6$) and $B = \{\edge{v_6}{v_7}\}$, then $A \ominus B = \{\edge{v_1}{v_2},\edge{v_2}{v_3}\}$.
\end{df}

We conclude this section with two basic lemmas about this operation.

\begin{la}\label{la:c}
For all patterns $C = \AB$ and $A' \preceq A$ and $B' \preceq B$,
\[
  \c{C} \le \c{A'} + \c{B' \ominus A'} + \c{C \ominus \{A',B'\}}.
\]
\end{la}

\begin{proof}
Each component of $G_C$ contains at least one component from at least one of the three vertex-disjoint graphs $G_{A'}$, $G_{B' \ominus A'}$ and $G_{C \ominus \{A',B'\}}$.
\end{proof}

\begin{la}\label{la:split}
For all patterns $C = \AB$ and $A' \preceq A$ and $B' \preceq B$ and pathsets $\A \in \P_{A}$ and $\B \in \P_{B}$,
\[
    \delta(\A \bowtie \B) \le 
    \pi_{A'}(\A) 
    \cdot
    \mu_{B' \ominus A'}(\proj_{B'}(\B)) 
    \cdot
    \mu_{C \ominus \{A',B'\}}(\A \bowtie \B).
\]
\end{la}

\begin{proof}
By Lemma \ref{la:basic-ineqs}(\ref{eq:basic2}),
\[
    \delta(\A \bowtie \B) \le 
    \pi_{A'}(\A) 
    \cdot
    \mu_{V_{B'} \setminus V_{A'}}(\proj_{B'}(\B)) 
    \cdot
    \mu_{V_C \setminus (V_{A'} \cup V_{B'})}(\A \bowtie \B).
\]
Since $V_{B' \ominus A'} \subseteq V_{B'} \setminus V_{A'}$ and $V_{C \ominus \{A',B'\}} \subseteq V_C \setminus (V_{A'} \cup V_{B'})$, by Lemma \ref{la:basic-ineqs}(\ref{eq:basic0}),
\begin{align*}
  \mu_{V_{B'} \setminus V_{A'}}(\proj_{B'}(\B)) 
  &\le \mu_{B' \ominus A'}(\proj_{B'}(\B)),\\
  \mu_{V_C \setminus (V_{A'} \cup V_{B'})}(\A \bowtie \B) 
  &\le \mu_{C \ominus \{A',B'\}}(\A \bowtie \B).
\end{align*}
Combining these inequalities finishes the proof.
\end{proof}

\section{Lower Bound for $\pcost$}\label{sec:lower-bound}

In this section we prove Theorem \ref{thm:pattern-lb}, our lower bound for $\pcost$. Recall that $\l{A}$ denote the length of the longest path in $G_A$, i.e.,\ the number of edges in the largest component of $G_A$.

\begin{reptheorem}{thm:pattern-lb}
\textup{(restated)\ }\itshape
For every pattern $A$ and pathset $\A \in \P_{A}$, 
\[
  \pcost_A(\A) \ge \m^{(1/6)\log(\l{A}) + \c{A}} \cdot \delta(\A).
\] 
\end{reptheorem}

To prove Theorem \ref{thm:pattern-lb}, first we define an auxiliary function $\PHI{} : \{\text{patterns}\} \to \R$. 
We then prove two lemmas: $\pcost_A(\A) \ge \m^{\PHI{A}} \delta(\A)$ (Lemma \ref{la:PHI}) 
and $\PHI{A} \ge \tsfrac16\log(\l{A}) + \c{A}$ 
(Lemma \ref{la:PHI-lb}).

\subsection{Definition of $\PHI{A}$}\label{sec:PHI}

\begin{df}\label{df:PHI}
Let $\PHI{} : \{$patterns$\} \to \R$ be the unique minimal function such that the following hold:
\begin{itemize}
  \item
    $\PHI{A} = 0$ if $A$ is empty, and $\PHI{A} = 2$ if $A$ is atomic,
  \item
    for every non-atomic pattern $C = \AB$ and sub-patterns $A' \preceq A$ and $B' \preceq B$,
    \begin{description}\setlength{\itemsep}{-5pt}
    \item[\ \ \ \ $(\dag)^C_{A',B\phantom{'}}$]\ \ \ \ 
      $\PHI{C} \ge \PHI{A'} + \c{B \ominus A'} + \c{C \ominus \{A',B\}}$,\\
    \item[\ \ \ \ $(\ddag)^C_{A',B'}$]\ \ \ \ 
      $\PHI{C} \ge 
      \ds\frac{\PHI{A'} + \PHI{B' \ominus A'} 
      + \c{C} + \c{C \ominus \{A',B'\}}}{2}$.
    \end{description}
\end{itemize}
We refer to ($\dag$) and ($\ddag$) as the ``one-sided'' and ``balanced'' inequalities. Note that since $\AB$ and $\BA$ are considered to be the same pattern, we also have the reverse inequalities $(\dag)^C_{B',A}$ and $(\ddag)^C_{B',A'}$. For better readability, we write $\PHI{A}$ inline and ${}^{\PHII{A}}$ in superscripts.
\end{df}

Some brief remarks on this definition:
\begin{enumerate}[---\hspace{.5pt}]
\item
Minimality of $\Phi$ among functions satisfying these inequalities means that for every non-atomic pattern $C = \AB$, at least one of the four inequalities $(\dag)^C_{A',B}$, $(\dag)^C_{B',A}$, $(\ddag)^C_{A',B'}$, $(\ddag)^C_{B',A'}$ is tight (i.e.\ holds with equality) for some $A' \preceq A$ and $B' \preceq B$. 

\item
Note that $\Phi$ is monotone decreasing with respect to sub-patterns, that is, $\PHI{A'} \le \PHI{A}$ for all $A' \preceq A$ (by inequalities ($\dag$)).

\item
$\Phi$ increases by means of the contribution of $\Delta$'s: if we remove the $\Delta$'s from $(\dag)^C_{A',B}$ and $(\ddag)^C_{A',B'}$ (replacing these inequalities by $\PHI{C} \ge \PHI{A'}$ and $\PHI{C} \ge \frac{1}{2}(\PHI{A'} + \PHI{B'})$ respectively), then we would have $\PHI{A} = 2$ for every nonempty pattern $A$. Intuitively, in the attempt to lower bound $\PHI{A}$, the objective of the game is to pick up as many $\Delta$'s as possible.

\item
For the patterns $A_k$ and $B_k$ defined in Appendix \ref{sec:key-examples}, we have $\PHI{A_k} \ge \PHI{A_{\lceil k/4 \rceil}} + 1$ by ($\dag$) and $\PHI{B_k} \ge \PHI{B_{\lceil (k-1)/2 \rceil}} + \frac{1}{2}$ by ($\ddag$) for all $k \ge 4$. It follows that $\PHI{A_k} \ge \frac{1}{2}\log k - O(1)$ and $\PHI{B_k} \ge \frac{1}{2}\log k - O(1)$.
\end{enumerate}

\subsection{Showing $\pcost_A(\A) \ge \m^{\PHI{A}} \delta(\A)$}\label{sec:first-lemma}

We now prove the most important lemma in the overall proof of Theorem \ref{thm:pattern-lb}. Lemma \ref{la:PHI} accounts for the definition of $\PHI{A}$ (essentially $\PHI{A}$ is the maximum function for which the argument of Lemma \ref{la:PHI} is valid). The two cases ($\dag$) and ($\ddag$) in the proof are inspired by the special cases proved in Appendix \ref{sec:easier}.

\begin{la}\label{la:PHI}
For every pattern $A$ and pathset $\A \in \P_{A}$, $\pcost_A(\A) \ge \m^{\PHII{A}}\delta(\A)$.
\end{la}

\begin{proof}
We argue by induction on patterns. The base case where $A$ is empty or atomic is trivial. For the induction step, consider a non-atomic pattern $C = \AB$ and assume the lemma holds for all smaller patterns.

We claim that it suffices to show that
\begin{equation}\label{eq:how-to-lb}
  \m^{\PHII{C}} \delta(\A \bowtie \B) \le \pcost_A(\A) + \pcost_B(\B)
\end{equation}
for all $\A \in \smallP_A$, $\B \in \smallP_B$. To see that this suffices, consider any $\C \in \P_{C}$. By the dual characterization of pathset complexity, there exists a covering $\C \subseteq \bigcup_i \A_i \bowtie \B_i$ by joins of small pathsets $\A_i$ and $\B_i$ such that $\pcost_C(\C) = \sum_i \pcost_A(\A_i) + \pcost_B(\B_i)$. Note that
\[
  \delta(\C) 
  \le \delta(\,\bigcup_i \A_i \bowtie \B_i) 
  \le \sum_i \delta(\A_i \bowtie \B_i). 
\]
Assuming (\ref{eq:how-to-lb}) holds for all $\A_i$ and $\B_i$, we have
\[
  \m^{\PHII{C}} \delta(\C) 
  \le \sum_i \m^{\PHII{C}} \delta(\A_i \bowtie \B_i)
  \le \sum_i \pcost_A(\A_i)+\pcost_B(\B_i)
  = \pcost_C(\C).
\]

We now turn to proving inequality (\ref{eq:how-to-lb}). Fix small pathsets $\A \in \smallP_A$ and $\B \in \smallP_B$. Note that $\A \bowtie \B \in \smallP_C$ by Lemma \ref{la:join}. Recall that at least one of the four inequalities $(\dag)^C_{A',B}$, $(\ddag)^C_{A',B'}$, $(\dag)^C_{B',A}$, $(\ddag)^C_{B',A'}$ is tight for some $A' \preceq A$ and $B' \preceq B$. By symmetry of the argument, we consider only the first two possibilities without loss of generality .

\paragraph{Case $(\dag)$ (one-sided induction case):}  Assume that there exists $A' \preceq A$ such that $(\dag)^C_{A',B}$ is tight, that is,
\begin{equation}\label{eq:casedag}
  \PHI{C} = \PHI{A'} + \c{B \ominus A'} + \c{C \ominus \{A',B\}}.
\end{equation}
By Lemma \ref{la:split}, we have 
\begin{align*}
  \delta(\A \bowtie \B) 
  \le 
  \pi_{A'}(\A) 
  \cdot
  \mu_{B \ominus A'}(\B) 
  \cdot
  \mu_{C \ominus \{A',B\}}(\A \bowtie \B).
\end{align*}
Since $\B$ is $B$-small and $\A \bowtie \B$ is $C$-small, we have
\begin{align*}
  \mu_{B \ominus A'}(\B) 
  \le 
  \m^{-\cc{B \ominus A'}}
  \quad\text{ and }\quad
  \mu_{C \ominus \{A',B\}}(\A \bowtie \B)
  \le 
  \m^{-\cc{C \ominus \{A',B\}}}.
\end{align*}
Combining these inequalities (and substituting $\delta(\proj_{A'}(\A))$ for $\pi_{A'}(\A)$), we have
\begin{equation}\label{eq:delta-bound}
  \delta(\A \bowtie \B) 
  \le
  \m^{-\cc{B \ominus A'}-\cc{C \ominus \{A',B\}}} 
  \delta(\proj_{A'}(\A)).
\end{equation}
Using the fact that $\pcost$ decreases under projections, together with the induction hypothesis, we have 
\begin{align*}
  &&\m^{\PHII{C}} \delta(\A \bowtie \B)
    &= \vphantom{\big|}\m^{\PHII{A'} + \cc{B \ominus A'} + \cc{C \ominus \{A',B\}}} 
     \delta(\A \bowtie \B)
    &&\text{(by (\ref{eq:casedag}))}&&\\
  &&&\le \vphantom{\big|}\m^{\PHII{A'}} \delta(\proj_{A'}(\A))
    &&\text{(by (\ref{eq:delta-bound}))}\\
  &&&\le \vphantom{\big|}\pcost_{A'}(\proj_{A'}(\A))
    &&\text{(ind.\ hyp.)}\\
  &&&\le \vphantom{\big|}\pcost_A(\A)
    &&\text{(Lemma \ref{la:proj-lemma})}.
\end{align*}
Therefore, (\ref{eq:how-to-lb}) holds in this case.

\paragraph{Case $(\ddag)$ (balanced induction case):}  Assume that there exist $A' \preceq A$ and $B' \preceq B$ such that $(\ddag)^C_{A',B'}$ is tight, that is,
\begin{equation}\label{eq:caseddag}
  \PHI{C} = \frac{\PHI{A'} + \PHI{B' \ominus A'} + \c{C} + \c{C \ominus \{A',B'\}}}{2}.
\end{equation}
By Lemma \ref{la:split}, we have
\begin{align*}
  \delta(\A \bowtie \B) 
    \le \pi_{A'}(\A) \cdot
      \mu_{B' \ominus A'}(\proj_{B'}(\B)) \cdot
      \mu_{C \ominus \{A',B'\}}(\A \bowtie \B).
\end{align*}
By definition of $\mu_{B' \ominus A'}$, there exists $z \in [n]^{\VV{B'} \setminus \VV{B' \ominus A'}}$ such that
\begin{equation*}
  \mu_{B' \ominus A'}(\proj_{B'}(\B)) 
    = \delta(\RHO{B' \ominus A'}{\proj_{B'}(\B)}{z}).
\end{equation*}
$C$-smallness of $\A \bowtie \B$ implies both
\begin{align*}
  \delta(\A \bowtie \B) \le \m^{-\cc{C}} 
  \quad\text{ and }\quad
  \mu_{C \ominus \{A',B'\}}(\A \bowtie \B)
    \le \m^{-\cc{C \ominus \{A',B'\}}}.
\end{align*}
Taking square roots and combining these inequalities, we have
\begin{equation}\label{eq:the-delta}
  \delta(\A \bowtie \B) 
  \le
  \sqrt{\m^{- \cc{C} - \cc{C \ominus \{A',B'\}}}
  \cdot \pi_{A'}(\A) \cdot \delta(\RHO{B' \ominus A'}{\proj_{B'}(\B)}{z})}.
\end{equation}

Using the fact that $\pcost$ decreases under projections and restrictions (Lemmas \ref{la:proj-lemma} and \ref{la:rest-lemma}), together with the induction hypothesis, we have
\begin{align}
\label{eq:A-ind}
  &&\m^{\PHII{A'}} \pi_{A'}(\A)
  = \m^{\PHII{A'}} \delta(\proj_{A'}(\A))
  &\le \pcost_{A'}(\proj_{A'}(\A))
    &&\text{(ind.\ hyp.)}
    &&\\
\notag
  &&&\le \pcost_A(\A)
    &&\text{(Lemma \ref{la:proj-lemma})}
\end{align}
and also
\begin{align}
\label{eq:B-ind}
  &&\m^{\PHII{B' \ominus A'}} \delta(\RHO{B' \ominus A'}{\proj_{B'}(\B)}{z})
  &\le \pcost_{B' \ominus A'}(\RHO{B' \ominus A'}{\proj_{B'}(\B)}{z})
    &&\text{(ind.\ hyp.)}&&\\
\notag
  &&&\le \pcost_{B'}(\proj_{B'}(\B))
    &&\text{(Lemma \ref{la:rest-lemma})}\\
\notag
  &&&\le \pcost_B(\B)
    &&\text{(Lemma \ref{la:proj-lemma})}.
\end{align}

We now finish the proof using the inequality of arithmetic and geometric means (AM-GM inequality):
\begin{align*}
  \m^{\PHII{C}} \delta(\A \bowtie \B)
  &= \sqrt{\vphantom{\big|}
    \m^{\PHII{A'} + \PHII{B' \ominus A'} + \cc{C} 
    + \cc{C \ominus \{A',B'\}}}} 
    \cdot \delta(\A \bowtie \B)
  &&\text{(by (\ref{eq:caseddag}))}\\
  &\le \vphantom{\Big|}\sqrt{\m^{\PHII{A'} + \PHII{B' \ominus A'}} 
  \cdot
  \pi_{A'}(\A)
  \cdot 
  \delta(\RHO{B' \ominus A'}{\proj_{B'}(\B)}{z})}
  &&\text{(by (\ref{eq:the-delta}))}\\
  &\le \frac12\Big(\m^{\PHII{A'}} \pi_{A'}(\A) +
    \m^{\PHII{B' \ominus A'}} \delta(\RHO{B' \ominus A'}{\proj_{B'}(\B)}{z})\Big)
  &&\text{(AM-GM ineq.)}\\
  &\le \vphantom{\Big|}\frac12\Big(\pcost_A(\A) + \pcost_B(\B)\Big)
  &&\text{(by (\ref{eq:A-ind}), (\ref{eq:B-ind})).}
\end{align*}
Therefore, (\ref{eq:how-to-lb}) holds in this case also, which concludes the proof.
\end{proof}

\subsection{Showing $\PHI{A} \ge \tsfrac16\log(\l{A}) + \c{A}$}\label{sec:second-lemma}

We now complete the proof of Theorem \ref{thm:pattern-lb} by proving Lemma \ref{la:PHI-lb} ($\PHI{A} \ge \frac16\log(\l{A}) + \c{A}$ for all patterns $A$). We require one preliminary lemma.

\begin{la}\label{la:A-S}
For every pattern $A$ and $A$-respecting $S \subseteq V_k$, we have $\PHI{A} \ge \PHI{A{\uhr}S} + \c{\smash{A{\uhr}\BAR S}}$.
\end{la}

\begin{proof}
We argue by induction on patterns. The lemma is trivial when $A$ is empty or atomic. For the induction step, consider any non-atomic pattern $C = \AB$ and assume the lemma holds for all smaller patterns. Let $S$ be any $C$-respecting subset of $V_k$. 

Noting that $C{\uhr}S = \{A{\uhr}S, B{\uhr}S\}$ and every sub-pattern of $A{\uhr}S$ has the form $A'{\uhr}S$ where $A' \preceq A$ (and similarly for $B{\uhr}S$), it follows that that at least one the four inequalities 
\[
  (\dag)^{C{\uhr}S}_{A'{\uhr}S,B{\uhr}S},\quad\
  (\ddag)^{C{\uhr}S}_{A'{\uhr}S,B'{\uhr}S},\quad\
  (\dag)^{C{\uhr}S}_{B'{\uhr}S,A{\uhr}S},\quad\
  (\ddag)^{C{\uhr}S}_{B'{\uhr}S,A'{\uhr}S}
\]
is tight for some $A' \preceq A$ and $B' \preceq B$. Once again, without loss of generality, we consider just the first two possibilities.

First, consider the case that there exists $A' \preceq A$ for which $(\dag)^{C{\uhr}S}_{A'{\uhr}S,B{\uhr}S}$ is tight, that is,
\begin{equation}\label{eq:case1}
  \PHI{C{\uhr}S} = \PHI{A'{\uhr}S} + \c{(B \ominus A'){\uhr}S} 
  + \c{(C \ominus \{A',B\}){\uhr}S}.
\end{equation} 
In this case, we have
\begin{align*}
  \PHI{C}
  &\ge
  \PHI{A'} 
    + \c{B \ominus A'} 
    + \c{C \ominus \{A',B\}}
  &&\text{(by $(\dag)^C_{A',B}$)}\\
  &\ge
  \PHI{A'} 
    + \c{B \ominus A'} 
    + \c{C \ominus \{A',B\}}\\
  &\mathrel{\phantom{=}}
    + \c{C{\uhr}\BAR S} 
    - \c{A'{\uhr}\BAR S} 
    - \c{(B \ominus A'){\uhr}\BAR S} 
    - \c{(C \ominus \{A',B\}){\uhr}\BAR S}
  &&\text{(Lemma \ref{la:c})}\\
  &=
  \PHI{A'}  
    - \c{A'{\uhr}\BAR S}
    + \c{B \ominus A'{\uhr}S} 
    + \c{C \ominus \{A',B\}{\uhr}S}
    + \c{C{\uhr}\BAR S}\\
  &\ge
  \PHI{A'{\uhr}S} 
    + \c{B \ominus A'{\uhr}S} 
    + \c{C \ominus \{A',B\}{\uhr}S}
    + \c{C{\uhr}\BAR S}
  &&\text{(ind.\ hyp.)}\\
  &\ge
  \PHI{C{\uhr}S}
    + \c{C{\uhr}\BAR S}
  &&\text{(by (\ref{eq:case1})).}
\end{align*}

Finally, consider the alternative that there exist $A' \preceq A$ and $B' \preceq B$ for which $(\ddag)^{C{\uhr}S}_{A'{\uhr}S,B'{\uhr}S}$ is tight, that is,
\begin{equation}\label{eq:case2}
  \PHI{C{\uhr}S} 
  = \frac{\PHI{A'{\uhr}S} + \PHI{(B' \ominus A'){\uhr}S} + \c{C{\uhr}S} + \c{(C \ominus \{A',B'\}){\uhr}S}}{2}.
\end{equation}
In this case, we have
\begin{align*}
  \PHI{C}
  &\ge
  \tsfrac12\big(
      \PHI{A'} 
      + \PHI{B' \ominus A'} 
      + \c{C} 
      + \c{C \ominus \{A',B'\}}
    \big)
  &&\text{(by $(\ddag)^C_{A',B'}$)}\\
  &\ge
  \tsfrac12
  \big(
      \PHI{A'} 
      + \PHI{B' \ominus A'} 
      + (\c{C{\uhr}S} + \c{C{\uhr}\BAR S}) 
      + \c{C \ominus \{A',B'\}}\big) + \mbox{}\\
  &\mathrel{\phantom{=}}
      \tsfrac12
      \big(\c{C{\uhr}\BAR S} 
      - \c{A'{\uhr}\BAR S} 
      - \c{(B' \ominus A'){\uhr}\BAR S} 
      - \c{(C \ominus \{A',B'\}){\uhr}\BAR S}
    \big)
  &&\text{(Lemma \ref{la:c})}\\  
  &=
  \tsfrac12\big(
      \PHI{A'}
      - \c{A'{\uhr}\BAR S} 
      + \PHI{B' \ominus A'}
      - \c{(B' \ominus A'){\uhr}\BAR S}
      + \c{C{\uhr}S}  
      + \c{(C \ominus \{A',B'\}){\uhr}S}
    \big)
    + \c{C{\uhr}\BAR S}\\
  &\ge
  \tsfrac12\big(
      \PHI{A'{\uhr}S} 
      + \PHI{(B' \ominus A'){\uhr}S}
      + \c{C{\uhr}S} 
      + \c{(C \ominus \{A',B'\}){\uhr}S}
    \big)
    + \c{C{\uhr}\BAR S}
  &&\text{(ind.\ hyp.)}\\
  &\ge
  \PHI{C{\uhr}S}
    + \c{C{\uhr}\BAR S}
  &&\text{(by (\ref{eq:case2})).}\vphantom{\big|}
\end{align*}
Having shown $\PHI{C} \ge \PHI{C{\uhr}S} + \c{C{\uhr}\BAR S}$ in both cases, we are done.
\end{proof}

\begin{la}\label{la:PHI-lb}
For every pattern $A$, $\PHI{A} \ge \tsfrac16\log(\l{A}) + \c{A}$.
\end{la}

\begin{proof}
We argue by induction on patterns. The base case where $A$ is empty or atomic is trivial. For the induction step, let $A$ be a non-atomic pattern and assume the lemma holds for all smaller patterns. We will consider a sequence of cases. In each case, after showing that $\PHI{A} \ge \tsfrac16\log(\l{A}) + \c{A}$ under a given hypothesis, we will proceed assuming the negation of that hypothesis. The sequences of cases is summarized at the end of the proof.

First, consider the case that $G_A$ is disconnected (i.e.\ $\c{A} \ge 2$). Let $S$ be the largest component of $G_A$. We have
\begin{align*}
  &&&&\PHI{A} 
    &\ge \PHI{A{\uhr}S} + \c{A{\uhr}\BAR S}
    &&\text{(Lemma \ref{la:A-S})}&&&&\\
  &&&&&\ge \tsfrac16\log(\l{A{\uhr}S}) + \c{A{\uhr}S} + \c{A{\uhr}\BAR S}
    &&\text{(ind.\ hyp.)}\\
  &&&&&= \tsfrac16\log(\l{A}) + \c{A}.
\end{align*}
This proves the lemma in the case where $G_A$ is disconnected.

Therefore, we proceed under the assumption that $G_A$ is connected (i.e.\ $\c{A} = 1$). Without loss of generality, we assume that $G_A = P_k$ (i.e.\ $\l{A} = k$). Our goal is to show that
\[
  \PHI{A} \ge \tsfrac16\log(k)+1.
\]
Since $\PHI{A} \ge 2$ for all nonempty patterns, we may further assume that $k > 2^6$. (Below, we will only need the assumption that $k > 8$.)

Consider the case that there exists a sub-pattern $A' \preceq A$ such that $|E_{A'}| \ge k/8$ and $\c{A'} \ge 2$. Note that $\l{A'} \ge |E_{A'}|/\c{A'}$ (i.e.\ the number of edges in the largest component of $G_{A'}$ is at least the number of edges in $G_{A'}$ divided by the number of components in $G_{A'}$). We have
\begin{align*}
  \PHI{A} \ge \PHI{A'}
  &\ge \tsfrac16\log(\l{A'}) + \c{A'}
  &&\text{(ind.\ hyp.)}\\
  &\ge \tsfrac16\log(k) - \tsfrac12 - \tsfrac16\log(\c{A'}) + \c{A'}
  &&\text{($\l{A'} \ge |E_{A'}|/\c{A'} \ge k/8\c{A'}$)}\\
  &\ge \tsfrac16\log(k) - \tsfrac12 - \tsfrac16\log(2) + 2
  &&\text{($\c{A'} \ge 2$)}\\
  &= \tsfrac16\log(k) + \tsfrac43\\
  &> \tsfrac16\log(k) + 1.
\end{align*}
This proves the lemma in this case.

Therefore, we proceed under the following assumption:
\begin{equation}\tag{$\hspace{-1pt}{\divideontimes}\hspace{-1pt}$}
\label{eq:ast}
 \text{for all $A' \preceq A$, if 
 $|E_{A'}| \ge k/8$
 then }\c{A'} = 1.
\end{equation}
Going forward, the following notation will be convenient: for a proper sub-pattern $B \prec A$, let $\parent{B}$ denote the parent of $B$ in $A$, and let $\sib{B}$ denote the sibling of $B$ in $A$. Note that $\parent{B} = \{B,\sib{B}\} \preceq A$.

It is easy to see that there exist proper sub-patterns $B,Z \prec A$ such that 
\[
  v_0 \in V_B,\qquad
  v_k \in V_Z,\qquad
  |E_B|,|E_Z| < k/8,\qquad
  |E_{\parent{B}}|,|E_{\parent{Z}}| \ge k/8.
\]
Fix any choice of such $B$ and $Z$. Note that $G_{\parent{B}}$ and $G_{\parent{Z}}$ are connected by (\ref{eq:ast}). In particular, $G_{\parent{B}}$ is a path of length $|E_{\parent{B}}|$ with initial endpoint $v_0$, and $G_{\parent{Z}}$ is a path of length $|E_{\parent{Z}}|$ with final endpoint $v_k$.

Consider the case that $\l{\parent{B}} < k/2$ and $\l{\parent{Z}} < k/2$. Note that $V_{\parent{B}}$ and $V_{\parent{Z}}$ are disjoint and hence $\parent{Z} \ominus \parent{B} = \parent{Z}$. Let $Y$ denote the least common ancestor of $\parent{B}$ and $\parent{Z}$ in $A$. We have
\begin{align*}
  \PHI{A} \ge \PHI{Y}
  &\ge \tsfrac12\big(\PHI{\parent{B}} + \PHI{\parent{Z} \ominus \parent{B}} + \c{Y} + \c{Y \ominus \{\parent{B},\parent{Z}\}}\big)
  &&\text{(by $(\ddag)^Y_{\parent{B},\parent{Z}}$)}\\
  &= \tsfrac12\big(\PHI{\parent{B}} + \PHI{\parent{Z}}\big) + \tsfrac12
  &&\text{($\c{Y} \ge 1$)}\\
  &\ge \tsfrac12\big(\tsfrac16\log(\l{\parent{B}}) + \c{\parent{B}} + \tsfrac16\log(\l{\parent{Z}}) + \c{\parent{Z}}\big) + \tsfrac12
  &&\text{(ind.\ hyp.)}\\
  &\ge \tsfrac16\log(k/8) + \tsfrac32\\
  &= \tsfrac16\log(k) + 1.
\end{align*}
(We remark that this is the only place in the proof where the inequality $(\ddag)$ is used and the only tight case which forces $1/6$.) 

Therefore, we proceed under the assumption that $\l{\parent{B}} \ge k/2$ or $\l{\parent{Z}} \ge k/2$. Without loss of generality, we assume that $\l{\parent{B}} \ge k/2$. (We now forget about $Z$ and $\parent{Z}$.)

Before continuing, let's take stock of the assumptions we have made so far:
\[
  G_A = P_k,\quad\
  \text{(\ref{eq:ast})},\quad\
  B \preceq A,\quad\
  v_0 \in V_B,\quad\
  |E_B| < k/8,\quad\
  |E_{\parent{B}}| = \l{\parent{B}} \ge k/2.
\]
Going forward, we will define vertices $v_r,v_s,v_t$ where $0 < r < s < t \le k$. 

We first define $v_r \in B$ and $v_t \in \sib{B}$ as follows: Let $\{v_0,\dots,v_r\}$ be the component of $G_B$ containing $v_0$. (That is, the component of $v_0$ in $G_B$ is a path whose initial vertex is $v_0$; let $v_r$ be the final vertex in this path.) Let $v_t$ be the vertex in $V_{\sib{B}}$ with maximal index $t$ (i.e.\ furthest away from $v_0$).

Note that $E_B$ contains edges $\edge{v_i}{v_{i+1}}$ for all $i \in \{0,\dots,r-1\} \cup \{t,\dots,\lceil k/2 \rceil-1\}$. (In the event that $t < k/2$, since $G_{\parent{B}} = G_B \cup G_{\sib{B}}$ is a path of length $\ge k/2$ and $G_{\sib{B}}$ does not contain vertices $v_{t+1},\dots,v_{\lceil k/2 \rceil}$, it follows that $G_B$ contains all edges between $v_t$ and $v_{\lceil k/2 \rceil}$.) Therefore, $r + (k/2) - t \le |E_B| < k/8$. It follows that
\[
  t-r > 3k/8.
\]

Next, note that $|E_{\sib{B}}| \ge |E_{\parent{B}}| - |E_B| \ge (k/2) - (k/8) > k/8$. It follows that there exists a proper sub-pattern $C \prec \sib{B}$ such that
\[
  v_t \in V_C,\qquad
  |E_C| < k/8,\qquad 
  |E_{\parent{C}}| \ge k/8.
\]
Fix any choice of such $C$.

Consider the case that $\l{\parent{C}} < 3k/8$. Since $G_{\parent{C}}$ is connected (by (\ref{eq:ast})) and $v_t \in V_{\parent{C}}$ and $t - r > 3k/8$, it follows that $V_{\parent{C}} \cap \{v_0,\dots,v_r\} = \emptyset$ and hence $\c{B \ominus \parent{C}} \ge 1$. We have
\begin{align*}
  &&\PHI{A} \ge \PHI{\parent{B}}
  &\ge \PHI{\parent{C}} + \c{B \ominus \parent{C}} + \c{\parent{B} \ominus \{B,\parent{C}\}}
    &&\text{(by $(\dag)^{\parent{B}}_{\parent{C},B}$)}&&\\
  &&&\ge \PHI{\parent{C}} + 1\\
  &&&\ge \tsfrac16\log(\l{\parent{C}}) + \c{\parent{C}} + 1
    &&\text{(ind.\ hyp.)}\\
  &&&\ge \tsfrac16\log(k/8) + 2\\
  &&&> \tsfrac16\log(k) + 1.
\end{align*}

Therefore, we proceed under the assumption that $\l{\parent{C}} \ge 3k/8$. Since $E_{\parent{C}} = E_C \cup E_{\sib{C}}$, we have
\[
  |E_{\sib{C}}| 
  \ge |E_{\parent{C}}| - |E_C| 
  > (3k/8) - (k/8) = k/4.
\]
We now define vertex $v_s \in V_C$. Since $v_t$ is the vertex of $G_{\sib{B}}$ with maximal index, it follows that $\edge{v_t}{v_{t+1}} \notin E_{\sib{B}}$ and hence $\edge{v_t}{v_{t+1}} \notin E_C$ (since $C \prec \sib{B}$). Therefore, the component of $G_C$ containing $v_t$ is a path with final vertex $v_t$; let $v_s$ be the initial vertex in this path. That is, $\{v_s,\dots,v_t\}$ is the component of $G_C$ which contains $v_t$.

Recall that $t-r > 3k/8$ and note that $t-s \le |E_C| < k/8$. Therefore,
\[
  s - r = (t - r) - (t - s) > (3k/8) - (k/8) = k/4.
\]
We now claim that there exists a proper sub-pattern $D \prec \sib{C}$ such that 
\[
  k/8 \le |E_D| < k/4.
\]
To see this, note that there exists a chain of sub-patterns $\sib{C} = D_0 \succ D_1 \succ \dots \succ D_j$ such that $D_j$ is atomic and $D_i = \parent{D_{i-1}}$ and $|E_{D_i}| \ge |E_{\sib{D}_i}|$ for all $i \in \{1,\dots,j\}$. Since $|E_{D_0}| > k/4$ and $|E_{D_j}| = 1$ and $|E_{D_{i-1}}| = |E_{D_i}| + |E_{\sib{D}_i}| \le 2|E_{D_i}|$, it must be the case that there exists $i \in \{1,\dots,j\}$ such that $k/8 \le |E_{D_i}| < k/4$.

Since $|E_D| \ge k/8$, (\ref{eq:ast}) implies that $G_D$ is connected. Since $|E_D| < k/4$ and $s-r > k/4$, it follows that $V_D$ cannot contain both $v_r$ and $v_s$. We are now down to our final two cases: either $v_r \notin V_D$ or $v_s \notin V_D$.

First, suppose that $v_r \notin V_D$. We have $\c{B \ominus D} \ge 1$ and hence
\begin{align*}
  &&\PHI{A} \ge \PHI{\parent{B}}
    &\ge \PHI{D} + \c{B \ominus D} + \c{\parent{B} \ominus \{B,D\}}
    &&\text{(by $(\dag)^{\parent{B}}_{D,B}$)}&&\\
  &&&\ge \PHI{D} + 1\\
  &&&\ge \tsfrac16\log(\l{D}) + \c{D} + 1
    &&\text{(ind.\ hyp.)}\\
  &&&\ge \tsfrac16\log(k/8) + 2\\
  &&&> \tsfrac16\log(k) + 1.
\end{align*}
Finally, we are left with the alternative that $v_s \notin V_D$. In this case $\c{C \ominus D} \ge 1$ and hence (substituting $C$ for $B$ in the above), we have
\begin{align*}
  \PHI{A} \ge \PHI{\parent{C}}
  \ge \PHI{D} + \c{C \ominus D} + \c{\parent{C} \ominus \{C,D\}}
  \ge \PHI{D} + 1
  > \tsfrac16\log(k) + 1.
\end{align*}

We have now covered all cases. In summary, we considered cases in the following sequence:
\begin{enumerate}[\quad\ \,1.\ ]
\setlength{\itemsep}{0pt}
  \item
    \makebox[2.9in]{$\c{A} \ge 2$\hfill}
    else assume wlog $G_A = P_k$,    
 \item
    \makebox[2.9in]{$\exists A' \prec A$ with $\c{A'} \ge 2$ and $\l{A'} \ge k/8$\hfill} 
    else assume (\ref{eq:ast}),
  \item
    \makebox[2.9in]{$|E_{\parent{B}}| < k/2$ and $|E_{\parent{Z}}| < k/2$\hfill} 
    else assume wlog $|E_{\parent{B}}| \ge k/2$,
  \item
    \makebox[2.9in]{$|E_{\parent{C}}| < 3k/8$\hfill}
    else assume $|E_{\parent{C}}| \ge 3k/8$,
  \item
    $v_r \notin E_D$ or $v_s \notin E_D$.
\end{enumerate}
Since $\PHI{A} \ge \tsfrac16\log(\l{A}) + \c{A}$ in each case, the proof is complete.
\end{proof}

As we have now proved Lemmas \ref{la:PHI} and \ref{la:PHI-lb}, this completes the proof of Theorem \ref{thm:pattern-lb} and hence also of Theorem \ref{thm:main}.

\section{Conclusion}\label{sec:conclusion}

We proved the first super-polynomial separation in the power of bounded-depth boolean formulas vs.\ circuits via technique based on the notion of pathset complexity. The most obvious question for future research is whether pathset complexity can be used to derive lower bounds for distance $k(n)$ connectivity in other models of computation. 

We conclude with a comment extending our results to the {\em average-case} setting. Let $p(n) = {\Theta(n^{-\frac{k+1}{k}})}$ be the exact threshold function such that
\[
  \Pr_{G = G(n,p)} [\, G \in \STCONN(k(n)) \,] = \frac{1}{2}
\]
where $G(n,p)$ is the Erd\Horig{o}s-R\'enyi random graph with edge probability $p(n)$. Our proof of Theorem \ref{thm:main} is easily adapted to give the same $n^{(1/6)\log k - O(1)}$ lower bound for bounded-depth formulas $F$ which satisfy
\[
   \Pr_{G = G(n,p)} [\, F(G) = 1 \iff G \in \STCONN(k(n)) \,] \ge \frac{1}{2} + \eps
\]
for any constant $\eps > 0$. Using the idea behind Proposition \ref{prop:upper-bound}, we can construct formulas $F$ of size $n^{(1/2)\log k + O(1)}$ (the best worst-case upper bound I know of is size $n^{\log k + O(1)}$) and depth $O(\log k)$ which solve $\STCONN(k(n))$ in a strong average-case sense:
\[
   \Pr_{G = G(n,p)} [\, F(G) = 1 \iff G \in \STCONN(k(n)) \,] \ge 1 - \exp(-n^{\Omega(1)}).
\]
It would be interesting to close the gap between $\smash{\frac16\log k}$ and $\smash{\frac12\log k}$ in these bounds.

\section*{Acknowledgements}

I want to thank Osamu Watanabe and Rahul Santhanam for many helpful discussions, Stasys Jukna for his interest in this work and stimulating discussions at a time when I was badly stuck on the proof, and Igor Carboni Oliveira for valuable feedback on an earlier draft of this paper.


\appendix{}

\section{Key Examples}\label{sec:key-examples}

We introduce two key examples of patterns, denoted $A_k$ and $B_k$, and present upper bounds for $\pcost$ with respect to these patterns. In the next section, we prove lower bounds for two classes of patterns which generalize $A_k$ and $B_k$. The arguments in these special cases show up in the two cases ($\dag$) and ($\ddag$) of our main lower bound (Theorem \ref{thm:pattern-lb}).

\begin{notn}\label{notn:shift2}
Recall Notation \ref{notn:shift1} for {\em $s$-shifted} pattern graphs $G^{\shift s}$ and pathsets $\A^{\shift s}$. For a pattern $A$ and integer $s$, we define the {\em $s$-shifted} pattern $A^{\shift s}$ analogously by replacing each label $\edge{v_i}{v_{i+1}}$ with the label $\edge{v_{i+s}}{v_{i+s+1}}$.
\end{notn}

\begin{df}[Patterns $A_k$ and $B_k$]
We define patterns $A_k$ and $B_k$ for all $k \ge 1$ by the following induction. Let $A_1 = B_1 \defeq$ the atomic pattern labeled by $\edge{v_0}{v_1}$. For $k \ge 2$, let $A_k \defeq \{A_j^{\vphantom{\shift j}},A_{k-j}^{\shift j}\}$ where $j = \lceil k/2 \rceil$, and let $B_k \defeq \{B_{k-1}^{\vphantom{\shift 1}},B_{k-1}^{\shift 1}\}$. For example, the explicit pictures of $A_8$ and $B_4$ are:
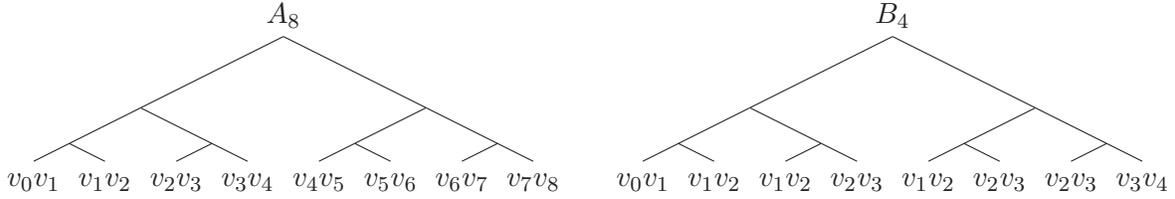
\begin{center}
    \mbox{}\hfill\hfill
    \Tree [.$A_8$ [ [ \makebox[15pt]{$\edge{v_0}{v_1}$} \makebox[15pt]{$\edge{v_1}{v_2}$} ] 
	       [ \makebox[15pt]{$\edge{v_2}{v_3}$} \makebox[15pt]{$\edge{v_3}{v_4}$} ] ]
           [ [ \makebox[15pt]{$\edge{v_4}{v_5}$} \makebox[15pt]{$\edge{v_5}{v_6}$} ] 
           [ \makebox[15pt]{$\edge{v_6}{v_7}$} \makebox[15pt]{$\edge{v_7}{v_8}$} ] ] ]
    \hfill\hfill\hfill
    \Tree [.$B_4$ [ [ \makebox[15pt]{$\edge{v_0}{v_1}$} \makebox[15pt]{$\edge{v_1}{v_2}$} ] 
	       [ \makebox[15pt]{$\edge{v_1}{v_2}$} \makebox[15pt]{$\edge{v_2}{v_3}$} ] ]
           [ [ \makebox[15pt]{$\edge{v_1}{v_2}$} \makebox[15pt]{$\edge{v_2}{v_3}$} ] 
           [ \makebox[15pt]{$\edge{v_2}{v_3}$} \makebox[15pt]{$\edge{v_3}{v_4}$} ] ] ].
     \hfill\hfill\mbox{}
\end{center}
\end{df}

Intuitively, the pattern $A_k$ corresponds to the recursive doubling algorithm for $\PATH(k,n)$. Note that we have essentially already encountered this pattern in the proof of Proposition \ref{prop:upper-bound} (our upper bound for $\cost_{P_k}$). In fact, this proof shows:
\begin{cor}\label{cor:upper-bound-A}
For all $\A \in \P_{A_k}$, 
$\pcost_{A_k}(\A) \le O(kn^{(1/2)\lceil\log k\rceil + 2})$.\qed
\end{cor}

The pattern $B_k$ has a different nature than $A_k$. Whereas sub-patterns $A_j^{\vphantom{\shift j}}$ and $A_{k-j}^{\shift j}$ of $A_k$ overlap at only a single vertex $v_j$, sub-patterns $B_{k-1}$ and $B_{k-1}^{\shift 1}$ of $B_k$ overlap to the maximum possible extent. Despite this difference, it turns out that there is also a reasonable upper bound for $\pcost_{B_k}$.

\begin{prop}\label{prop:upper-bound-B}
For all $\B \in \P_{B_k}$, 
$
  \pcost_{B_k}(\B) \le 2^k n^{\ln(k+1) + 1}.
$
\end{prop}

\begin{proof}
We present a similar argument to the proof of Proposition \ref{prop:upper-bound}. For all $k \ge 1$, define $\B_k \in \P_{B_k}$ by
\[
  \B_k \defeq \{x \in [n]^{V_k} : x_0,\dots,x_k \le n^{1-1/(k+1)}\}.
\]
We have $\B_{k-1} \bowtie \B_{k-1}^{\shift 1}
  = \{x \in [n]^{V_k} : x_0,\dots,x_k \le n^{1-1/k}\}$.
For all $1 \le t_0,\dots,t_k \le n^{1/k(k+1)}$, let
\[
  \COPY_{t_0,\dots,t_k}(\B_{k-1} \bowtie \B_{k-1}^{\shift 1}) 
  \defeq
  \big\{x \in [n]^{V_k} : 
  t_i-1 < \frac{x_i}{n^{1-1/k}} \le t_i
  \text{ for all } 0 \le i \le k\big\}.
\] 
Note that
\[
  \B_k = \bigcup_{1 \le t_0,\dots,t_k \le n^{1/k(k+1)}} \COPY_{t_0,\dots,t_k}(\B_{k-1} \bowtie \B_{k-1}^{\shift 1}).
\]
Using (sub-additivity) and (join rule), together with the invariance of $\pcost$ under coordinate-wise permutations of $[n]$ and under shifts, we have
\[
  \pcost_{B_k}(\B_k) \le 2n^{1/k}\pcost_{B_{k-1}}(\B_{k-1}).
\]
This recurrence, together with the base case $\pcost_{B_1}(\B_1) = n$, implies
\[
  \pcost_{B_k}(\B_k) \le 2^{k-1} n^{1+(1/2)+\dots+(1/k)} \le 2^k n^{\ln(k+1)}.
\]
Noting that $[n]^{V_k}$ is covered by $n$ copies of $\B_k$, we have $\pcost_{B_k}([n]^{V_k}) \le 2^k n^{\ln(k+1)+1}$. The proposition then follows using (monotonicity).
\end{proof}

In Appendix \ref{sec:easier} we prove matching lower bounds for $\pcost_{A_k}$ and $\pcost_{B_k}$. In fact, these lower bounds apply to two classes of patterns which include $A_k$ and $B_k$. While the upper bounds for $\pcost_{A_k}$ and $\pcost_{B_k}$ are quite similar, our lower bound arguments are significantly different. The arguments in these two special cases---a ``one-sided'' induction for $\pcost_{A_k}$ and a ``balanced'' induction using the AM-GM inequality for $\pcost_{B_k}$---show up in the two cases ($\dag$) and ($\ddag$) of our general lower bound (Theorem \ref{thm:pattern-lb}). For this reason, the reader might find the results in Appendix \ref{sec:easier} to be a helpful warm-up.

\begin{rmk}
The pathsets $\A_k$ and $\B_k$ which show up in the proofs of our upper bounds are of a particularly simple form: they are {\em rectangular} subsets of $[n]^{V_k}$. In Appendix \ref{sec:rectangular} we discuss a notion of {\em rectangular pathset complexity} $\pcost^{\mr{rect}}$. Proving lower bounds for $\pcost^{\mr{rect}}$ turns our to be much easier than for $\pcost$. We present an example (the  ``palindrome pathset'') which illustrates the difficulty in attempting to generalize this easier lower bound to the non-rectangular setting.
\end{rmk}

\section{Lower Bound for $\pcost$: Special Cases}\label{sec:easier}

We prove easier special cases of our lower bound for $\pcost$ with respect to two classes of patterns which include the key examples $A_k$ and $B_k$ introduced in \S\ref{sec:key-examples}. Although the results of this appendix are not used in the main body of the paper, the arguments in the proof show up in the two cases ($\dag$) and ($\ddag$) of our general lower bound.

\begin{df}\label{df:psi}\
\begin{enumerate}[(i)]
\item
For a pattern $A$,
\begin{enumerate}[---]
  \item
    let $\End{A} \subseteq V_A$ denote the set of {\em endpoints} in $G_A$ (i.e.\ vertices of in-degree or out-degree zero), and let $\Int{A} \defeq V_A \setminus \End{A}$ denote the set of {\em interior vertices} in $G_A$,
  \item
    let $\Con{A}$ denote the set of {\em intervals} in $G_A$ (i.e.\ subsets of $V_A$ which are connected in $G_A$).
\end{enumerate}
Note that $\l{A} = \ds\max_{I \in \Con{A}} |A|-1$ and $\c{A} = |\End{A}|\,/\,2$.
\item
The classes of {\em end-joining} and {\em fully connected} patterns are defined as follows:
\begin{enumerate}[---]
\item
$A$ is {\em end-joining} if no edge of $P_k$ labels more than one leaf of $A$ (equivalently, $E_{A_1} \cap E_{A_2} = \emptyset$ for all non-atomic sub-patterns $\{A_1,A_2\} \preceq A$),
\item
$A$ is {\em fully connected} if $G_{A'}$ is connected (i.e.\ $\c{A'} = 1$) for all sub-patterns $A' \preceq A$.
\end{enumerate}
Note that patterns $A_k$ and $B_k$ are both fully connected, while only $A_k$ is end-joining (for $k \ge 3$).
\item
Functions $\psino_A,\psifc_A : \P_A \to \R$ are defined as follows:
\begin{enumerate}[---]
  \item
    for end-joining patterns $A$,
    \[
      \psino_A(\A) \defeq 
      \m\vphantom{\Big|}^{\frac12\big(\log(\l{A}) + \c{A}\big)}
      \sqrt{\ts\Ex_{z \in [n]^{\End{A}}}\big[\, \delta(\RHO{\Int{A}}{\A}{z})^2 \,\big]},
    \]
  \item
    for fully connected patterns $A$,
    \[
      \psifc_A(\A) \defeq \max_{I^{\vphantom A} \in \Con{A}}
      \m\vphantom{\Big|}^{\frac14\big(\log(|I|+1) + |I \cap \End{A}|\big)}
      \cdot  \pi_I(\A).
    \]
\end{enumerate}
For non-end-joining patterns $A$, we set $\psino_A(\A) \defeq 0$, and for non-fully connected patterns $A$, we set $\psifc_A(\A) \defeq 0$.
\end{enumerate}
\end{df}

\begin{prop}\label{prop:special-lbs}
Both $\psino$ and $\psifc$ are lower bounds on pathset complexity. That is, for every pattern $A$ and pathset $\A \in \P_A$, we have $\psino_A(\A) \le \pcost_A(\A)$ and $\psifc_A(\A) \le \pcost_A(\A)$. In particular,
\begin{align*}
  &\pcost_{A_k}([n]^{V_k}) \ge \psino_{A_k}([n]^{V_k}) \ge \m^{\frac12(\log(k) + 1)} \ge n^{\frac12\log k},\vphantom{\Big|}\\
  &\pcost_{B_k}([n]^{V_k}) \ge \psifc_{B_k}([n]^{V_k}) \ge \m^{\frac14(\log(k+1) + 2)} \ge n^{\frac14\log k}.\vphantom{\Big|}
\end{align*}
\end{prop}

Recall from Remark \ref{rmk:dual2} the dual characterization of $\pcost$ as the unique coordinate-wise maximal function from pairs $(A,\A)$ to $\R$ which satisfies inequalities (base case), (monotone), (sub-additive) and (join rule). It is easy to see that $\psino$ and $\psifc$ satisfy the first three of these inequalities. To prove Proposition \ref{prop:special-lbs}, it suffices to show that $\psino$ and $\psifc$ also satisfy inequality (join rule). We begin with $\psino$.

\begin{la}\label{la:psi2}
For every non-atomic end-joining pattern $C = \AB$ and small pathsets $\A \in \smallP_A$ and $\B \in \smallP_B$,
\[
  \psino_C(\A \bowtie \B) \le \max\{\psino_A(\A),\psino_B(\B)\}.
\]
\end{la}

\begin{proof}
Without loss of generality, assume that $\l{A} \ge \l{B}$. After making three observations, will show that $\psino_C(\A \bowtie \B) \le \psino_A(\A)$.

First, note that each connected component of $G_C$ ($=G_A \cup G_B$) is the union of at most $\c{A} + \c{B} - \c{C} + 1$ components of $G_A$ and $G_B$. It follows that $\l{C} \le (\c{A} + \c{B} - \c{C} + 1)\cdot \l{A}$.

Since $C$ is end-joining, $\End{C}$ is the symmetric difference of $\End{A}$ and $\End{B}$. By the Cauchy-Schwartz inequality,
\begin{align*}
 \Ex_{c \in [n]^{\End{C}}} 
     \big[\, 
     &\delta(\pRHO{\Int{C}}{\A \bowtie \B}{c})^2 
     \,\big]\\
  &=
    \Ex_{\substack{x \in [n]^{\End{A} \setminus \End{B}}\\
    y \in [n]^{\End{B} \setminus \End{A}}}}
     \Big[\,
     \Big(
     \Ex_{z \in [n]^{\End{A} \cap \End{B}}} 
     \big[\,\delta(\RHO{\Int{A}}{\A}{xz})\cdot 
        \delta(\RHO{\Int{B}}{\B}{yz})
     \,\big]
     \Big)^2
     \,\Big]\\
  &\le 
     \Ex_{\substack{x \in [n]^{\End{A} \setminus \End{B}}\\
            z \in [n]^{\End{A} \cap \End{B}}}}
     \big[\,\delta(\RHO{\Int{A}}{\A}{xz})^2
     \,\big]
     \Ex_{\substack{y \in [n]^{\End{B} \setminus \End{A}}\\
          z \in [n]^{\End{A} \cap \End{B}}}}
     \big[\,\delta(\RHO{\Int{B}}{\B}{yz})^2
     \,\big]\\
  &= 
    \Ex_{a \in [n]^{\End{A}}} \big[\,\delta(\RHO{\Int{A}}{\A}{a})^2\,\big]
     \Ex_{b \in [n]^{\End{B}}} \big[\,\delta(\RHO{\Int{B}}{\B}{b})^2\,\big].
\end{align*}
We next note that $B$-smallness of $\B$ implies
\[
  \Ex_{b \in [n]^{\End{B}}} 
  \big[\,\delta(\RHO{\Int{B}}{\B}{b})^2\,\big]
  \le
  \Ex_{b \in [n]^{\End{B}}} 
  \big[\,\delta(\RHO{\Int{B}}{\B}{b})\,\big]
  =
  \delta(\B)
  \le
  \m^{-\c{B}}.
\]
Putting these inequalities together, we have
\begin{align*}
  \psino_C(\A \bowtie \B)
  &= \m\vphantom{\Big|}^{\frac12\big(\log(\l{C}) + \c{C}\big)}
     \sqrt{\ts
     \Ex_{c \in [n]^{\End{C}}} \big[\, 
     \delta(\RHO{\Int{C}}{\A \bowtie \B}{c})^2 \,\big]}
     \\
  &\le \m\vphantom{\Big|}^{\frac12\big(\log(\l{A}) + \log(\c{A} + \c{B} - \c{C} + 1) + \c{C} - \c{B}\big)}
     \sqrt{\ts
     \Ex_{a \in [n]^{\End{A}}} \big[\,\delta(\RHO{\Int{A}}{\A}{a})^2\,\big]}\\
  &= \m\vphantom{\Big|}^{\frac12\big(\log(\c{A} + \c{B} - \c{C} + 1) + \c{C} - \c{B} - \c{A}\big)}
     \cdot  
     \psino_A(\A)\\
  &\le \psino_A(\A)
\end{align*}
using the fact that $\log(s+1) \le s$ for every integer $s \ge 0$. We get $\psino_C(\A \bowtie \B) \le \max\{\psino_A(\A),\psino_B(\B)\}$ as required.\end{proof}

We next show that $\psifc$ satisfies inequality (join rule).

\begin{la}\label{la:psi2}
For every non-atomic fully connected pattern $C = \AB$ and small pathsets $\A \in \smallP_A$ and $\B \in \smallP_B$,
\[
  \psifc_C(\A \bowtie \B) \le \frac{\psifc_A(\A) + \psifc_B(\B)}{2}.
\]
\end{la}

\begin{proof}
Fix $I \in \Con{A}$ such that 
\[
  \psifc_C(\A \bowtie \B) = \m\vphantom{\Big|}^{\frac14\big(\log(|I|+1) + |I \cap \End{C}|\big)} \pi_I(\A \bowtie \B).
\]
We consider various cases depending on $|I \cap \End{C}| \in \{0,1,2\}$. The most important case is where $|I \cap \End{C}| = 2$ (i.e.\ $I$ contains both endpoints of $G_C$). Because $G_C$ is connected, this means that $I = V_C$ ($= V_A \cup V_B$) and hence $\pi_I(\A \bowtie \B) = \delta(\A \bowtie \B)$.

Within this case, the most important sub-case is where $|E_A|,|E_B| \ge \frac12|E_C|$. In this sub-case, we argue as follows. Without loss of generality, $V_C = \{v_0,\dots,v_k\}$ (i.e.\ $G_C$ is the path $P_k$) and $v_0 \in V_A$ and $v_k \in V_B$. Let $j = \lfloor\frac{k-1}{2}\rfloor$ and $J = \{v_0,\dots,v_j\}$ and $K = \{v_{k-j},\dots,v_k\}$ and note that $J \in \Con{A}$ and $K \in \Con{B}$. Since $v_0 \in J \cap \End{A}$ and $v_k \in K \cap \End{B}$, we have
\begin{equation}\label{eq:blep}
  \log(k+2) \le \log(|J|+1) + |J \cap \End{A}|
  \quad\text{and}\quad
  \log(k+2) \le \log(|K|+1) + |K \cap \End{B}|.
\end{equation}

Next, observe that $\delta(\A \bowtie \B) \le \m^{-1}$ by $C$-smallness of $\A \bowtie \B$. We also have the bound $\delta(\A \bowtie \B) \le \pi_J(\A) \cdot  \pi_K(\B)$ (since $J \cap K = \emptyset$). Taking the geometric mean of these two inequalities, we have
\[
  \delta(\A \bowtie \B) \le \m^{-1/2}  \sqrt{\pi_J(\A) \cdot \pi_K(\B)}.
\]
Putting these pieces together, we have
\begin{align*}
  \psifc_C(\A \bowtie \B) 
  &= \m\vphantom{\Big|}^{\frac14\big(\log(k+2)+2\big)} 
  \delta(\A \bowtie \B)\vphantom{\Big|}\\
  &\le \vphantom{\frac12}\m\vphantom{\Big|}^{\frac14\log(k+2)}
  \sqrt{\pi_J(\A) \cdot \pi_K(\B)}\\
  &\le \frac12\Big(\m\vphantom{\Big|}^{\frac14\log(k+2)}\pi_J(\A)
  + \m\vphantom{\Big|}^{\frac14\log(k+2)}\pi_K(\B)\Big)
  &&\text{(AM-GM ineq.)}\\
  &\le \frac12\Big(\m\vphantom{\Big|}^{\frac14\big(\log(|J|+1) + 
  |J \cap \End{A}|\big)}\pi_J(\A)
  + \m\vphantom{\Big|}^{\frac14\big(\log(|K|+1) + |K \cap \End{B}|\big)}
  \pi_K(\B)\Big)
  &&\text{(by (\ref{eq:blep}))}\\
  &\le \frac12\Big(\psifc_A(\A) + \psifc_B(\B)\Big)
  &&\text{(ind.\ hyp.)}
\end{align*}

In all other cases (i.e.\ when $|I \cap \End{C}| < 2$ or $\min\{|E_A|,|E_B|\} < |E_C|/2$), the inequality is proved by finding $J \in \Con{A}$ or $K \in \Con{B}$ such that $|I \cap \End{C}| < |J \cap \End{A}|$ or $|K \in \End{B}|$. We omit the analysis of these cases, since the arguments are not relevant to our main pathset complexity lower bound.
\end{proof}

Having shown that $\psino$ and $\psifc$ both satisfying (join rule), the proof of Proposition \ref{prop:special-lbs} is complete. Combining our upper and lower bounds for $\pcost_{A_k}$ and $\pcost_{B_k}$ (Corollary \ref{cor:upper-bound-A} and Propositions \ref{prop:upper-bound-B} and \ref{prop:special-lbs}), we have

\begin{cor}
With respect to patterns $A_k$ and $B_k$, the pathset complexity of the complete $P_k$-pathset $[n]^{V_k}$ has the following bounds:
\begin{align*}
  n^{\frac12\log k - O(1)} &\le \pcost_{A_k}([n]^{V_k}) \le k n^{\frac12\log k + O(1)},\vphantom{\big|}\\  
  n^{\frac14\log k - O(1)} &\le \pcost_{B_k}([n]^{V_k}) \le 2^k n^{\ln k + O(1)}.\vphantom{\big|}
\end{align*}
\end{cor}

Since $\PHI{B_k} = \frac12\log k - O(1)$ (as noted in \S\ref{sec:PHI}), Theorem \ref{thm:pattern-lb} gives the stronger lower bound $\pcost_{B_k}([n]^{V_k}) \ge \m^{(1/2)\log k - O(1)} = n^{(1/2)\log k - O(1)}$. Even after extensively studying this special case, we were unable to narrow the gap between $\frac12\log k$ and $\ln k$ ($\approx 0.69\log k$) in the exponent of $n$ in $\pcost_{B_k}([n]^{V_k})$.

\section{Rectangular Pathsets}\label{sec:rectangular}

A set $X \subseteq [n]^V$ is {\em rectangular} if there exist sets $S_i \subseteq [n]$, $i \in V$, such that $X = \{x \in [n]^V : x_i \in S_i$ for all $i \in V\}$. For a pattern graph $G$, let $\scr R_G = \{\A \in \P_G : \A$ is rectangular$\}$ and $\smallR_G = \scr R_G \cap \smallP_G$. For $\A \in \scr R_G$, we define {\em rectangular pathset complexity} $\pcost_G^{\mr{rect}}(\A)$ exactly like pathset complexity $\pcost_G(\A)$ (Definition \ref{df:pc1}) except with $\scr R_G$ and $\smallR_G$ replacing $\P_G$ and $\smallP_G$. Analogously, we define $\pcost_A^{\mr{rect}}(\A)$ for patterns $A$. Note that $\pcost_A(\A) \le \pcost_A^{\mr{rect}}(\A)$ for all $\A \in \scr R_A$.

\begin{rmk}
I venture to guess that $\pcost_A(\A) = \pcost_A^{\mr{rect}}(\A)$ for all $\A \in \scr R_A$, but have no idea how to prove this.
\end{rmk}

We have remarked that our upper bounds on $\pcost_{A_k}$ and $\pcost_{B_k}$ (Corollary \ref{cor:upper-bound-A} and Proposition \ref{prop:upper-bound-B}) involved only rectangular pathsets. It follows that the same upper bounds apply to $\pcost_{A_k}^{\mr{rect}}$ and $\pcost_{B_k}^{\mr{rect}}$.

As for lower bounds on $\pcost^{\mr{rect}}$, this turns out to be significantly easier than our lower bound for $\pcost$. Similar to our lower bound for fully connected patterns in Appendix \ref{sec:easier}, we can lower bound $\pcost_G^{\mr{rect}}(\A)$ for all $\A \in \scr R_G$ in terms of the projection densities $\pi_S(\A)$ where $S \in \Con{G}$ via a function similar to $\psino$. 

A key difference when it comes rectangular pathsets is that $\pi_S = \mu_S$ (projection density $=$ maximum restriction density) and hence smallness of rectangular pathsets is preserved under projections to a union of components (cp.\ Remark \ref{rmk:proj}). This fact turns out to greatly simplify the task of proving a lower bound for $\pcost^{\mr{rect}}$.

The next example shows that projections of non-rectangular pathsets can be tricky. This illustrates the difficulty in generalizing the lower bound for $\pcost^{\mr{rect}}$ to the non-rectangular setting.

\begin{ex}
For $k \ge 1$, let $\Pal_{2k} \in \P_{P_{2k}}$ be the ``palindrome pathset''
\[
  \Pal_{2k} = \big\{x \in [n]^{0,\dots,2k} : x_{k-i} = x_{k+i} \text{ for all } 0 \le i \le k\big\}.
\]
\end{ex}

The palindrome pathset $\Pal_{2k}$ has low density, while having the maximum projection over vertices $v_0,\dots,v_k$:
\[
  \delta(\Pal_{2k}) = n^{-k}
  \quad\text{ and }\quad
  \pi_{\{v_0,\dots,v_k\}}(\Pal_{2k}) = 1.
\]
It turns out that $\Pal_{2k}$ is inexpensive to construct, given the right pattern. Let $\PalPat_{2k}$ be the pattern 
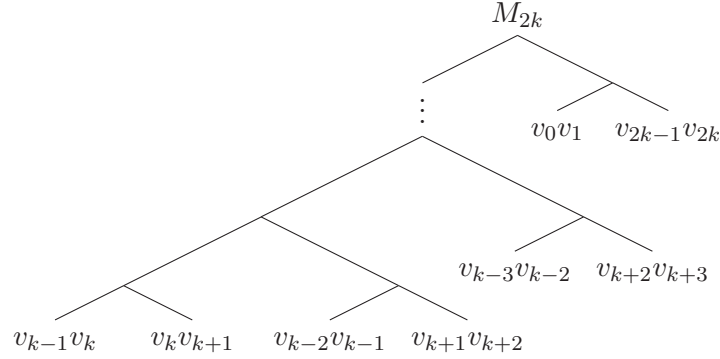
\[
  \text{\Tree 
  [.$\PalPat_{2k}$
    [.$\vdots$  
        [
          [ \makebox[40pt]{$\edge{v_{k-1}}{v_k}$} \makebox[40pt]{$\edge{v_k}{v_{k+1}}$} ]
          [ \makebox[40pt]{$\edge{v_{k-2}}{v_{k-1}}$} \makebox[40pt]{$\edge{v_{k+1}}{v_{k+2}}$} ]
        ]!\qsetw{140pt}
        [
          \makebox[40pt]{$\edge{v_{k-3}}{v_{k-2}}$} \makebox[40pt]{$\edge{v_{k+2}}{v_{k+3}}$}
        ]
    ]!\qsetw{60pt}
    [
      \makebox[30pt]{$\edge{v_0}{v_1}$} \makebox[30pt]{$\edge{v_{2k-1}}{v_{2k}}$}
    ]
  ]}
\]
It is easy to show that
\[
  \pcost_{\PalPat_{2k}}(\Pal_{2k}) \le O(kn^2).
\]
On the other hand, for any fully connected pattern $C$ with graph $P_{2k}$ (such as $A_{2k}$ or $B_{2k}$), the lower bound of Appendix \ref{sec:easier} implies
\[
  \xi_C(\Pal_{2k}) 
    \ge
    \psifc_C(\Pal_{2k}) 
    \ge 
      \m\vphantom{\Big|}^{\frac14\big(\log(|\{v_0,\dots,v_k\}|+1) + |\{v_0,\dots,v_k\} \cap \End{A}|\big)}
      \cdot
      \pi_{\{v_0,\dots,v_k\}}(\A)
    = n^{\Omega(\log k)}.
\]

\end{document}